\documentclass[sigconf]{acmart}

\AtBeginDocument{%
  \providecommand\BibTeX{{%
    \normalfont B\kern-0.5em{\scshape i\kern-0.25em b}\kern-0.8em\TeX}}}

\usepackage{files/packages}
\usepackage{files/new_commands}

\copyrightyear{2020}
\acmYear{2020}
\setcopyright{acmlicensed}\acmConference[CCS '20]{Proceedings of the 2020 ACM SIGSAC Conference on Computer and Communications Security}{November 9--13, 2020}{Virtual Event, USA}
\acmBooktitle{Proceedings of the 2020 ACM SIGSAC Conference on Computer and Communications Security (CCS '20), November 9--13, 2020, Virtual Event, USA}
\acmPrice{15.00}
\acmDOI{10.1145/3372297.3423363}
\acmISBN{978-1-4503-7089-9/20/11}
\setcopyright{none}
\renewcommand\footnotetextcopyrightpermission[1]{}
\begin{document}
 \fancyhead{}
\title{Estimating g-Leakage via Machine Learning} 

\author{Marco Romanelli}
\email{marco.romanelli@inria.fr}
\affiliation{%
  \institution{Inria, \'Ecole Polytechnique, IPP, Universit\`a di Siena}
  \streetaddress{1 Rue Honor\'e d'Estienne d'Orves}
  \city{Palaiseau}
  \state{France}
  \postcode{91120}
}
\author{Konstantinos Chatzikokolakis}
\email{kostas@chatzi.org}
\affiliation{%
  \institution{University of Athens}
 \city{Athens}
  \state{Greece}
  \postcode{10679}
}
\author{Catuscia Palamidessi}
\email{catuscia@lix.polytechnique.fr}
\affiliation{%
  \institution{Inria, \'Ecole Polytechnique, IPP}
   \streetaddress{1 Rue Honor\'e d'Estienne d'Orves}
  \city{Palaiseau}
  \state{France}
  \postcode{91120}
}
\author{Pablo Piantanida}
\email{pablo.piantanida@centralesupelec.fr }
\affiliation{%
  \institution{CentraleSupelec, CNRS, Universit\'e Paris Saclay}
  \streetaddress{3 Rue Joliot Curie}
  \city{Gif-sur-Yvette}
  \state{France}
 \postcode{91190}
}

\renewcommand{\shortauthors}{M. Romanelli, et al.}


\begin{abstract}
\rev{This paper considers the problem of estimating the information leakage of a system in the black-box scenario, i.e. when the system's internals are unknown to the learner, or too complicated to analyze, and the only available information are pairs of input-output data samples, obtained by submitting queries to the system or provided by a third party. 
The \emph{frequentist} approach relies on counting 
the frequencies to estimate the input-output conditional probabilities, 
however this method is not accurate when the domain of possible outputs is large. To overcome this difficulty, 
the estimation of the Bayes error of the ideal classifier was recently investigated using Machine Learning (ML) models, and it has been shown to be more accurate thanks to the ability of those models to  learn the input-output correspondence. However, the Bayes vulnerability is only suitable to describe \emph{one-try} attacks.  A more general and flexible measure of leakage is  the \emph{\gv}, which encompasses several 
different types of adversaries, with different goals and capabilities. 
We propose a novel approach to perform black-box  estimation of the \gv\ using ML which does not require to estimate the conditional probabilities and is suitable for a large class of  ML algorithms. 
First, we formally show the learnability for all data distributions. Then, we evaluate the performance via various experiments using k-Nearest Neighbors and Neural Networks. Our approach outperform the frequentist one when the observables domain is large.}
\end{abstract}
%

\ccsdesc[500]{Security and privacy~Formal security models}
\ccsdesc[500]{Security and privacy~Formal methods and theory of security}
\ccsdesc[500]{Security and privacy~Information flow control}
\ccsdesc[500]{Computing methodologies~Neural networks}
\ccsdesc[500]{Computing methodologies~Machine learning}
\keywords{ \gv\ estimation; machine learning; neural networks} 

\maketitle


\section{Introduction}
The information leakage of a system is a fundamental concern of computer security, and  measuring the amount of  sensitive information that an adversary can obtain by   observing the outputs of a given system 
is of the utmost importance to understand whether such leakage can be tolerated or must be considered a major security flaw. 
Much research effort has been dedicated to studying and proposing solutions to this problem, see for instance  \cite{Clark:01:ENTCS,Kopf:07:CCS,Chatzikokolakis:08:IC,Boreale:09:IC,Smith:09:FOSSACS,Alvim:12:CSF,Alvim:14:CSF,Chatzikokolakis:19:CSF}. So far, this area of research, known as quantitative information flow (QIF)\FIX{,}
has mainly focused on the so-called white-box scenario\rev{, i.e. assuming that the channel of the system is known, or can  be computed by analyzing the system's internals. }
This channel consists of the conditional probabilities of the outputs (observables) given the inputs (secrets).  

The white-box assumption, however, is not always realistic: \FIX{sometimes} the system is unknown, or anyway it is too complex, so that an analytic computation becomes hard if not impossible to be performed. Therefore, it is important to consider also black-box approaches where we only assume the availability of a finite set of  input-output pairs generated by the system. \rev{A further specification of a black box model is the ``degree of inaccessibility'', namely whether or not we are allowed to interact with the system to obtain these  pairs, or they are just provided by a third party. Both scenarios are realistic, and in our paper we consider the two cases.}

The estimation of the internal probabilities of a system's channel have been investigated in~\cite{Chothia:13:CSF} and~\cite{Chothia:14:ESORICS} via a frequentist paradigm, i.e., relying on the computation of the frequencies of the outputs  given some inputs. However, this  approach does not scale to applications for which the output space is very large since a prohibitively large number of samples would be necessary to achieve good results and fails on continuous alphabets unless some strong assumption on the distributions are made. In order to overcome this limitation, the authors \FIX{of}~\cite{Cherubin:19:SP} exploited the fact that Machine Learning (ML) \FixCat{algorithms}  provide a  better scalability to black-box measurements. \FixCat{Intuitively, the advantage of the ML approach over the frequentist one is its generalization power: while 
the frequentist method can only draw conclusions based on counts on the available samples,   ML  is able to extrapolate from the samples and provide better prediction (generalization) for the rest of the universe.}

In particular, ~\cite{Cherubin:19:SP} proposed to use k-Nearest Neighbors (k-NN) \FixCat{to measure  the basic QIF metrics, the Bayes vulnerability~\cite{Smith:09:FOSSACS}. 
This is the expected probability of success of an adversary that has exactly one attempt at his disposal (one-try), and tries to maximize the chance of guessing the right  secret.} 
The  Bayes vulnerability corresponds to the converse of the error of the ideal Bayes classifier that, given any \FixCat{sample (observable)}, tries to predict its corresponding \FixCat{class (secret)}. Hence the idea is  to build a model that approximates such classifier, and estimate its expected error.  
The \FixCat{main} takeaway is that any ML rule which is \emph{universally consistent} (i.e., approximates the ideal Bayes classifier) has a guarantee on the accuracy of the estimation, i.e., the  error in the estimation of the leakage tends to vanish as the number of training samples grows large.

The method of~\cite{Cherubin:19:SP}, however, is limited to the Bayes vulnerability, which models only one particular kind of adversary. 
As argued in~\cite{Alvim:12:CSF}\FIX{,} there are many other realistic ones. 
For instance, adversaries  whose aim is to guess only a part of the secret, or a property of the secret,
or    that have multiple tries at their disposal. 
To represent a larger class of attacks, ~\cite{Alvim:12:CSF} introduced the so-called  $g$-\emph{vulnerability}.
This metric is very general, and in particular it encompasses the Bayes vulnerability.

\begin{table}[htb!]
\centering
	\begin{tabular}{|l|l|}
		\hline
		\multicolumn{1}{| c |}{\textbf{Symbol}} & \multicolumn{1}{ c |}{\textbf{Description}}  \\ 
		\hline
		    \multicolumn{1}{| c |}{$x \in \Xcal$} & \multicolumn{1}{ l |}{a secret}  \\ \multicolumn{1}{| c |}{$w \in \Wcal$} & \multicolumn{1}{ l |}{a guess}  \\
		    \multicolumn{1}{| c |}{$y \in \Ycal$} & \multicolumn{1}{ l |}{an observable output by the system}  \\ 
		    \multicolumn{1}{| c |}{$X$} & \multicolumn{1}{ l |}{\rev{random var. for secrets taking values $\mathit{x \in \Xcal}$}}  \\ 
			\multicolumn{1}{| c |}{$W$} & \multicolumn{1}{ l |}{\rev{random var. for guesses taking values $\mathit{w \in \Wcal}$}}  \\ 
			\multicolumn{1}{| c |}{$Y$} & \multicolumn{1}{ l |}{\rev{random var. for observables taking values $\mathit{y \in \Ycal}$}}\\
			\multicolumn{1}{| c |}{$|\Scal|$} & \multicolumn{1}{ l |}{size of a set $\Scal$ }  \\ 
			\multicolumn{1}{| c |}{\FIX{$\mathcal{P}(\mathcal{S})$}} & \multicolumn{1}{ l |}{\FIX{Distribution over a set of symbols $\mathcal{S}$}}  \\
			\multicolumn{1}{| c |}{$\Hcal$} & \multicolumn{1}{ l |}{class of learning functions $f$}  \\ 
			\multicolumn{1}{| c |}{$\pi$, $P_X$} & \multicolumn{1}{ l |}{prior distribution over the secret space}  \\ 
			\multicolumn{1}{| c |}{$\hat\pi$, $\widehat{P}_X$} & \multicolumn{1}{ l |}{empirical prior distribution over the secret space}  \\ 
			\multicolumn{1}{| c |}{$C$} & \multicolumn{1}{  l|}{Channel matrix}  \\
			\multicolumn{1}{| c |}{$\Joint{\pi}{C}$} & \multicolumn{1}{ l |}{joint distribution from prior $\pi$ and channel $C$}\\
			\multicolumn{1}{| c |}{\FIX{${P_{XY}}$}} & \multicolumn{1}{ l |}{joint probability distribution}\\
			\multicolumn{1}{| c |}{${\widehat{P}_{XY}}$} & \multicolumn{1}{ l |}{empirical joint probability distribution}\\
			\multicolumn{1}{| c |}{$P_{Y\mid X}$} & \multicolumn{1}{  l|}{conditional probability of $Y$ given $X$}  \\
			\multicolumn{1}{| c |}{$\widehat{P}_{Y\mid X}$} & \multicolumn{1}{  l|}{empirical conditional probabilities }  \\
			\multicolumn{1}{| c |}{$\mathbb{P}$} & \multicolumn{1}{ l |}{probability measure}\\	
			\multicolumn{1}{| c |}{$\FIX{\mathbb{\E[\cdot]}}$} & \multicolumn{1}{ l |}{\FIX{expected value}}\\										
			\multicolumn{1}{| c |}{$\FIX{g(w,x)}$} & \multicolumn{1}{ l |}{\FIX{\rev{gain function: guess $w$ and secret $x$ as inputs}}}\\
			\multicolumn{1}{| c |}{$G$} & \multicolumn{1}{ l |}{\rev{gain matrix of size $|\Wcal|\times|\Xcal|$ for a specific $g$}}\\
			\multicolumn{1}{| c |}{$V_g$} & \multicolumn{1}{ l |}{\gv}\\
			\multicolumn{1}{| c |}{$V(f)$} & \multicolumn{1}{ l |}{\gv\ functional}\\
			\multicolumn{1}{| c |}{$\widehat{V}_n(f)$} & \multicolumn{1}{ l |}{\rev{empirical $g$-vuln. functional evaluated on $n$ samples}}\\
		\hline
	\end{tabular}
	\vspace{0.2cm}
	\caption{Table of symbols.}
	\label{symbol_table}
\end{table}

In this paper, we propose an approach to the black-box estimation of  $g$-vulnerability via ML. 
The idea is to reduce the problem to that of approximating the Bayes classifier, so that \emph{any universally consistent ML algorithm can be used for the purpose}. 
This reduction essentially takes into account the impact of the gain function in the generation of the training data, and we propose two methods to obtain this effect, 
which we call  \emph{channel pre-processing} and  \emph{data pre-processing}, respectively. 
We evaluate our approach via  experiments on various channels and  gain functions. \FIX{In order to show the generality of our approach, 
we use two different ML algorithms, namely k-NN and Artificial Neural Networks (ANN), and we compare their performances.} The experimental results show that our approach provides accurate  estimations, and that \FIX{it} outperforms by far the frequentist approach when the observables domain is large. 
\subsection{Our contribution}
\begin{itemize}
    \item We propose a novel approach to the black-box estimation of  $g$-vulnerability based on ML. To the best of our knowledge, this is the first time that  a method to estimate $g$-vulnerability in a black-box fashion is introduced.
    \item We provide statistical guarantees showing the learnability of the \gv\ for all distributions and 
we derive  distribution-free bounds on the accuracy of its estimation.
    \item We validate the performance of our method via several experiments using k-NN and ANN models. 
    The code  is available at the URL \url{https://github.com/LEAVESrepo/leaves}.
\end{itemize}
\subsection{Related work}
One important aspect to keep in mind when measuring leakage is the kind of attack that we want to model. In~\cite{Kopf:07:CCS}, K\"opf and Basin identified various kinds of adversaries and showed that they can be captured by known entropy measures. In particular it focussed on the adversaries corresponding to Shannon and Guessing entropy. 
In~\cite{Smith:09:FOSSACS} Smith proposed another notion, the R\'enyi min-entropy,  to measure the system's leakage when the attacker has only one try at its disposal and attempts to make its  best guess. The R\'enyi min-entropy is the logarithm of the Bayes vulnerability, which is the  expected probability of the adversary to guess the secret correctly. 
The Bayes vulnerability is the converse of the Bayes error, which was already proposed as a measure of leakage in \cite{Chatzikokolakis:08:JCS}. 

Alvim et al.~\cite{Alvim:12:CSF} generalized the notion of Bayes vulnerability to that  of \gv, by introducing a parameter (the gain function $g$) that describes the adversary's payoff. 
The \gv\  is  the expected gain of the adversary in a one-try attack. 
  
The idea of estimating   the \gv\ in the  using ML  techniques  is inspired by the seminal work~\cite{Cherubin:19:SP}, which used k-NN algorithms to estimate the Bayes vulnerability,  a paradigm shift with respect to the previous  frequentist approaches~\cite{Chatzikokolakis:10:TACAS, Chothia:13:CSF, Chothia:14:ESORICS}. Notably,~\cite{Cherubin:19:SP} showed that  universally consistent learning rules allow  to achieve much more precise estimations than the frequentist approach when considering large or even continuous output domains which would be intractable otherwise. 
However,  \cite{Cherubin:19:SP}  was limited to the Bayes 
adversary. In contrast, our approach handles any adversary that can be modeled by a
gain function $g$. Another novelty w.r.t. \cite{Cherubin:19:SP} is that we consider also  ANN algorithms, which in various experiments appear to perform better than the k-NN ones. 
 
Bordenabe and Smith~\cite{Bordenabe:16:CSF} investigated the indirect leakage induced by a channel (i.e., leakage on sensitive information not in the domain of the channel),  and proved a fundamental equivalence between Dalenius min-entropy leakage under arbitrary correlations and g-leakage under arbitrary gain functions. This result is similar to our \autoref{thm:channel-pre-processing}, and  it opens the way to the possible extension of our approach to this more general leakage scenario.


\section{Preliminaries}\label{sec:preliminaries}
In this section, we \FIX{recall} some useful notions from QIF and ML.

\subsection{\FIX{Quantitative information flow}}

\label{qif_notions}
Let $\Xcal$ be a set of secrets and \FIX{$\Ycal$ a set of observations}.
The adversary's initial knowledge about the secret\FIX{s} is modeled by a prior
distribution 
$\mathcal{P}(\mathcal{X})$
\FIX{(namely $P_X$)}.
A system is modeled as a probabilistic \emph{channel} from $\Xcal$ to $\Ycal$,
described by a stochastic matrix $C$, whose elements $C_{xy}$ give the
probability to observe $y\in\Ycal$ when the input is $x\in\Xcal$ (namely $P_{Y|X}$).
Running $C$ with input $\pi$ induces a joint distribution on $\Xcal\times\Ycal$ denoted
by $\Joint{\pi}{C}$.

In the $g$-leakage framework \cite{Alvim:12:CSF}
an adversary is described by a set $\Wcal$ of
\emph{guesses} (or \emph{actions}) that \FIX{it} can make about the secret,
and by a \emph{gain function} $g(w,x)$ expressing the gain of selecting the guess
$w$ when the real secret is $x$.
The prior \gv{} is  the \emph{expected gain} of an optimal guess,
given a prior distribution on secrets:
\begin{equation}
	\label{prior_g_v}
	\mathit{V_g (\pi) \defsym \max_{w \in \mathcal{W}} \sum_{x\in \mathcal{X}}\pi_x \cdot g(w,x)}\, .
\end{equation}
In the posterior case, the adversary observes the output of the system which allows
to improve \FIX{its} guess. \FIX{Its} expected gain is given by
the posterior \gv, according to
\begin{equation}
	\label{posterior_g_v}
	\mathit{V_g (\pi, C) \defsym \sum_{y\in \Ycal}\max_{w \in \Wcal} \sum_{x\in \Xcal} \pi_x \cdot C_{xy} \cdot g(w,x)\, .}
\end{equation}
Finally, the multiplicative\footnote{Originally  the multiplicative version of $g$-leakage  was defined as the $\log$ of the definition given here. In recent literature  the $\log$ is not used anymore. Anyway, the two definitions are equivalent for system  comparison, since  $\log$ is  monotonic.} and additive $g$-leakage quantify how much \FIX{a specific channel $C$} \emph{increases} the
vulnerability of the system:
\begin{equation}
	\label{g_leak}
	\mathit{\mathcal{L}^{\rm M}_g (\pi, C) \defsym \frac{V_g (\pi, C)}{V_g (\pi)}~, \quad  \mathcal{L}^{\rm A}_g (\pi, C) \defsym  V_g (\pi, C) - V_g (\pi)}\, .
\end{equation}
%

The choice of the gain function $g$ allows to model a variety of different adversarial
scenarios. The simplest case is the \emph{identity} gain function, given by $\Wcal = \Xcal$,
$\GId(w,x)=1$ iff $x=w$ and $0$ otherwise. This gain function models an adversary
that tries to guess the secret exactly in one try; $V_{\GId}$ is the \emph{Bayes-vulnerability}, which corresponds 
to the complement of the Bayes error (cfr.~\cite{Alvim:12:CSF}).

However, the interest in $g$-vulnerability lies in the fact that many more adversarial
scenarios can be captured by a proper choice of $g$. For instance, taking
$\Wcal = \Xcal^k$ with $g(w,x)=1$ iff $x\in w$ and $0$ otherwise, models an adversary
that tries to guess the secret correctly in $k$ tries. Moreover,
guessing the secret \emph{approximately} can be easily expressed by constructing $g$
from a metric $d$ on $\Xcal$; this is a standard approach in the area of~\emph{location
privacy}~\cite{Shokri:11:SP,Shokri:14:CCS} where $g(w,x)$ is taken
to be inversely proportional to the Euclidean distance between
$w$ and $x$. Several other gain functions are discussed in
\cite{Alvim:12:CSF}, while \cite{Alvim:16:CSF} shows that \emph{any} vulnerability
function satisfying basic axioms can be expressed as $V_g$ for a properly constructed $g$.

The main focus of this paper is estimating the posterior $g$-vulnera\-bility of the system from such samples.
Note that, given $V_g(\pi,C)$, estimating $\mathit{\mathcal{L}^{\rm M}_g(\pi, C)}$ and $\mathit{\mathcal{L}^{\rm A}_g(\pi, C)}$ is straightforward,
since $V_g(\pi)$ only depends on the prior (not on the system) \FIX{and} it can be either computed
analytically or estimated from the samples.
%



\subsection{Artificial Neural Networks}
\FIX{We provide a short review of the aspects of ANN that are relevant for this work. For further details, 
we refer to} \cite{Bishop:07:ISS, Goodfellow:16:ACML, Hastie:01:SV}. Neural networks 
are usually modeled as directed graphs with weights on the connections and nodes that forward information through ``activation functions'', often introducing non-linearity (such as sigmoids or soft-max). In particular,  we consider an instance of learning known as \emph{supervised learning}, where input samples are provided to the network model together with target labels (supervision). From the provided data and by means of iterative updates of the connection weights, the network learns how the data and respective labels are distributed. The training procedure, known as back-propagation, is an optimization process aimed at minimizing a loss function that quantifies the quality of the network's prediction w.r.t. the  data. 

\FIX{Classification problems are a typical example of tasks for which supervised learning works well. Samples are provided together with target labels which represent the classes they belong to. A model can be trained using these samples and, later on, it can be used to predict the class of new samples.}



\rev{According to the No Free Lunch theorem (NFL)~\cite{Wolpert:96:NC}, in general, it cannot be said that ANN are better than other ML methods. However, it is well known that the NFL can be broken by additional information on the data or the particular problem we want to tackle, and, nowadays, for most applications and available data, especially in multidimensional domains, ANN models outperform other methods therefore representing the state-of-the-art.}
\FIX{\subsection{k-Nearest Neighbors}
The k-NN algorithm is one of the simplest algorithms used to classify a new sample given a training set of samples labelled as belonging to specific classes. This algorithm assumes that the space of the features is equipped with a notion of distance. The basic idea is the following: every time we need to classify a new sample, we find the $k$ samples whose features are closest to those of the new one (nearest neighbors). Once the $k$ nearest neighbors are selected, a majority vote over their class labels is performed to decide which class should be assigned to the new sample. For further details, as well as for an extensive analysis of the topic, we refer to the chapters about k-NN in~\cite{Shalev-Shwartz:14:UML,Hastie:01:SV}}.

\section{Learning \gv: Statistical Bounds}
\label{learning_gv}
\rev{This section introduces the mathematical problem of learning \gv, characterizing universal learnability in the present framework. We derive distribution-free bounds on the accuracy of the estimation, implying the estimator's statistical consistence. }

\subsection{Main definitions}
\FIX{We consider the   problem of estimating the \gv\ from samples via ML models, and we show that the analysis of this estimation} can be conducted in the general statistical framework of maximizing an expected functional using observed samples.  \FIX{The idea} can be described using three components:
\begin{itemize}
\item A generator of random secrets $x\in \Xcal$ with $| \Xcal|< \infty$, drawn independently from a fixed but unknown distribution  $P_X(x)$;
\item a channel that returns an observable $y\in \Ycal$  with $| \Ycal|< \infty$ for  every input \FIX{$x$}, according to a conditional distribution $P_{Y|X}(y|x)$, also fixed and unknown; 
\item a learning machine capable of implementing a set of rules $f\in \Hcal$, where $\Hcal$ denotes the class of functions ${f:\Ycal\rightarrow\Wcal}$ and $\Wcal$ is the finite set of guesses. 
\end{itemize}
Moreover let us note that
\begin{equation}
g: \mathcal{W} \times \Xcal  \rightarrow  [a,b]
\end{equation} 
where $a$ and $b$ are finite real values, and $\Xcal$ and $\Wcal$ are finite sets.
The problem of learning the \gv\ is that of choosing the function ${f:\Ycal\rightarrow\Wcal}$ which maximizes  the functional $V(f)$, representing the expected gain, defined as:
\begin{equation}
V(f) \defsym  \sum_{(x,y) \in \Xcal\times \Ycal} g\big( f(y), x \big) P_{XY} (x,y).\label{eq-g-leakege}
\end{equation}
Note that  $f(y)$ corresponds to the ``guess'' $w$, for the given $y$,  in~\eqref{posterior_g_v}. 
The maximum of $V(f)$ is the \gv, namely: 
\begin{equation}
V_g \defsym  \max\limits_{f \in\Hcal} V(f).\label{eq-max}
\end{equation}

\subsection{Principle of the empirical \gv\ maximization}\label{empirical_max}
Since we are in the black box scenario,
the joint probability distribution  $P_{XY} \equiv \Joint{\pi}{C}$ 
is unknown. We assume, however, the availability 
of $m$   independent and identically distributed (i.i.d.) \FIX{samples} drawn according to \FIX{$P_{XY}$} that can be used as a training set \trainingsetsub\   to solve the maximization of $f$ over $\mathcal{H}$ and additionally $n$ i.i.d. samples are available to be used as a validation\footnote{We call \validationsetsub\ validation set rather than test set,  since we use it to estimate the \gv\ with the learned $f^\star_m$, rather than to measure  the quality of  $f^\star_m$.} set \validationsetsub\ to estimate the average in \eqref{eq-g-leakege}.  Let us denote these sets as:
\trainingsetsub\ 
$
 \defsym \big\{ (x_1,y_1),\dots, (x_m,y_m) \big\} 
$ 
and 
\validationsetsub\ 
$
 \defsym \big\{ (x_{m+1},y_{m+1}),\dots, (x_{m+n},y_{m+n}) \big\}
$, respectively. 

In order to maximize the \gv\ functional \eqref{eq-g-leakege}  for an unknown probability measure   $ P_{XY}$,  the following principle is usually applied. 
The expected \gv\ functional  \FIX{$V(f)$}   is replaced by the \emph{empirical \gv\ functional}:
\begin{equation}
\widehat{V}_n (f) \defsym  \frac{1}{n}\sum_{(x,y) \in \mathcal{T}_n} g\big( f(y), x \big), \label{eq-empirical-g-leakege}
\end{equation}
which is evaluated on \validationsetsub\ rather than $P_{XY}$. This estimator is clearly unbiased in the sense that 
$
\E\big[ \widehat{V}_n (f)\big] = V(f). 
$


Let $f^\star_m$ denote the empirical optimal rule given by 
\FIX{
\begin{equation}\label{eq-solving-emr}
f^\star_m\defsym  \argmax\limits_{f\in\Hcal} \widehat{V}_m (f),\,\, 
\widehat{V}_m (f) \defsym \frac1m \sum_{(x,y) \in \mathcal{D}_m} g\big( f(y), x \big),
\end{equation}
which is evaluated on \trainingsetsub\ rather than $P_{XY}$.  The function $f^\star_m$ is the optimizer according to \trainingsetsub, namely the best way among the functions $f:\Ycal\rightarrow\Wcal$  to approximate $V_g$ by  maximizing  $\widehat{V}_m (f)$ over the class of functions $\mathcal{H}$. This principle is known in statistics  as the Empirical Risk Maximization (ERM).}

Intuitively, we would like $f^\star_m$  to give a good approximation of the \gv,  in the sense that its expected gain
\begin{equation}
V(f^\star_m) \; =  \sum_{(x,y) \in \Xcal\times \Ycal} g\big( f^\star_m(y), x \big) P_{XY} (x,y) 
\end{equation}
should be close to $V_g$. Note that the difference
\begin{equation}\label{eq-missing}
V_g - V(f^\star_m)=\max\limits_{f\in\Hcal} V(f)  - V(f^\star_m)   
\end{equation}
is always non negative and represents the gap by selecting a possible suboptimal function $f^\star_m$. Unfortunately, we are not able to compute $V(f^\star_m)$ either, since $P_{XY}$ is unknown. \rev{In its place, we have to use its \mrev{estimation} $\widehat{V}_n(f^\star_m)$ 
Hence, the real \mrev{estimation} error is $|V_g - \widehat{V}_n(f^\star_m)|$. Note that:
\begin{equation}\label{eq:new}
| V_g - \widehat{V}_n(f^\star_m) | \leq (V_g - V(f^\star_m) ) + | \widehat{V}_n(f^\star_m) - V(f^\star_m)|,
\end{equation}
\mrev{where the first term in the right end side represents the error induced by using the trained model $f_{m}^\star$ and the second represents the error induced by estimating the \gv\ over the $n$ samples in the validation set.}
}

 By using basics principles from statistical learning theory, we study two main questions: 
\begin{itemize}
\item When does the estimator $\widehat{V}_n (f^\star_m)$ work? What are the conditions for its statistical consistency? 
\item How well does $\widehat{V}_n (f^\star_m)$ approximate $V_g$? In other words, how fast does the sequence of largest empirical g-leakage values  converge to the largest g-leakage function? This is related to the so called rate of generalization of a learning machine that implements the ERM principle. 
\end{itemize}

\subsection{Bounds on the estimation accuracy}

\rev{In the following we  use  the  notation $\sigma^2_f = \mathit{Var}(g(f(Y),X))$, where $\mathit{Var}(Z)$ stands for the variance of the random variable $Z$. 
Namely,  $\sigma^2_f$ is   the variance of the gain for a given function $f$. 
We also use $\mathbb{P}$ to represent  the probability induced by sampling the training and validation sets over the distribution $P_{XY}$.}

\rev{Next proposition,  proved in Appendix \ref{AppC}, states that we can probabilistically delimit the two parts of the bound in~\eqref{eq:new}  in terms of   the sizes $m$ and $n$ of the training and validation sets.}

\begin{restatable}[Uniform deviations]{proposition}{mainpropositionrestatable}\label{mainproposition}
\rev{For all $\varepsilon>0$, 
\begin{align}
&\mathbb{P}\left(  \big |  \widehat{V}_n(f^\star_m) - V(f^\star_m)\big |  \geq \varepsilon  \right) \leq2 \E    \, \exp  \left(-\frac{n\,\varepsilon^2}{2\,\sigma^2_{f^\star_m}+\nicefrac{2\,\left(b-a\right)\varepsilon}{3}}\right),\label{eq-Pro1}
\end{align}
where the expectation is taken w.r.t. the random training set, and
\begin{align}
&\mathbb{P}\left(   V_g - {V}(f^\star_m)  \geq \varepsilon  \right) \leq 2\sum_{f\in\mathcal{H}}\exp\left(-\frac{m\,\varepsilon^2}{8\sigma_f^2+\nicefrac{4\left(b-a\right) \varepsilon}{3}}\right).
\label{eq-Pro2}
\end{align}
}
\end{restatable}
Inequality \eqref{eq-Pro1} shows that the estimation error due to the use of a validation set  in $\widehat{V}_n(f^\star_m) $ instead of the true expected gain ${V}(f^\star_m)$  vanishes with the number of validation samples. On the other hand, expression \eqref{eq-Pro2}  implies `learnability' of an optimal $f$, i.e., the suboptimality of $f^\star_m$ learned using the training set $\widehat{V}_m(f^\star_m) $ 
vanishes with the number of training samples. 

\rev{On the other hand the bound in~\cref{eq-Pro2}  depends on the underlying distribution and the structural properties of the class $\mathcal{H}$ through the variance. However, it does not explicitly depend on the optimization algorithm (e.g., stochastic gradient descent) used to learn the function $f_m^\star$ from the training set, which just comes from assuming the worst-case upper bound over all optimization procedures. The selection of a ``good” subset of candidate functions, having a high probability of containing an almost optimal classifier, is a very active area of research in ML~\cite{Chizat:20:COLT}, and hopefully there will be some result available soon, which together with our results will provide a practical method to estimate the bound on the errors. In~\cref{mod_sel} we
discuss heuristics to choose a ``good'' model, together with a new set of experiments showing the impact of different architectures.}

\begin{restatable}[]{theorem}{thmestimateconvergence}\label{thm_estimate_convergence}
The averaged estimation error of the \gv\ can be bounded as follows:
\begin{equation}
\E \big|V_g-  \widehat{V}_n(f^\star_m)  \big| \leq     V_g-  \E \big[ {V}(f^\star_m)  \big]   + \E \big|{V}(f^\star_m)-  \widehat{V}_n(f^\star_m)  \big| , 
\end{equation}
where the expectations are understood over all possible training and validation sets drawn according to $P_{XY}$. Furthermore, \linebreak
\mrev{let $\sigma^2=max_{f\in \Hcal}\text{Var}\big(g(f(Y),X)\big)$, then
:
\begin{align}
	\E \big|{V}(f^\star_m)-  \widehat{V}_n(f^\star_m)  \big|& \leq  \frac{4\eta}{n}\exp\left(-\frac{n\sigma^2}{2\eta}\right)
\nonumber\\ 
&+\sqrt{\frac{2\sigma^2\eta\pi}{n}}\erf\left(\nicefrac{\sigma^2}{\sqrt{\displaystyle\frac{2\sigma^2\eta\pi}{n}}}\right)
,\label{frst_bound}
\end{align}
where $\eta=(1+\nicefrac{(b-a)}{3})$ for $\sigma^2\leq\varepsilon$, and, otherwise,
\begin{align}
	&V_g-  \E \big[ {V}(f^\star_m)  \big] \leq  |\mathcal{H}|\frac{8(1+\eta)}{m}\exp\left(-\frac{m\sigma^2}{4(1+\eta)}\right)
+\nonumber\\
	&|\Hcal|\sqrt\frac{4\sigma^2(1+\eta)\pi}{m}\erf\left(\nicefrac{\sigma^2}{\displaystyle\sqrt\frac{4\sigma^2(1+\eta)\pi}{m}}\right),
\label{scnd_bound}
\end{align}
with $\erf(\theta)\defsym \frac{2}{\sqrt{\pi}}\displaystyle\int_{0}^{\theta}\exp(-\mu^2)d\mu$.
}
\end{restatable}
Interestingly, the term~\rev{$V_g-  \E \big[ {V}(f^\star_m)  \big]$} is the average error induced when estimating the function $f^\star_m$ using $m$ samples from the training set while~\rev{$\E \big|{V}(f^\star_m)-  \widehat{V}_n(f^\star_m)  \big| $} indicates the average error incurred when estimating the \gv\ using $n$ samples from the validation set. Clearly, in~\cref{frst_bound}, the scaling of these bounds with the number of samples are very different which can be made evident by using the order notation:
\mrev{\begin{align} 
\sup_{P_{XY}} \E \big|{V}(f^\star_m)-  \widehat{V}_n(f^\star_m)  \big|  & \in   \mathcal{O}\left(  \frac{1}{\sqrt{ n}}\right), \label{missing-2}\\
\sup_{P_{XY}}\left\{ V_g-  \E \big[ {V}(f^\star_m)  \big]\right\}    &\in   \mathcal{O}\left(   \frac{|\mathcal{H}|}{\sqrt{m}}\right). \label{missing-1}
\end{align}
These  bounds indicate that the error in \eqref{missing-2} vanishes much faster than the error in  \eqref{missing-1} and thus, the size of the training set, in general, should be kept larger than the size of the validation set, i.e.,  $n \ll m$. }


\subsection{Sample complexity}
We now study how large the validation set should be in order to get a good estimation. 
For ${\varepsilon, \delta>0}$, we define the sample complexity as the set of smallest integers $M(\varepsilon, \delta)$ and $N(\varepsilon, \delta)$ sufficient to guarantee that the gap between the true \gv\ and the estimated $ \widehat{V}_n(f^\star_m)$ is at most $\varepsilon$ with at least $1-\delta$ probability: 
\begin{definition}{}
\label{sample_complexity_definition}
For ${\varepsilon, \delta>0}$, let all pairs $\big(M(\varepsilon, \delta),N(\varepsilon, \delta)\big)$  be the set of  smallest $(m,n)$ sizes of training and validation sets such that:  
\begin{align}
\sup_{P_{XY}} \mathbb{P}\left[ | V_g -  \widehat{V}_n(f^\star_m) | > \varepsilon  \right] \leq \delta.
\end{align}
\end{definition}

Next result says that we can bound  the sample complexity in terms of  $\varepsilon, \delta, \sigma$,  and $|b-a|$ 
(see Appendix~\ref{AppD} for the proof).
\begin{restatable}[]{corollary}{trainingsamplesupperboudrestatable}\label{training_samples_upperbound}
The sample complexity  of the ERM algorithm \gv\   is bounded from above  by the set of values  satisfying: 
\rev{
\begin{align}
M(\varepsilon, \delta)&\leq \frac{8\,\sigma^2+\nicefrac{4\, \left(b-a\right)\varepsilon}{3}}{\varepsilon^2}\,\ln\left(\frac{2\, |\Hcal|}{\delta-\Delta}\right),
\label{sample_complex1}\\
N(\varepsilon, \delta)&\leq \frac{2\,\sigma^2+\nicefrac{2\,\left(b-a\right)\varepsilon}{3}}{ \varepsilon^2}\,\ln\left(\frac{2}{\Delta}\right),
\label{sample_complex2}
\end{align}
for all $\Delta$ such that $0< \Delta <\delta$.
}
\end{restatable}

 The theoretical results of this section are very general and \rev{cannot take full advantage of the specific properties  of a particular model or data distribution. }\rev{ In particular, it is important to emphasize that the upper bound in~\eqref{eq:new} is independent of the learned function $f^\star_m$ because ${|\widehat{V}_n(f^\star_m) - V(f^\star_m)|\leq\max_{f\in\mathcal{H}} |\widehat{V}_n(f) - V(f)|}$ and because of~\eqref{upperbounds1}, and thus  these bounds are independent of the specific algorithm and training sets in used to solve the optimization in \eqref{eq-solving-emr}. }
Furthermore, the  $f$  maximizing  $| {V}(f) - \widehat{V}_n(f) |$ in those in-equations  is not necessarily what the algorithm would choose. Hence the bounds given in  \autoref{thm_estimate_convergence} and \autoref{training_samples_upperbound}  in general are not tight. However, these theoretical bounds provide a worst-case measure from which learnability holds for all data sets.

In the next section, we will propose an approach for  selecting   $f^\star_m$ and   estimating   $V_g$. The experiments  in Section~\ref{experiments}  suggest that our method usually estimates $V_g$ much more accurately than what is indicated by \autoref{thm_estimate_convergence}.  


\section{From  \gv\ to Bayes vulnerability via pre-processing}
\label{pre-processing}

This is the core section of the paper, which describes our approach to select the $f^\star_m$  to estimate  $V_g$. 

In principle one could train a neural network  to learn $f^\star_m$ by using $-\widehat{V}_m(f)$ as the loss function, and 
minimizing  it over the $m$ training samples (cfr. \autoref{eq-solving-emr}).  
However, this would require  $\widehat{V}_m(f)$ to be a differentiable function of the weights of the neural network, 
so that its gradient can be computed during the back-propagation. 
Now, the problem is that the  $g$ component of $\widehat{V}_m(f)$ is essentially a non-differentiable function, so it would need to be approximated by a suitable differentiable (surrogate) function, (e.g., as it is the case of the Bayes error via the cross-entropy).
Finding an adequate  differentiable function to replace each   possible $g$   may be a challenging task in practice. If this surrogate does not preserve the original dynamic of the gradient of $g$ with respect to $f$, the learned  $f$  will be far from being optimal.


In order to circumvent this issue, we propose a different approach, which presents two main advantages:
\begin{enumerate}
\item \label{advantageOne} it reduces the problem of learning $f^\star_m$ to a standard classification problem, 
therefore it does not require a different loss function to be implemented for each adversarial scenario;
\item  it can be implemented by using \emph{any} universally consistent learning algorithm (i.e., any ML algorithm approximating the ideal Bayes classifier).
\end{enumerate}


The reduction described in the above list~(\autoref{advantageOne}) is  
based on the idea that, in the $g$-leakage framework, 
the adversary's goal is not to directly infer the actual secret 
$x$, but rather to select the optimal guess $w$ about the secret. 
As a consequence, the training of the ML classifier to produce $f^\star_m$ 
should not be done on pairs of type $(x,y)$, but rather of type $(w,y)$, 
expressing the fact that the best guess, 
in the particular run which produced $y$, is $w$. 
This shift from  $(x,y)$ to  $(w,y)$ is via a
\emph{pre-processing} and we propose two distinct and systematic ways 
to perform this transformation,  called \emph{data} 
and \emph{channel pre-processing}, respectively. 
The two methods are illustrated in the following sections. 

We remind that, according to~\autoref{learning_gv}, we restrict, wlog,  to  non-negative $g$'s. If $g$ takes negative values, then it can be shifted by adding $ - \min_{w,x} g(w,x)$, without consequences for the $g$-leakage value (cfr. \cite{Alvim:12:CSF,Alvim:14:CSF}).
Furthermore we assume that there exists at least a pair $(x, w)$ such that $\pi_x \cdot g(w,x)>0$. Otherwise $V_g$ would be $0$ and the problem of estimating it will be trivial.

\subsection{Data pre-processing}
\label{data_preproc}

The data pre-processing  technique  is completely black-box in the sense that it does not need access to the channel. We 
only assume the availability of a set of pairs of type $(x,y)$, sampled according to $\Joint{\pi}{C}$,  
the input-output distribution of the channel.  
This set could be provided by a third party, for example. 
We divide the set in \trainingsetsub\ (training) and  \validationsetsub\ (validation), containing $m$ and $n$ pairs, respectively. 

For the sake of simplicity, to describe this technique  we assume that $g$ takes only integer values, in addition to being non-negative. 
The construction for the general case is discussed in Appendix~\ref{App_data_preproc_no_real}. 
\begin{algorithm}[!htb]
\small
	 \caption{Algorithm for data pre-processing}\label{alg:dataprep}
	\emph{Input}: \trainingsetsub; \emph{Output}: \trainingsetdtpreprocsub; \\
	1. \trainingsetdtpreprocsub $\vcentcolon =\emptyset$;\\
	2. For each $x,y$,   let $u_{xy}$ be the  number of copies of $(x,y)$ in \trainingsetsub;\\
	3. For each $x,y,w$,   add $u_{xy}\cdot g(w,x)$ copies of $(w,y)$ to \trainingsetdtpreprocsub.
\end{algorithm}

The idea behind the data pre-processing technique is that the effect of the gain function can be represented in the transformed dataset 
by amplifying the impact of the guesses in proportion to their reward. For example, consider a  pair $(x,y)$    in \trainingsetsub, and assume that 
the reward for the guess $w$ is $g(w,x) = 5$, while for another guess $w'$ is $g(w',x) = 1$. Then in the transformed dataset \trainingsetdtpreprocsub\ 
this pair will contribute with $5$ copies of $(w,y)$ and only $1$ copy of $(w',y)$. The transformation is described in Algorithm~\ref{alg:dataprep}. Note that in general 
it causes an expansion of the original dataset. 

\subsubsection*{Estimation of $V_g$}  Given \trainingsetsub, we construct the set \trainingsetdtpreprocsub\ of pairs  $(w,y)$  according to  Algorithm~\ref{alg:dataprep}. 
Then, we use \trainingsetdtpreprocsub\ to train a classifier $f^\star_{m^\prime}$, using an algorithm that approximates the ideal Bayes classifier. 
As proved below, $f^\star_{m^\prime}$ gives the same mapping $\Ycal\rightarrow\Wcal$  as the optimal empirical rule $f^\star_m$ on \trainingsetsub\ (cfr. \autoref{empirical_max}).
Finally, we use $f^\star_m$ and \validationsetsub\ to  compute  the estimation of  $V_g(\pi,C)$  as in  \eqref{eq-empirical-g-leakege}, with $f$ replaced by $f^\star_m$. 
 
%
\subsubsection*{Correctness}
We first need some notation. 
For each $(w,y)$, define:
\begin{equation}
U(w,y) \; \defsym \;  {\sum_x  \pi_x \cdot C_{xy} \cdot g(w,x)}~, 
\end{equation}
which represents the ``ideal''  proportion of copies of  $(w,y)$ that \trainingsetdtpreprocsub\ should contain.
From $U(w,y)$ we can now derive the ideal joint distribution on $\Wcal\times\Ycal$ and the marginal 
on $\Wcal$: 
\begin{eqnarray} \label{alpha}
P_{WY}(w,y) \,\defsym\, \frac{U(w,y)}{\alpha},
&\quad\mbox{where}\quad&
\alpha \;\defsym\; {\sum_{y,w} U(w,y)}\;,
\end{eqnarray}
(note that $\alpha>0$ because of the assumption on $\pi$ and $g$),
\rev{\begin{eqnarray} 
\xi_w &\defsym& \sum_y P_{WY}(w,y). \label{marginal}
\end{eqnarray}}
The  channel of the  conditional probabilities of $y$ given $w$ is:
\rev{\begin{equation}\label{channelE}
E_{wy} \; \defsym \; \frac{P_{WY}(w,y)}{\xi_w}~.
\end{equation}}

Note that \rev{$P_{WY} = \Joint{\xi}{E}$}. By construction, it is clear that the  \trainingsetdtpreprocsub\ generated by  Algorithm~\ref{alg:dataprep} could have been generated, 
 with the same probability, by sampling \rev{$\Joint{\xi}{E}$}. 
The following theorem, whose proof is  in Appendix~\ref{App_data_preproc}, establishes that the \gv\  of $\Joint{\pi}{C}$ is equivalent to the Bayes vulnerability  
of \rev{ $\Joint{\xi}{E}$}, and hence it is correct  to estimate $f^\star_m$ as an empirical  Bayes classifier $f^\star_{m^\prime}$ trained on  \trainingsetdtpreprocsub.
\begin{restatable}[Correctness of data pre-processing]{theorem}{datapreprocthm}\label{data_preproc_thm}
Given  a prior $\pi$, a channel $C$, and a gain function $g$, we have
	\rev{\[
		V_g(\pi,C)
		~=
		\alpha\cdot V_{\GId}(\xi, E) 
		~,
	\]}
	where 
\rev{$\alpha, \xi$} and $E$ are those defined in   \eqref{alpha}, \eqref{marginal} and \eqref{channelE}, respectively, and $\GId$ is the identity  function (cfr. \autoref{sec:preliminaries}), i.e., the gain function corresponding to the  Bayesian adversary.  
\end{restatable}
\rev{
\subsubsection*{Estimation error}
To reason about the error we 
need to consider  the optimal empirical classifiers. 
Assuming that we can perfectly match the 
$U(w,y)$ above with the  $u_{xy}$ of Algorithm~\ref{alg:dataprep}, we can repeat the same reasoning as above, thus obtaining $\widehat{V}_m(f)=\alpha\cdot \widehat{V}_{m^\prime}(f)$,
where  ${V}_m(f)$ is the empirical functional defined in \eqref{eq-solving-emr}, and $\widehat{V}_{m^\prime}(f)$ is the corresponding empirical functional  for $\GId$ evaluated in  $\mathcal{D}_{m^\prime}$:  
\begin{align}
\widehat{V}_{m^\prime}(f)\defsym \frac{1}{m'} \sum_{(w,y) \in \mathcal{D}_{m^\prime}} \GId \big( f(y), w \big)
\end{align}
$f^\star_{m'}$ is the maximizer of this functional, i.e. $f^\star_{m^\prime}\defsym  \argmax\limits_{f\in\Hcal} \widehat{V}_{m^\prime} (f)$.
Therefore we have:
{\small
\begin{align*}
f^\star_{m}  =   \argmax\limits_{f\in\Hcal} \widehat{V}_m (f) 
= \argmax\limits_{f\in\Hcal} (\alpha\cdot \widehat{V}_{m^\prime}(f) )
 = \argmax\limits_{f\in\Hcal} \widehat{V}_{m^\prime}(f) 
= f^\star_{m^\prime}.
\end{align*}
}
A bound on the estimation error of this method can therefore be obtained by using the theory developed in previous section, 
applied to the Bayes classification problem. 
} 
\rev{
Remember that the estimation error is $f^\star_{m}$ is  $| V_g - \widehat{V}_n(f^\star_m) |$. With respect to the estimation error of the corresponding Bayes classifier, we have a magnification of a factor $\alpha$ as shown by the following formula, where $\widehat{V}_{n^\prime}$ represents the empirical functional for the Bayes classifier: 
{\small 
\begin{align*}
| V_g - \widehat{V}_n(f^\star_m) | = | \alpha\cdot V_{\GId} - \alpha\cdot  \widehat{V}_{n^\prime}(f^\star_{m^\prime}) | = \alpha\cdot |  V_{\GId} -   \widehat{V}_{n^\prime}(f^\star_{m^\prime}) | .
\end{align*}
}However, the normalized estimation error (cfr.~\cref{experiments}) remains  the same because both numerator and denominator are magnified by a factor $\alpha$.
}

\rev{
Concerning the probability that the error is above a certain threshold $\varepsilon$, we have the same bound as those for the Bayes classifier 
 in \autoref{mainproposition} and the other results of previous section, where $\varepsilon$ is replaced by $\alpha \varepsilon$, $m, n$ by $m', n'$, $\sigma^2$ by $\alpha^2 \sigma^2$,
 and  $|b-a| = 1$ (because it's a Bayes classifier). 
 It may sounds a bit surprising that the error for the estimation of the \gv\ is not much worse than for the estimation of the Bayes error, but we recall that we are essentially solving the same problem, only every quantity is magnified by a factor $\alpha$. 
 Also, we are assuming that we can match perfectly $U_{xy}$ by $u_{xy}$. When $g$ ranges in a large domain this may not be possible, and we should rather resort to the channel pre-processing method described in the next section.
}

\subsection{Channel pre-processing}
\label{channel_preproc}

For this technique we assume  \emph{black-box access} to the system, i.e., that we can execute the system while controlling each input, and collect the corresponding output.

The core idea behind this technique is to transform the input of $C$ into entries of type $w$, and to ensure that the distribution on the $w$'s  reflects the corresponding rewards expressed by $g$. 

More formally, let us define a distribution $\tau$ on $\Wcal$ as follows: 
\begin{eqnarray} \label{eq:tau_beta}
\tau_w \,\defsym\, \frac{\sum_x \pi_x \cdot g(w,x)}{\beta}
&\;\;\mbox{where}\;\; & 
\beta \,\defsym\, \sum_{x,w} \pi_x \cdot g(w,x)\, ,
\end{eqnarray}
(note that $\beta$ is strictly positive because of the assumptions on $g$ and $\pi$), and let us define the following  matrix $R$ from $\Wcal$ to $\Xcal$:
\begin{eqnarray}\label{eq:R} 
R_{wx} \,\defsym\, \frac{1}{\beta}\cdot \frac{1}{\tau_w} \cdot  \pi_x \cdot g(w,x)\, .
\end{eqnarray}
It is easy to check that $R$ is a stochastic matrix, hence the composition $RC$ is a channel. 
It is important  to emphasize the following: \\

{\sc Remark}\;
In the above definitions, $\beta, \tau$ and $R$ depend solely on $g$ and $\pi$, and not on $C$.
\\[1ex]
The above property is crucial to our goals, because in the black-box approach we are not supposed to rely on the knowledge of $C$'s internals. 
We now illustrate how we can estimate $V_g$ using the pre-processed channel $RC$. 

\subsubsection*{Estimation of \, $V_g$}  Given $RC$ and $\tau$, we build a set \trainingsetchpreprocsub\  consisting of   pairs  of type $(w,y)$ sampled from 
$\Joint{\tau}{RC}$. We also construct a set \validationsetsub\ of   pairs  of type $(x,y)$ sampled from $\Joint{\pi}{C}$.
Then, we use \trainingsetchpreprocsub\ to train a classifier $f^\star_{m}$, using an algorithm that approximates the ideal Bayes classifier. 
Finally, we use $f^\star_m$ and \validationsetsub\ to  compute  the estimation of  $V_g(\pi,C)$  as in  \eqref{eq-empirical-g-leakege}, with $f$ replaced by $f^\star_m$. 

Alternatively, we could estimate  $V_g(\pi,C)$ by computing the empirical Bayes error of  $f^\star_m$ on a validation set \validationsetsub\ of type $(w,y)$  sampled from 
$\Joint{\tau}{RC}$, but the estimation would be less precise. Intuitively, this is because $RC$ is more ``noisy'' than $C$.  
 \subsubsection*{Correctness}
The correctness of the channel pre-processing method is given by the following theorem, which shows that we can learn $f^\star_{m}$ by 
training a Bayesian classifier on a  set sampled from $\Joint{\tau}{RC}$.

\begin{restatable}[Correctness of channel pre-processing]{theorem}{channelpreproc}\label{thm:channel-pre-processing}
	Given a prior $\pi$ and a gain function $g$, we have that, for any channel $C$:
	\[
		V_g(\pi,C)
		~=
		\beta\cdot V_{\GId}(\tau, RC) 
		\quad
		\text{for all channels }C.
	\]
	where $\beta$, $\tau$ and $R$ are those defined in \eqref{eq:tau_beta} and \eqref{eq:R}. 
\end{restatable}

Interestingly,  a result  similar to \autoref{thm:channel-pre-processing} is also given in \cite{Bordenabe:16:CSF}, although the context is completely different from ours: 
the focus of \cite{Bordenabe:16:CSF}, indeed, is to study how the leakage of $C$ on $X$ may induce also a leakage of other sensitive information  
$Z$ that has nothing to do with $C$ (in the sense that is not information manipulated by  $C$). We intend to explore this connection in the context 
of a possible extension of our approach to this more general scenario. 
 
\rev{
\subsubsection*{Estimation error}
Concerning the estimation error, the results are essentially the same as in previous section (with $\alpha$ replaced by $\beta$).
As for the bound on the probability of error, the results are worse, because the original variance $\sigma^2$ is magnified by the channel pre-processing, 
which introduces a further factor of randomness in the sampling of training data in \ref{eq-Pro1}, which means that in practice this bound is more difficult to estimate.
}

\subsection{Pros and cons  of the two methods}\label{sec:ComparisonDataChannel}
The fundamental advantage of data pre-processing is that it allows to estimate $V_g$ from just samples of the system, without even black-box access.
In contrast to channel pre-processing, however, this method is particularly sensitive to the values of the gain function $g$.
Large gain values will increase the size of \trainingsetdtpreprocsub, with consequent increase of the computational cost for estimating the \gv.
Moreover, if $g$ takes real values then we need to apply the technique described in Appendix~\ref{App_data_preproc_no_real}, 
which can lead to a large increase of the dataset as well.
In contrast, the channel pre-processing method has the advantage of  controlling the size of the training set, but it can be applied only when it is possible to interact with the channel by providing input and collecting output. 
Finally, from the precision point of view, we expect the estimation based on  data pre-processing   to be more accurate when $g$ consists of small integers, because the channel pre-processing introduces some extra noise in the channel. 
\section{Evaluation}
\label{experiments}

In this section we evaluate our approach to the estimation of $g$-vulnerability. 
We consider four different scenarios:
\begin{enumerate}
    \item $\Xcal$ is a set of (synthetic) numeric data, the channel $C$ consists of \emph{geometric noise},  and  $g$ is the \emph{multiple guesses} gain function, 
    representing an adversary that is allowed to make several attempts to discover the secret.  
   \item $\Xcal$ is a set of locations from the Gowalla dataset~\cite{Gowalla:2011:Online},  $C$  is the optimal noise of Shokri et al.~\cite{Shokri:14:CCS}, and $g$ is one of the functions used to evaluate the privacy loss in \cite{Shokri:14:CCS}, namely a function anti-monotonic on the distance, representing the idea that the more the adversary's guess is close  to the target (i.e., the real location), the more he gains.
    \item $\Xcal$ is the Cleveland heart disease dataset~\cite{Cleveland:1989:Online},  $C$ is a \emph{differentially private} (DP) mechanism~\cite{Dwork:06:TCC,Dwork:06:ICALP}, and $g$ assigns higher values to worse heart conditions, modeling an adversary that aims at discovering whether a patient is at risk (for instance, to deny his application for health insurance).
   \item $\Xcal$ is a set of passwords of $128$ bits and  $C$ is a \emph{password checker} that leaks the time before the check fails, but mitigates the timing attacks by applying some random delay and the bucketing technique (see, for example, \cite{Koepf:09:CSF}). The function $g$ represents the part of the password under attack.
\end{enumerate}
For each scenario, we proceed in the following way:
\begin{itemize}
    \item We consider $3$ different samples sizes for the training sets that are used to train the ML models and learn the $\Ycal \rightarrow \Wcal$ remapping. This is to evaluate how the precision of the estimate depends on the amount of data available,    and on its relation with the size of  $|\Ycal|$. 
    \item In order to evaluate the variance of the precision, for each training size we create $5$ different training sets, and 
    \item for each trained model we estimate the $g$-vulnerability using $50$ different validation sets.
\end{itemize}

\subsection{Representation of the results and metrics}
We graphically represent the results of the experiment as box plots, using one box for each size. More precisely, given a specific size, let $\widehat{V}^{ij}_n$ be the \gv\ estimation on the $j$-th validation set computed with a model trained over the $i$-th training set (where $i\in\{1,\ldots, 5\}$ and $j\in\{1,\ldots, 50\}$). Let  $V_g$ be the real \gv\ of the system. We define the \emph{normalized estimation error} $\delta_{ij}$ and the  mean value $\overline{\delta}$  of the $\delta_{ij}$'s as follows:
\begin{align}
\delta_{ij}\defsym\frac{|\widehat{V}^{ij}_n-V_g|}{V_g}\, , \quad \textrm{with} \quad 
\overline{\delta} \defsym \frac{1}{250} \sum_{i=1}^{5}\sum_{j=1}^{50}\delta_{ij}\, .
\end{align}
In the graphs, the $\delta_{ij}$'s are  reported next to the box corresponding to the   size, and 
$\overline{\delta}$ is the black horizontal line inside the box. 

Thanks to the normalization the $\delta_{ij}$'s allow  to compare   results among different scenario and different levels of (real) \gv. 
Also, we argue that  the percentage of the error is more meaningful than the absolute value. 
The interested reader, however, can find in Appendix~\ref{extra_plots} also   plots  showing the estimations of  the \gv\, and their distance from the  real  value. 

We also consider the following typical measures of precision:

\begin{align}
   & \textrm{dispersion} \defsym \sqrt{\frac{1}{250}\sum_{i=1}^{5}\sum_{j=1}^{50}(\delta_{ij}-\overline{\delta})^2}\, ,\\
	&\textrm{total error} \defsym \sqrt{\frac{1}{250}\sum_{i=1}^{5}\sum_{j=1}^{50}\delta_{ij}^2}\, .
\end{align}
The dispersion is an average measure of how far the normalized estimation errors are  from their mean value when using same-size training and validation sets. 
On the other hand, the total error is an average measure of the normalized estimation error, when using same-size training and validation sets. 

In order to make a fair comparison between the two   pre-processing methods, intuitively we need to use training and validation sets of ``equivalent size''. For the validation part, since the 
``best'' $f$ function has been already found and therefore we do not need any pre-processing anymore, ``equivalent size'' just means same size. But what does it  mean, in the case of training sets built with different pre-processing methods?  
Assume that we have a set of data \trainingsetsub\ 
coming from a third party collector (recall that  $m$ represents the size of the set),
and let \trainingsetdtpreprocsub\ 
be the result of the data pre-processing on \trainingsetsub.
Now, let \trainingsetchpreprocsub\ 
be the dataset obtained drawing 
samples according to the channel pre-processing method. 
Should we impose $m^{\prime\prime}= m$ or $m^{\prime\prime}= m^{\prime}$?
We argue that the right choice is the first one, because  the amount of ``real'' data  collected is $m$. 
Indeed, \trainingsetdtpreprocsub\ is generated synthetically from \trainingsetsub\ and cannot contain more information about  $C$ than \trainingsetsub, despite its larger size.

\subsection{Learning algorithms}
We consider two ML algorithms in the experiments:   k-Nearest Neighbors (k-NN) and   Artificial Neural Networks (ANN).
We have made however a slight modification of k-NN algorithm, due to the following reason: recall that, depending on the particular gain function, the data pre-processing method might create many instances where a certain observable $y$ is repeated multiple times in
pair with different $w$'s. For the k-NN algorithm, a very common choice is to consider a number of neighbors which is equivalent to natural logarithm of the
total number of training samples. In particular, when the data pre-processing is applied, this means that $k=\log(m^{\prime})$ nearest neighbors will be considered for the classification decision. Since $\log(m^{\prime})$ grows slowly with respect to $m^{\prime}$, it might happen that k-NN fails to
find the subset of neighbors from which the best remapping can be learned. 
To amend this problem, we modify the k-NN algorithm in the following way:   instead of looking for neighbors among all the $m^{\prime}$ samples, we only consider a subset of $l \leq m^{\prime}$ samples, where each value $y$ only appears once. After the $\log(l)$ neighbors have been detected among
the $l$ samples, we select $w$ according to a majority vote over the $m^{\prime}$ tuples $(w,y)$ created through the remapping.

The distance on which the notion of neighbor is based depends on the experiments. We have considered the standard distance among numbers in the first and fourth experiments, the Euclidean distance in the second one, and the Manhattan distance  in the third one, which is a stand choice for DP. 

Concerning the ANN models, their specifics  are in Appendix~\ref{AppG}. Note that,  for the sake of fairness, we  use the same architecture for both pre-processing methods, although we adapt number of epochs and batch size to the particular dataset we are dealing with. 


\subsection{Frequentist approach}
In the experiments, we will compare our method with 
the frequentist one.  This approach has been proposed originally in \cite{Chatzikokolakis:10:TACAS} for estimating mutual information, and extended successively also to min-entropy leakage~\cite{Chothia:13:CAV}. Although not considered in the literature,  the extension to the case of \gv\  is straightforward. 
The method consists in estimating the probabilities that constitute the channel matrix $C$, and then calculating   analytically the \gv\ on $C$. 
The precise definition is in ~\autoref{freq_description}.

In \cite{Cherubin:19:SP} it was observed that  in the case of the Bayes error
the frequentist approach performs poorly when the size of the observable domain  $|\Ycal|$ is large with respect to the available data. 
We want to study whether this is the case also for \gv. 

For the experiment on the multiple guesses the comparison is illustrated in the next section. For the other experiments, because of lack of space, we have reported it in the Appendix~\ref{extra_plots}.

\subsection{Experiment $1$: multiple guesses}
\begin{figure}[!htb]
	\centering
	\makebox[\linewidth][c]{\includegraphics[scale=.25]{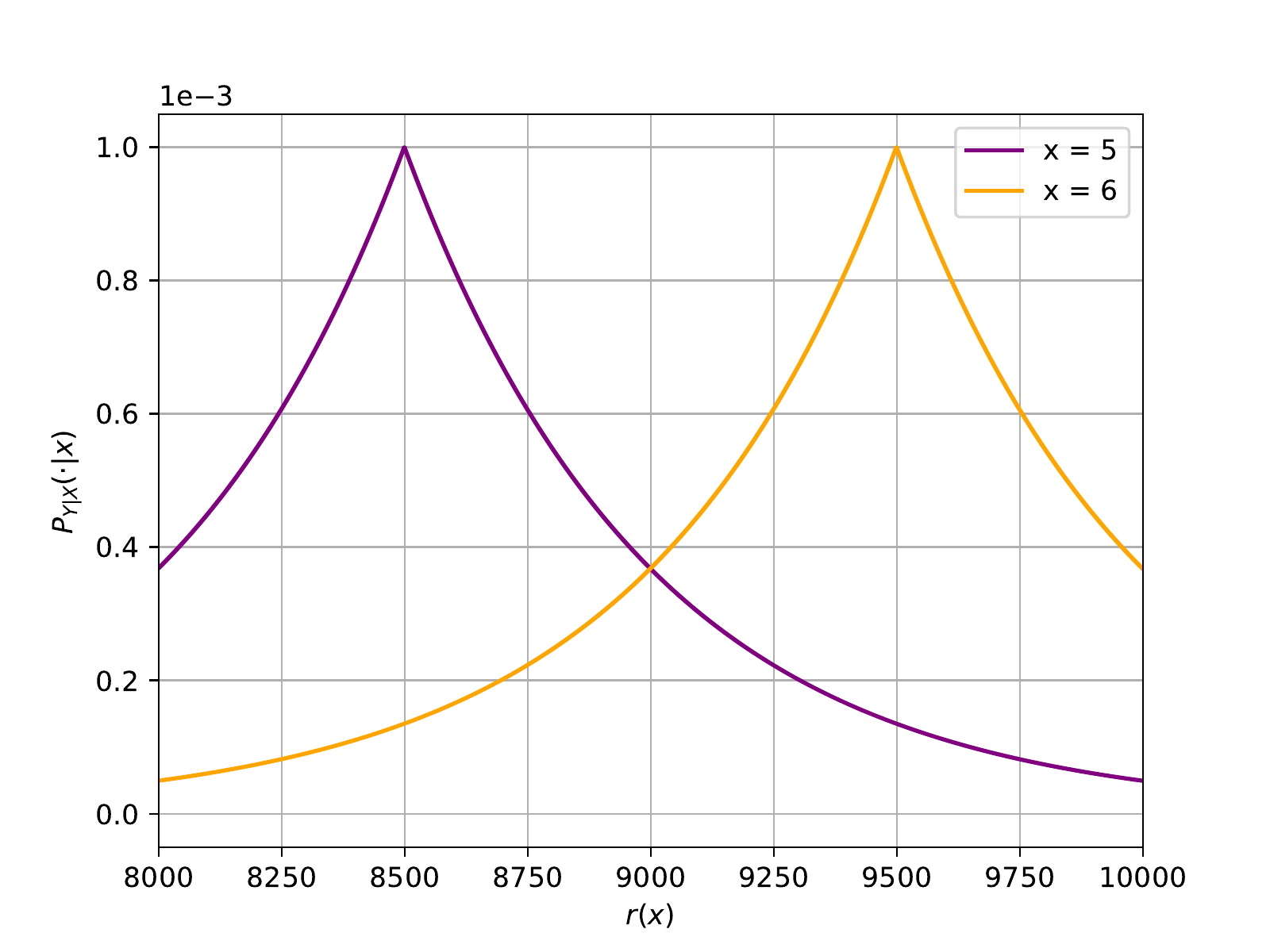}}
	\caption{The channel of  Experiment $1$. The two curves represent the distributions $P_{Y|X}(\cdot|x)$ for two adjacent secrets:  $x=5$ and $x=6$.}
	\label{fig:mult_guess_distr}
\end{figure}
\label{mult_guess_section}
We consider a system in which the secrets $\Xcal$ are the  integers between $0$ and $9$, and  the observables $\Ycal$ are the integers between $0$ and $15999$. 
Hence $|\Xcal|=10$ and  $|\Ycal|=16$K. The rows of the channel $C$ are geometric distributions centered on the corresponding secret:
\begin{equation}
    C_{xy} = P_{Y|X}(y|x)=\lambda\,\exp(-\nu|r(x)-y|)\, ,
    \label{geom_distr_multiguess}
\end{equation}
where:
\begin{itemize}
	\item $\nu$ is a parameter that determines how concentrated around $y=x$ the distribution is. In this experiment we set  $\nu=0.002$;
	\item $r$ is   an auxiliary function that reports $\Xcal$ to the same scale of $\Ycal$, and centers $\Xcal$ on 
	$\Ycal$. Here    $r(x) = 1000\, x + 3499.5$;
  	\item $\lambda=\nicefrac{e^{\nu}-1}{(e^{\nu}+1)}$ is a normalization factor
\end{itemize}
\autoref{fig:mult_guess_distr} illustrates the shape of $C_{xy}$, showing the distributions  $P_{Y|X}(\cdot|x)$ for two adjacent secrets $x=5$ and $x=6$. 
We consider an adversary that can make two attempts to discover the secret (two-tries adversary), and we define the  corresponding gain function as follows.
A guess $w\in\Wcal$ is one of all the possible combinations of $2$ different secrets from $\Xcal$, i.e., $w = \{x_{0}, x_{1}\}$ with $x_{0}, x_{1} \in \Xcal$ and $x_{0} \neq x_{1}$. 
Therefore $|\Wcal|=\binom{10}{2}=45$. The gain function $g$ is  then
\begin{equation}\label{g_multiple}
	g(w,x) = \begin{cases} 1  & \mbox{if } x\in w\\ 0  & \mbox{otherwise}\, . \end{cases}
\end{equation}
For this experiment we consider a uniform prior distribution $\pi$ on $\Xcal$. The true \gv\, for these particular $\nu$ and $\pi$, results to be $V_g = 0.892$. 
%
%
As training sets sizes we consider $10$K, $30$K and $50$K, and 50K for the validation sets.
%
%

%
\subsubsection{Data pre-processing}\label{mult_guess_data_prep} (cfr. Section~\ref{data_preproc}).
%
The plot in~\autoref{fig:mult_guess_est_knn_ann} shows the performances of the k-NN and ANN 
models  in terms of normalized estimation error, while 
\autoref{fig:mult_guess_est_freq_knn_ann} shows  the comparison with the  frequentist approach. 
As we can see,  the precision of the frequentist method is much lower, thus confirming that the trend observed in  \cite{Cherubin:19:SP} for the Bayes vulnerability holds also for \gv. 
\begin{figure}[!htb]
	\centering
	\makebox[\linewidth][c]{\includegraphics[scale=.34]{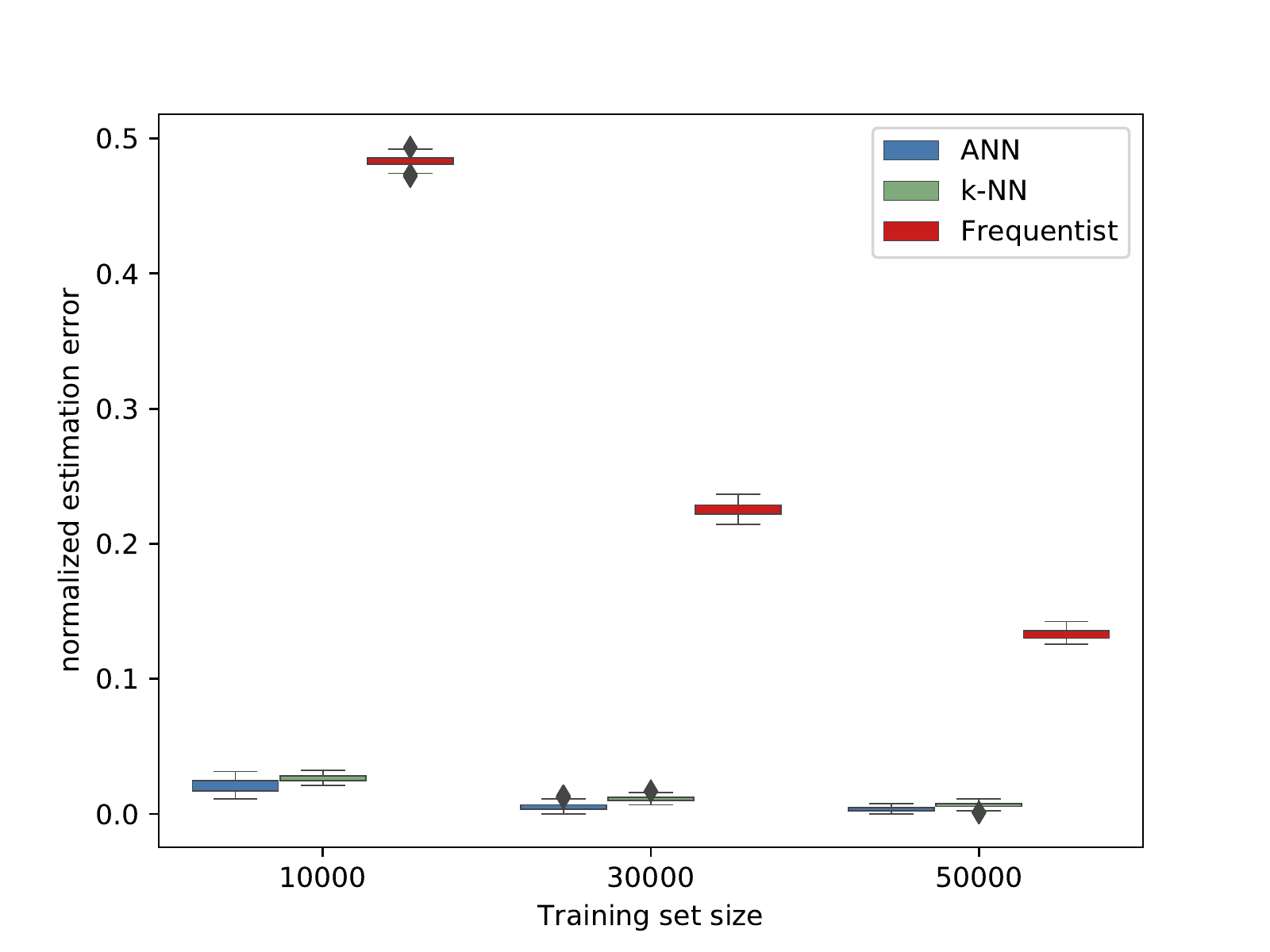}}
	\caption{Multiple guesses scenario, comparison between the frequentist   
	 and the  ML estimations with data pre-processing.}
	\label{fig:mult_guess_est_freq_knn_ann}
\end{figure}
\begin{figure}[!htb]
	\centering
	\makebox[\linewidth][c]{\includegraphics[scale=.34]{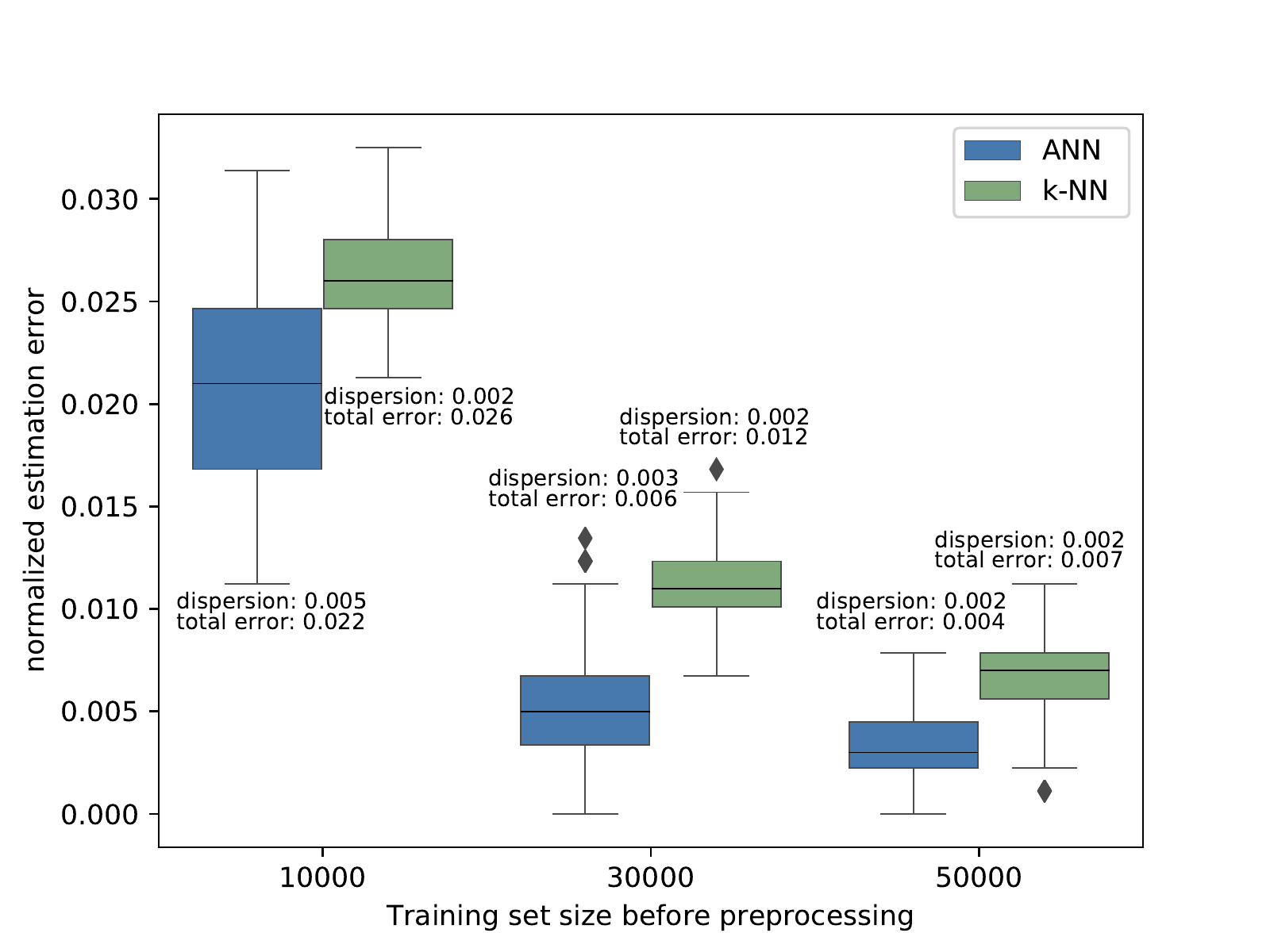}}
	\caption{Multiple guesses scenario, magnification of the part of \autoref{fig:mult_guess_est_freq_knn_ann} on the k-NN and ANN estimations.}
	\label{fig:mult_guess_est_knn_ann}
\end{figure}
It is worth noting that, in this experiment,  the pre-processing of each sample $(x,y)$ creates $9$ samples (matching $y$ with each possible $w \in \Wcal$ such that $w = \{x, x^{\prime}\}$ with $x^{\prime} \neq x$). This means that the sample size of the pre-processed sets is $9$ times the size of the original ones. 
For functions $g$ representing more than $2$ tries this pre-processing method may create training sets too large. In the next section we will see that the alternative channel pre-processing method can be a good compromise. 
%
\subsubsection{Channel pre-processing} (cfr. Section~\ref{channel_preproc}) 
The results for hannel pre-processing are reported in \autoref{fig:mult_guess_est_freq_channel_preproc}. As  we can see, the estimation is worse than with data pre-processing, especially for the  k-NN algorithm. This was to be expected, since the random sampling to match the effect of $g$ introduces a further level of confusion, as explained in Section~\ref{channel_preproc}.  
Nevertheless, these results are still  much better than the frequentist case, so it is a good alternative method to apply when the use of data pre-processing would generate validation sets that are too large, which could be the case when the matrix representing $g$ contains large numbers with a small common divider. 
\begin{figure}[!htb]
	\centering
	\makebox[\linewidth][c]{\includegraphics[scale=.34]{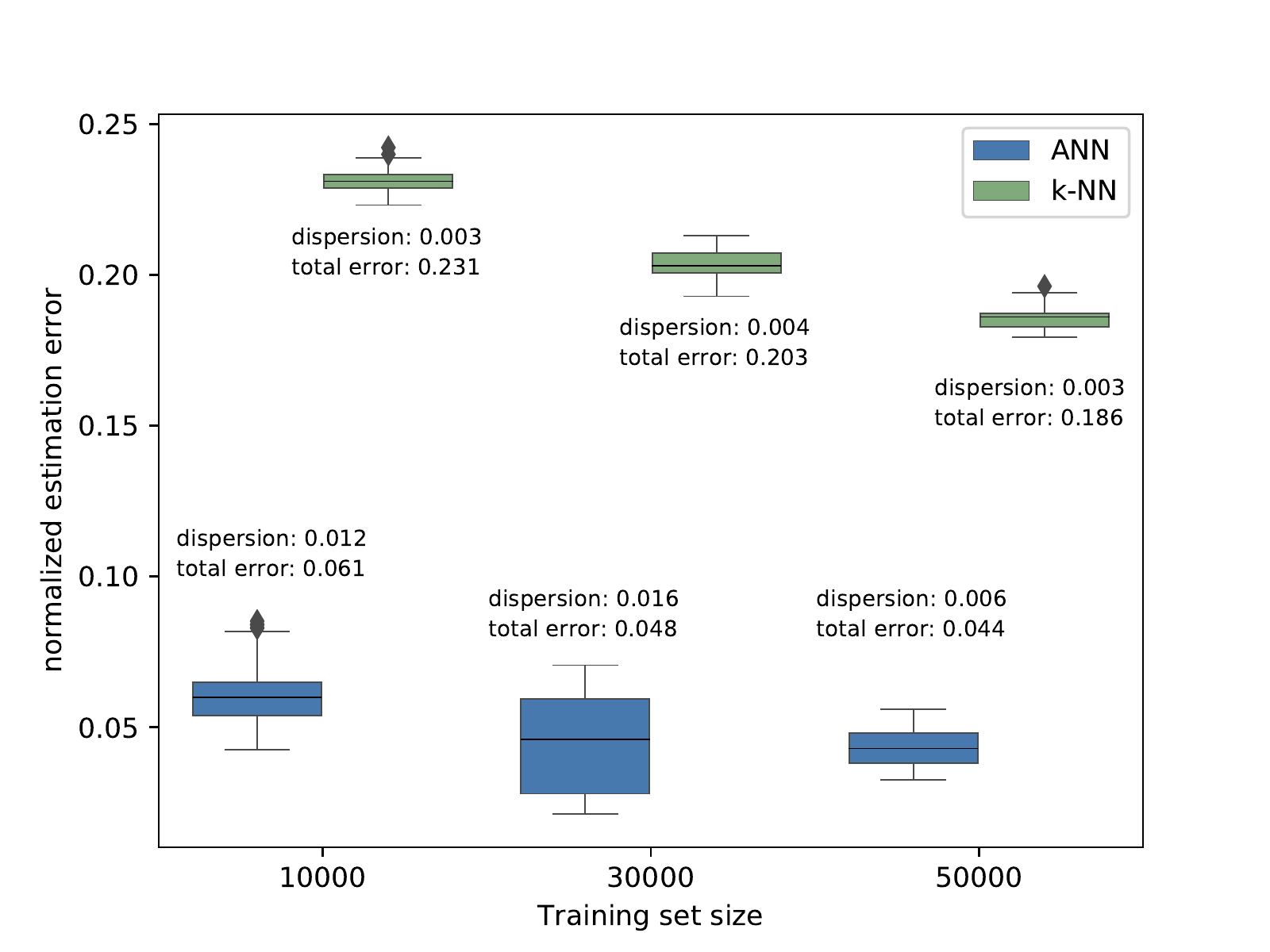}}
	\caption{Multiple guesses scenario, k-NN and ANN estimation with channel pre-processing}
	\label{fig:mult_guess_est_freq_channel_preproc}
\end{figure}
Additional plots  are provided in Appendix~\ref{extra_plots}.

\subsection{Experiment 2: location privacy }
In this section we estimate the \gv\ of a typical  system for  location privacy protection. 
We use data from the open Gowalla dataset~\cite{Gowalla:2011:Online}, which contains the coordinates of  users' check-ins. 
In particular, we consider a square region  in San Francisco, USA, centered in  (latitude,  longitude) =  ($37.755$, $-122.440$), and with $5$Km long sides. In this area Gowalla contains $35162$ check-ins. 

We discretize the region  in 400 cells of 250m long side, and we assume that the adversary's goal is to discover the cell of a check-in. The frequency of the Gowalla check-ins per cell is represented by the heat-map in \autoref{fig:htmp}. 
From these frequencies we can directly estimate the distribution representing the prior  of the secrets~\cite{ElSalamouny:20:EuroSP}. 
\begin{figure}[!htb]
	\centering
	\includegraphics[scale=.10]{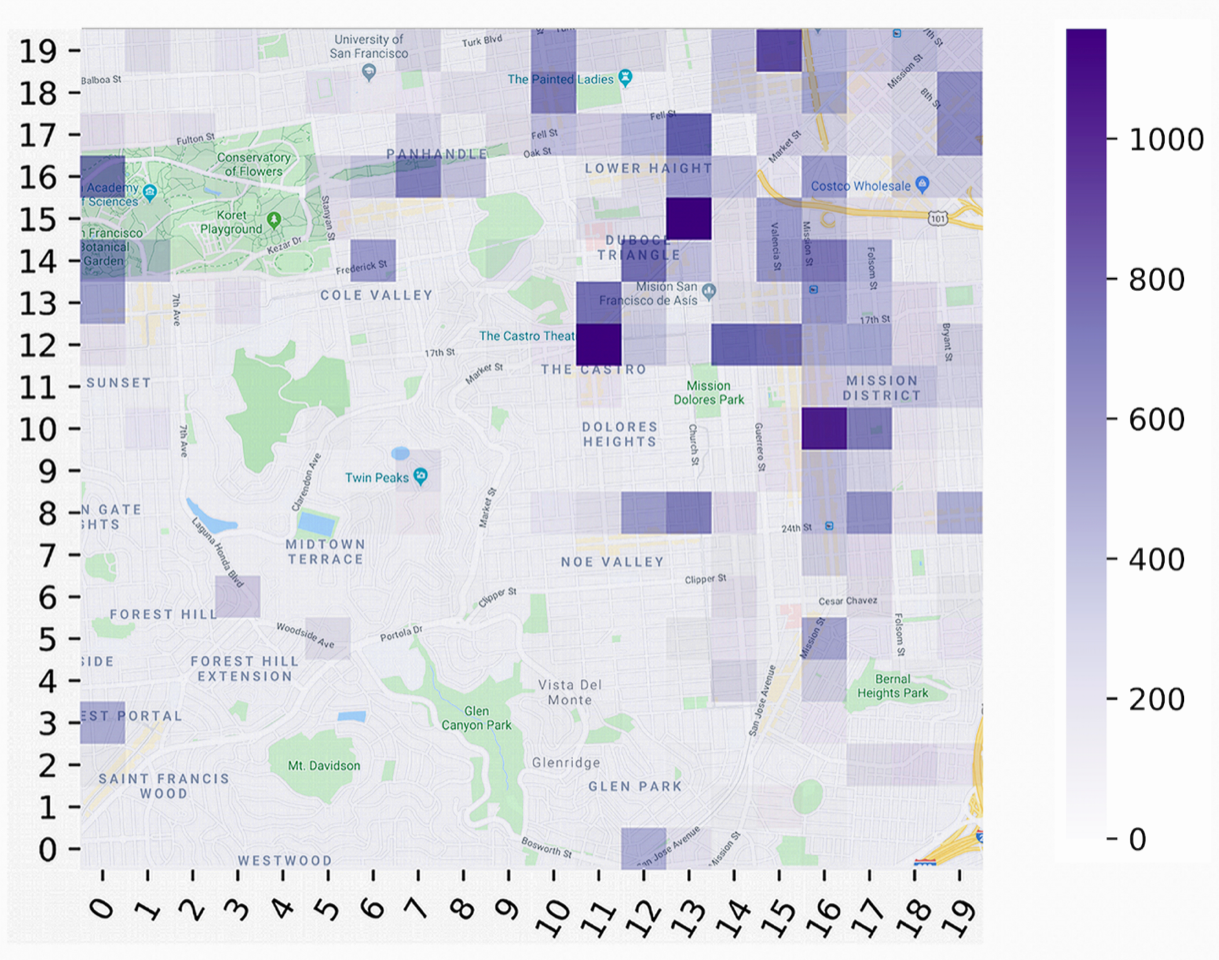}
	\caption{Heat-map representing the Gowalla check-ins distribution in the area of interest; the density of check-ins in each cell is reported in the color bar on the side}
	\label{fig:htmp}
\end{figure}

The channel $C$ that we consider here is the optimal obfuscation mechanism proposed in~\cite{Shokri:14:CCS} to protect location privacy under a utility constraint. 
We recall that the framework  of~\cite{Shokri:14:CCS}  assumes two loss functions, 
one for utility and one for privacy. The utility loss of a mechanism, for a certain prior, is defined as the expected utility loss of the noisy data generated according to the prior and the  mechanism. 
The privacy loss is defined in a similar way, except that we allow the attacker to ``remap'' the noisy data so to maximize the privacy  loss.
For our experiment, we use the Euclidean distance as loss function for the utility, and the $g$ function defined in the next paragraph as loss function for the privacy. 
For further details on the construction of the optimal mechanism we refer to~\cite{Shokri:14:CCS}. 


We define  $\Xcal, \Ycal$ and $ \Wcal$ to be the set of the cells. 
Hence $|\Xcal| = |\Ycal|=|\Wcal| = 400$. 
We consider a  $g$ function representing the precision of the guess in terms of Euclidean distance:  
the   closer   the guess is to the real location, the higher  is the attacker's gain. 
Specifically,  our $g$ is illustrated in \autoref{fig:diamond}, where the central cell represents the real location $x$. 
For a generic ``guess'' cell $w$, the number written in $w$ represent $g(w,x)$.
\footnote{Formally, $g$ is defined as 
$
g(w, x) = \lfloor(\gamma \exp(-\alpha d(w,x) / l))\rceil,
$
where $\gamma=4$ is the maximal gain, $\alpha=0.95$ is a normalization coefficient to control the skewness of the exponential, $d$ is the euclidean distance and $l=250$ is the length of the cells' side. The symbol $\lfloor\cdot\rceil$ in this context represents the rounding to the closest integer operation.}
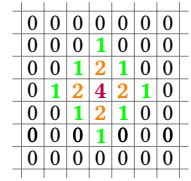
\begin{figure}[!htb]
\centering
\begin{tikzpicture}[scale=0.60]
\draw[step=0.5cm,gray,very thin] (-3.2, -3.2) grid (0.7,0.7);
\node at (-1.25,-1.25) {\color{purple}{\bf 4}};
\node at (-1.25,-0.75) {\color{orange}{\bf 2}};
\node at (-1.25,-1.75) {\color{orange}{\bf 2}};
\node at (-1.75,-1.25) {\color{orange}{\bf 2}};
\node at (-0.75,-1.25) {\color{orange}{\bf 2}};
\node at (-0.75,-0.75) {\color{green}{\bf 1}};
\node at (-0.75,-1.75) {\color{green}{\bf 1}};
\node at (-1.75,-0.75) {\color{green}{\bf 1}};
\node at (-1.75,-1.75) {\color{green}{\bf 1}};
\node at (-1.25,-0.25) {\color{green}{\bf 1}};
\node at (-1.25,-2.25) {\color{green}{\bf 1}};
\node at (-2.25,-1.25) {\color{green}{\bf 1}};
\node at (-0.25,-1.25) {\color{green}{\bf 1}};
\node at (-1.25, 0.25) {0};
\node at (-1.25,-2.75) {0};
\node at (-2.75,-1.25) {0};
\node at ( 0.25,-1.25) {0};
\node at (-0.25,-0.25) {0};
\node at (-0.25,-2.25) {0};
\node at (-2.25,-0.25) {0};
\node at (-2.25,-2.25) {0};
\node at ( 0.25, 0.25) {0};
\node at ( 0.25,-2.75) {0};
\node at (-2.75, 0.25) {0};
\node at (-2.75,-2.75) {0};
\node at (-0.75, 0.25) {0};
\node at (-0.25, 0.25) {0};
\node at (-1.75, 0.25) {0};
\node at (-2.25, 0.25) {0};
\node at ( 0.25, -0.25) {0};
\node at (-0.75, -0.25) {0};
\node at (-1.75, -0.25) {0};
\node at (-2.75, -0.25) {0};
\node at ( 0.25, -0.75) {0};
\node at (-0.25, -0.75) {0};
\node at (-2.25, -0.75) {0};
\node at (-2.75, -0.75) {0};
\node at ( 0.25, -1.75) {0};
\node at (-0.25, -1.75) {0};
\node at (-2.25, -1.75) {0};
\node at (-2.75, -1.75) {0};
\node at ( 0.25, -2.25) {0};
\node at (-0.75, -2.25) {0};
\node at (-1.75, -2.25) {0};
\node at (-2.75, -2.25) {0};
\node at ( 0.25, -2.25) {0};
\node at (-0.75, -2.25) {0};
\node at (-1.75, -2.25) {0};
\node at (-2.75, -2.25) {0};
\node at (-0.25, -2.75) {0};
\node at (-0.75, -2.75) {0};
\node at (-1.75, -2.75) {0};
\node at (-2.25, -2.75) {0};
\end{tikzpicture}
\caption{The ``diamond'' shape created by the gain function around the real secret; the values represent the gains assigned to each guessed cell $w$ when $x$ is the central cell.}
\label{fig:diamond}
\end{figure}

In this experiment we consider  training set sizes of  $100$, $1$k and $10$K samples respectively. 
After applying the data pre-processing transformation,    the  size of the resulting datasets is approximately 
$18$ times that of the original one. This was to be expected, since  the sum of the values of $g$ in  
\autoref{fig:diamond} is $20$. Note that this sum and the increase factor in the dataset do not necessarily coincide, because the latter is  also influenced by the prior and by the mechanism. 

\autoref{fig:geo_loc_knn_ann} and \autoref{fig:geo_loc_knn_ann_bis} show the performance of k-NN  and   ANN
for  both data  and   channel pre-processing. As expected, the  data pre-processing method is more precise than the channel pre-processing one, although only slightly. 
The ANN model is also slightly better than the k-NN in most of the cases. 

%
\begin{figure}[!htb]
	\centering
	\makebox[\linewidth][c]{\includegraphics[scale=.34]{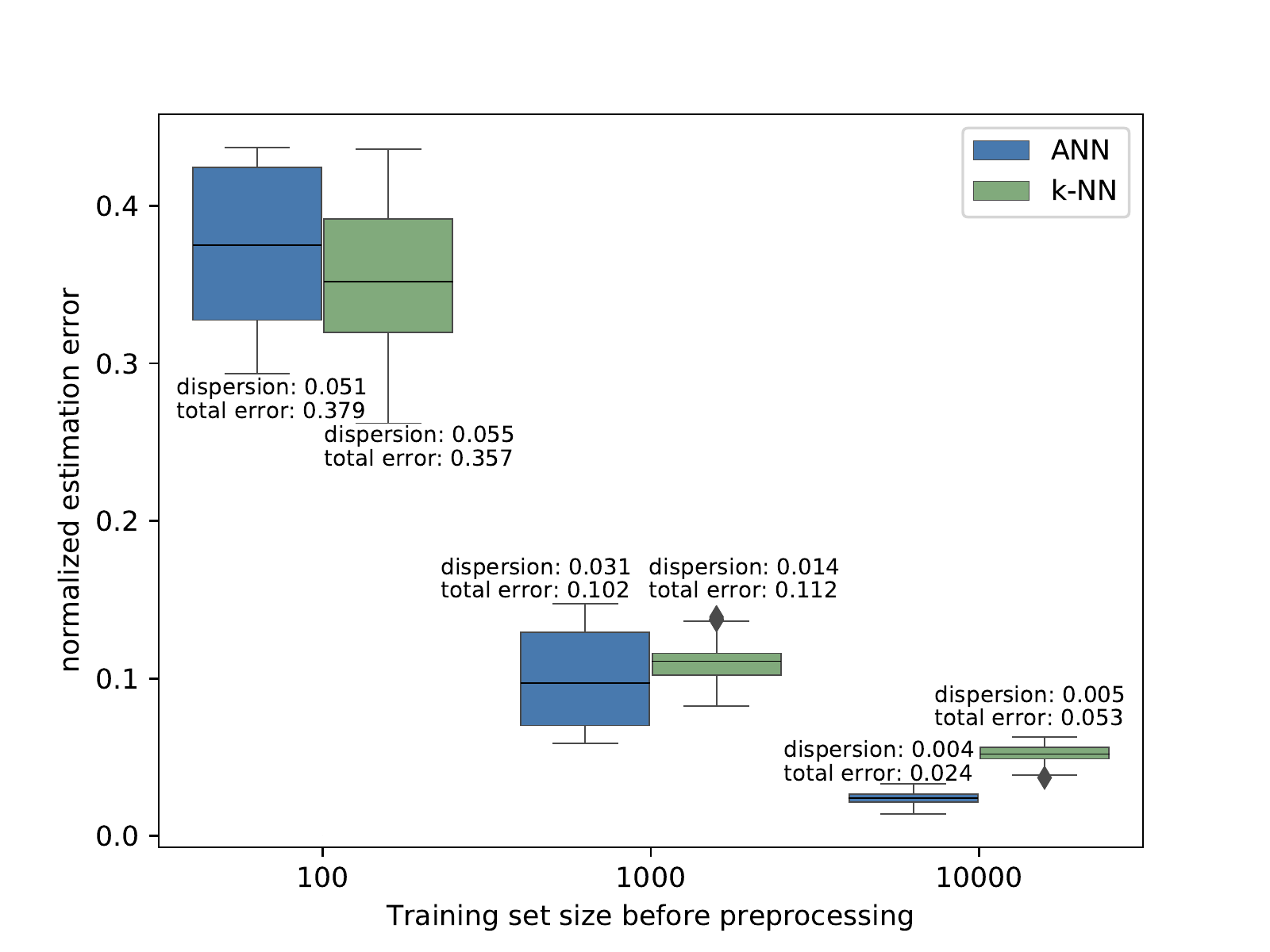}}
	\caption{Location privacy scenario, k-NN and ANN estimation with data pre-processing}
	\label{fig:geo_loc_knn_ann}
\end{figure}
\begin{figure}[!htb]
	\centering
	\makebox[\linewidth][c]{\includegraphics[scale=.34]{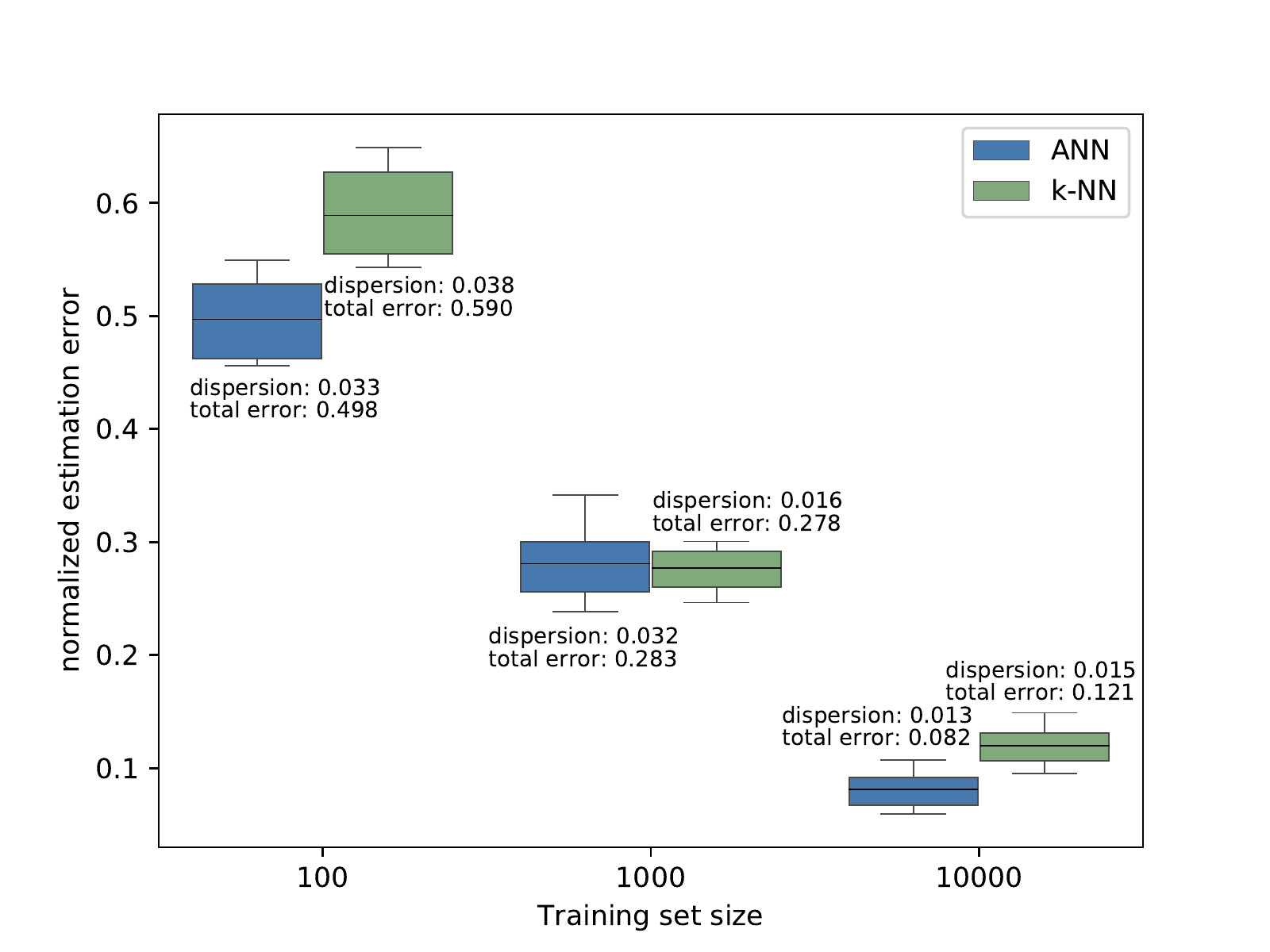}}
	\caption{Location privacy scenario, k-NN and ANN estimation with channel pre-processing}
	\label{fig:geo_loc_knn_ann_bis}
\end{figure}

\subsection{Experiment 3: differential privacy}
In this section we consider a  popular application of DP: individual data protection in medical datasets from which we wish to  extract some statistics via   counting queries. 
%
It is well known that the release of exact information from the database, even if it is only the result of statistical computation on the aggregated data, can leak sensitive information about the individuals. The solution proposed by DP is to obfuscate the released information with carefully crafted noise that obeys certain properties. The goal is to make it difficult to detect whether a certain individual is in the database or not. In other words, two adjacent datasets (i.e., datasets that differ only for the presence of one individual) should have almost the same likelihood to produce a certain observable  result. 

In our experiment, we consider the Cleveland heart disease dataset~\cite{Cleveland:1989:Online} which consist of $303$ records of patients with a medical heart condition. Each condition is labeled by an integer number  indicating the severity: from $0$,  representing a healthy patient, to $4$,   representing a patient whose life is at risk. 

We assume that the hospital publishes the histogram of the patients' records, i.e., the number of occurrences for each label. 
To protect  the patients' privacy, the hospital sanitizes  the histogram by adding  geometric noise (a typical DP mechanism) to each label's count. 
More precisely, if the count of a label is $z_1$, the probability that the corresponding published number is $z_2$ 
is defined by the distribution in \eqref{geom_distr_multiguess}, where  $x$ and $y$ are replaced by $z_1$ and $z_2$ respectively, and $r$ is  $1$.
Note that $z_1$, the real count,  is an integer between $0$ and $303$, while  its noisy version $z_2$ ranges on all integers. 
As for the  value of $\nu$, in this experiment we set it to  $1$.

The secrets space $\Xcal$ is set to be a set of two elements: the full dataset, and the dataset with  one  record less. 
These are  adjacent  in the sense of DP, and, as customary in DP, we assume that the record on which the two databases differ is the target of the adversary. 
The observables space $\Ycal$ is the set of the $5$-tuples produces by the noisy counts of the $5$ labels. 
$\Wcal$ is set to be the same as $\Xcal$.

We assume that the adversary  is  especially interested in finding out whether the patient has a serious condition. 
The function $g$ reflects this preference by assigning higher value to higher labels. Specifically, we set $\Wcal = \Xcal$ and                                                                                                                                  
\begin{equation}
    g(w, x) = \begin{cases} 
        0, & \mbox{if } w\neq x \\ 1, & \mbox{if } w = x \land x\in\{0,1,2\}\\ 2, & \mbox{if } w = x \land x\in\{3,4\}, 
    \end{cases}
\end{equation}

For the estimation, we consider  10K, 30K and 50K samples for the training sets, and  50K samples for the validation set. %
\begin{figure}[!htb]
	\centering
	\makebox[\linewidth][c]{\includegraphics[scale=.34]{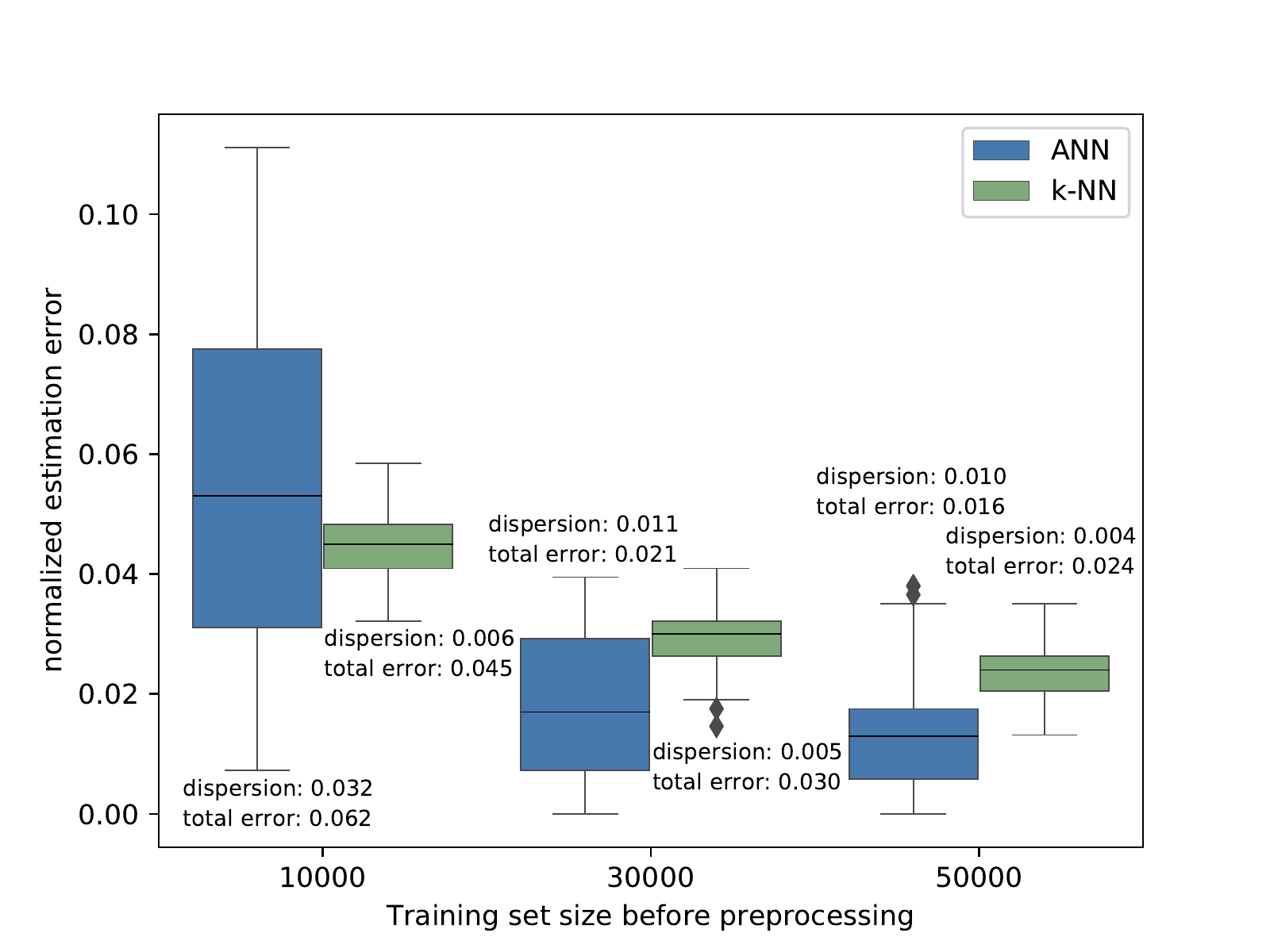}}
	\caption{Differential privacy scenario, k-NN and ANN estimation with data pre-processing}
	\label{fig:dp_knn_ann}
\end{figure}
\begin{figure}[!htb]
	\centering
	\makebox[\linewidth][c]{\includegraphics[scale=.34]{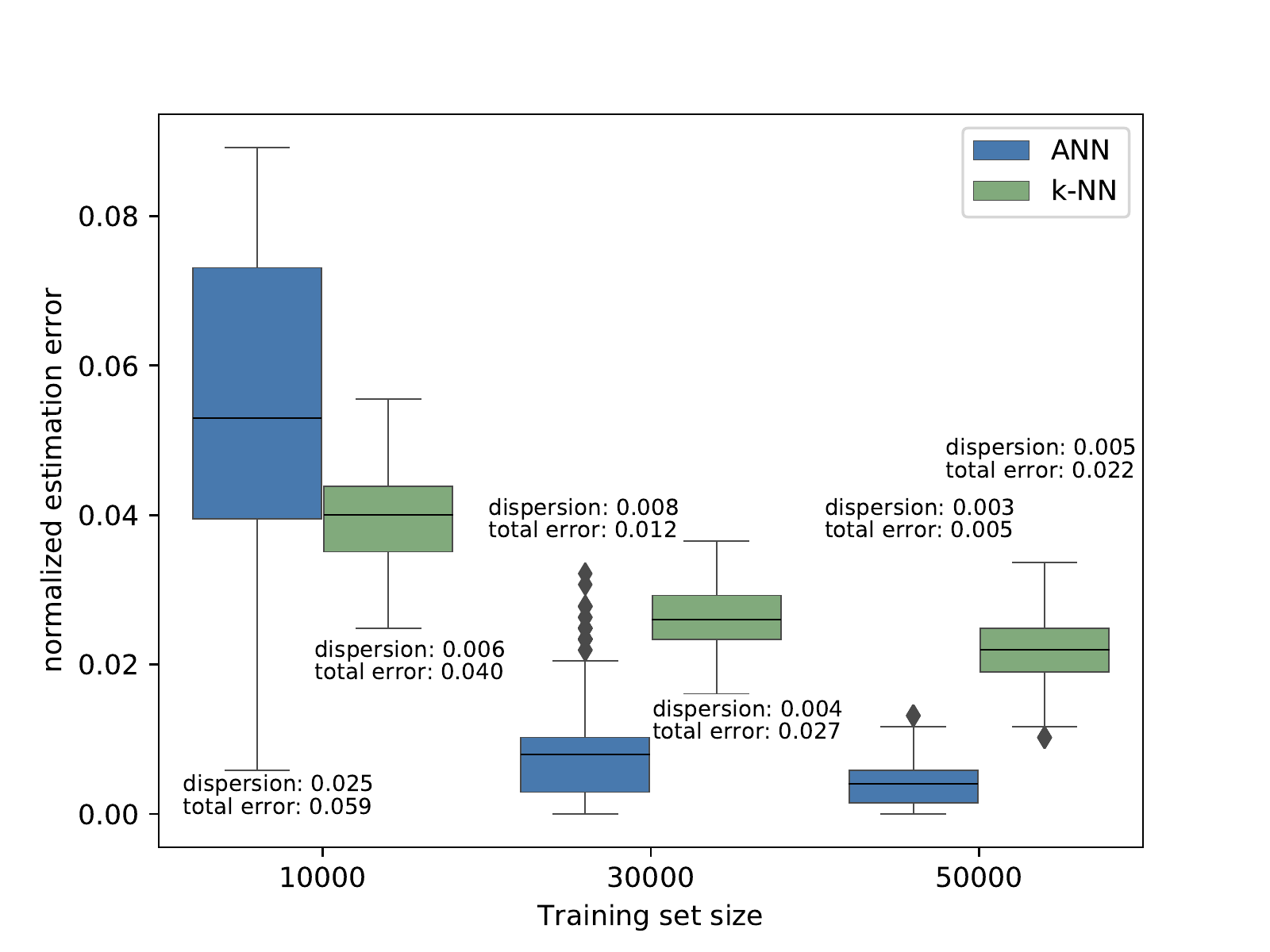}}
	\caption{Differential privacy scenario, k-NN and ANN estimation with channel pre-processing}
	\label{fig:dp_knn_ann_bis}
\end{figure}
For the experiments with k-NN we choose the Manhattan distance, which  is typical for DP\footnote{The Manhattan distance on histograms corresponds to the total variation distance on the 
distributions resulting from the normalization of these histograms.}. 
In the case of data pre-processing the size of the transformed training set  \trainingsetdtpreprocsub  is about $1.2$ times  the original     \trainingsetsub. 
The performances of the data and channel  pre-processing  are shown in Figures \ref{fig:dp_knn_ann} and \autoref{fig:dp_knn_ann_bis} respectively.
Surprisingly, in this case the data pre-processing method outperforms the channel pre-processing one, although only slightly. 
Additional plots, including the results for the frequentist approach, can be found in Appendix~\ref{extra_plots}.

\subsection{Experiment 4: password checker}
In this experiment we consider a password checker, namely a program that tests whether 
a given string corresponds to the password stored in the system. 
We assume that string and password are sequences of $128$ bits, an that the program is ``leaky'', in the sense that 
it checks the two sequences bit by bit and it stops checking as soon as it finds a mismatch, reporting failure. 
It is well known that this opens the way to a timing attack (a kind of side-channel attack), 
so we assume that the system tries to mitigate the threat  by adding some random delay, 
sampled from a Laplace distribution and then bucketing the reported time in 128 bins corresponding to the positions in the sequence
(or equivalently, by sampling the delay from a Geometric distribution, cfr. \autoref{geom_distr_multiguess}).
Hence the channel $C$ is a $2^{128} \times 128$ stochastic matrix. 

The typical attacker is an interactive one, which figures out larger and larger prefixes of the password by testing each bit at a time.  
We assume that the attacker has already figured out the first $6$ bits of the sequence and it is trying to figure out the $7$-th.  
Thus the prior $\pi$ is distributed (uniformly, we assume) only on the sequences formed by the known 
$6$-bits prefix and all the possible remaining $122$ bits, while the 
$g$ function assigns $1$ to the sequences whose $7$-th bit agrees with the stored password, and $0$ otherwise. 
Thus $g$ is a \emph{partition gain function}~\cite{Alvim:12:CSF}, and its particularity is that for such  kind of  functions   
data pre-processing and   channel pre-processing  coincide. 
This is because $g(w,x)$ is either $0$ or $1$, so  in both cases
we generate exactly one pair $(w,y)$ for each pair $(x,y)$ for which $g(w,x)=1$. 
Note that  in this case the data pre-processing transformation does not 
increase the training set, and the channel pre-processing 
transformation does not  introduce any additional noise. 
The $RC$ matrix (cfr. Section~\ref{data_preproc}) is a $2\times 128$ stochastic matrix. 
\begin{figure}[!htb]
	\centering
	\makebox[\linewidth][c]{\includegraphics[scale=.34]{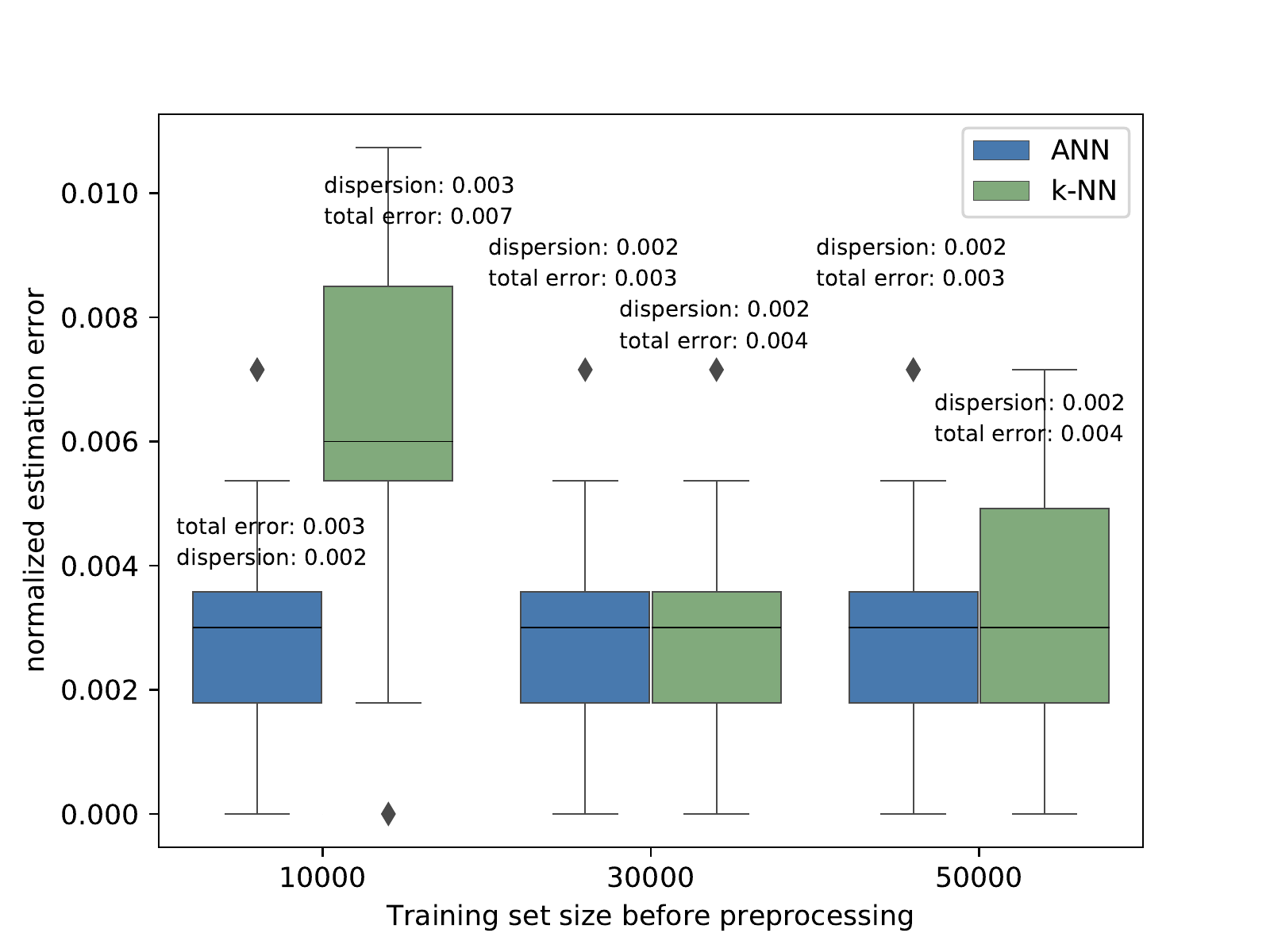}}
	\caption{Password checker scenario, k-NN and ANN estimation with data and channel pre-processing}
	\label{fig:psw_knn_ann}
\end{figure}
The experiments are done with  training sets of 10K, 30K and 50K samples.
The results are reported in Figure~\ref{fig:psw_knn_ann}. 
We note that the estimation error is quite small, especially in the ANN case. 
This is because the learning problem is particularly simple since, by considering the $g$-leakage and the 
preprocessing, we have managed to reduce the problem to learning a function of type $\Ycal \rightarrow \Wcal$,  
rather than $\Ycal \rightarrow \Xcal$, and there is a huge difference in size between 
$\Wcal$ and $\Xcal$ (the first is $2$  and the latter is $2^{128}$). 
Also the frequentist approach does quite well (cfr. Appendix~\ref{extra_plots}) , and this is because  $\Ycal$ is small. 
With a finer bucketing (on top of the Laplace delay), or no bucketing at all, 
we expect that the difference between the accuracy of the frequentist and of the ML estimation would be much larger. 

\rev{
\subsection{Discussion}
In almost all the experiments, our method gives much better result than the frequentist approach 
(see \autoref{fig:mult_guess_est_freq_knn_ann} and the other plots in the appendix \autoref{extra_plots}). 
The exception of the second experiment  can be explained  by the fact that 
the observable space is not very large, which is a scenario where the frequentist approach can be successful because the available data is enough to estimate the real distribution.  
In general, with the frequentist approach there is no real learning, therefore, if $|\Ycal|$ is large and the training set contains few samples, we cannot make a good guess with the observables never seen before \cite{Cherubin:19:SP}. 
In ML, on the contrary, we can still make an informed guess, as ML models are able to generalize from samples, especially when the Bayes error is small. 
}

\rev{
We observe that  ANN outperforms the k-NN in all experiments. This is because usually ANN models are better at generalizing, and hence provide better classifiers. In particular, k-NN are not very good when the distributions are not smooth with respect to the metric with respect to which the neighbor relation is evaluated~\cite{Cherubin:19:SP}.
}
\rev{
The data pre-processing method gives better results, than the channel pre-processing in all experiments except  the third one (DP), in which the difference is very small. The main advantage of using the channel pre-processing method is when the gain function is such that the data pre-processing would generate a set too large, as explained in Section~\ref{sec:ComparisonDataChannel}.
}

\rev{
The experiments show that our method is not too sensitive to the size of $|\Ycal|$. On the other hand, the size of $|\Xcal|$ is important, because the ML classifiers are in general more precise (approximate better the ideal Bayes classifier) when the number of classes are small. This affects the estimation error of both the  Bayes vulnerability and the \gv. However, for the latter there is also the additional problem of the magnification  due to the $g$. To better understand this point, consider a modification of the  last experiment, and assume that 
the password checker is not leaky, i.e., the observables are only $\mathit{fail}$ or $\mathit{success}$. A pair $(x, \emph{success})$ would have a negligible 
probability of appearing in the training set, hence our method, most likely, would estimate the vulnerability to be $0$. This is fine if we are trying to estimate the Bayes vulnerability, which is also negligible. But  the \gv\ may not be negligible, in particular if we consider a $g$ that gives an enormous gain for the success case. If we can ensure that all the pairs $(x, y)$  are represented in the training set in proportion to their probability in $P_{XY}$, then the above "magnification" in \gv\ is not a problem,  because our method will ensure the also  the pairs $(w, y)$ would be magnified (with respect to the the pairs $(w, y)$) in the same proportion. 
}


\section{Conclusion and future work}

We have proposed an approach to estimate the \gv\ of a system under the black-box assumption, using machine learning. 
The basic idea is to reduce the problem to learn the Bayes classifier on a set of pre-processed training data, and we have 
proposed two techniques for this transformation, with different advantages and disadvantages. 
We have then evaluated our approach on various scenarios, showing favorable results. 
We have compared our approach to the frequentist one, showing that the performances are similar on small 
observable domains, while ours performs better  on large ones. This is in line with what already observed in \cite{Cherubin:19:SP} 
for the estimation of the Bayes error. 

As future work, we plan to test our framework on more real-life scenarios such as the web fingerprinting attacks \cite{Cherubin:17:POPETS,Cherubin:17:POPETS1} 
and the AES cryptographic algorithm \cite{Cherisey:19:CHES}. We also would like to consider the more general case, often considered in Information-flow security,  
of channels that have both ``high'' and ``low'' inputs, where the first are the secrets and the latter are data visible to, or even controlled by, the adversary.  
Finally, a more ambitious goal is to use our approach to 
minimize the $g$-vulnerability of complex systems, 
using a GAN based approach, along the lines of \cite{Romanelli:20:CSF}. 


\begin{acks}
\rev{
This research was supported by DATAIA ``Programme d'Investis\-sement d'Avenir'' (ANR-17-CONV-0003). It was also supported by the ANR project REPAS, and by the Inria/DRI project LOGIS. The work of Catuscia Palamidessi was supported by the project HYPATIA, funded by  the European Research Council (ERC) under the European Union's Horizon 2020 research and innovation program, grant agreement n. 835294. This project has received funding from the European Union’s Horizon 2020 research and innovation programme under the Marie Skłodowska-Curie grant agreement No 792464.}
\end{acks}

\bibliographystyle{ACM-Reference-Format}
\bibliography{files/biblio}


\appendix
\appendix
\section{Auxiliary Results}
\label{AppA}
\begin{proposition}[\rev{Bernstein's inequality \cite{Boucheron:13:CI}}]\label{lemma-Bernstein}
\rev{Let\linebreak ${Z_1,\dots,Z_n\sim Z}$ be i.i.d. random variables such that $Z\in[ a,b]$ almost surely and let $S_n=\frac{1}{n}\sum_{i=1}^n Z_i-\mathbb{E}\left[Z\right]$ and 
$v= \text{Var}(Z)$ the variance of $Z$. Then, for any $\varepsilon>0$, we have: 
\begin{align}
    \label{bernstein_bound}
    \mathbb{P}\left[ S_n \geq \varepsilon\right]\leq \exp\left(-\frac{n\, \varepsilon^2}{2\,v+\nicefrac{2\left(b-a\right)\varepsilon}{3}}\right).
\end{align}
Compared to the Hoeffding's inequality, it is easy to check that, for regimes where $\varepsilon$ is small, Bernstein's inequality offers tighter bounds for $v \ll (b-a)^2$.}
\end{proposition}
\mrev{\begin{lemma}\label{lemma-inequality} 
Let $\sigma^2=\text{Var}(Z)$ and let $Z$ be a real-valued random variable such that for all $0<t\leq \sigma^2$, 
\begin{equation}
\mathbb{P}(Z\geq t) \leq 2q \exp \left( -\frac{t^2}{r^2}\right).
\end{equation} 
Then,
\begin{equation}
\int_{0}^{\sigma^2} \mathbb{P}(Z\geq t) dt \leq qr\sqrt{\pi}\erf\left(\frac{\sigma^2}{r}\right),
\end{equation} 
where, for large $x$, 
\begin{align}
\label{taylor}
\erf(x) \approx 1 - \frac{\exp(-x^2)}{x\sqrt{\pi}}+ \mathcal{O}\left(x^{-1}\exp(-x^2)\right).
\end{align}
\end{lemma}
\begin{proof}
\begin{align*}
 \int_{0}^{\sigma^2} \mathbb{P}(Z\geq t) dt &\leq  \int_{ 0}^{\sigma^2} 2q \exp \left( -\frac{t^2}{r^2}\right) dt\\
& = qr\sqrt{\pi}\erf\left(\frac{\sigma^2}{r}\right),
\end{align*}
and~\cref{taylor} follows from the Taylor's expansion of the $\erf$ function.
\end{proof}
}


\section{Proofs for the statistical bounds}\label{AppB}

\rev{The following lemma   is a simple adaption of the uniform deviations of relative frequencies from probabilities theorems 
in~\cite{Devroye:96:PTPR}.}  

\begin{lemma}\label{lemmaworstcase}
\rev{The following inequality holds: 
\begin{align}
\label{upperbounds1}
V_g- V(f^\star_m)  & \leq  2 \max\limits_{f \in\Hcal} \big |  \widehat{V}_m(f) - V(f)\big |.
\end{align}}
\end{lemma}
\begin{proof}
\begin{align}
\nonumber
 V_g  &- V(f^\star_m)   =  V(f^\star)  - \widehat{V}_m(f^\star_m)   +   \widehat{V}_m(f^\star_m)   - V(f^\star_m) \label{eq-missing-appedix1}\\
& \leq  V(f^\star)  - \widehat{V}_m(f^\star_m)  +   \big|\widehat{V}_m(f^\star_m)   - V(f^\star_m)\big|\\
& \leq  V(f^\star)  - \widehat{V}_m(f^\star)  +   \big|\widehat{V}_m(f^\star_m)   - V(f^\star_m)\big|\\
& \leq  \big| V(f^\star)  - \widehat{V}_m(f^\star)  \big| +   \big|\widehat{V}_m(f^\star_m)   - V(f^\star_m)\big|\\
& \leq \max\limits_{f\in\Hcal} \big|\widehat{V}_m(f) -V(f)\big| + \max\limits_{f\in\Hcal}  \big|\widehat{V}_m(f)   - V(f)\big|\\ 
& \leq 2\max\limits_{f\in\Hcal} \big|\widehat{V}_m(f) -V(f)\big|.\label{eq-missing-appedix2}
\end{align}
\end{proof}

\subsection{Proof of~\autoref{mainproposition} }
\mainpropositionrestatable*
\label{AppC}
\begin{proof}
\rev{
We first prove   \eqref{eq-Pro1}. 
We have:
\begin{align}
& \mathbb{P}\left(  \big |  \widehat{V}_n(f^\star_m) - V(f^\star_m)\big |  \geq \varepsilon  \right)  \nonumber\\
& = \E_{\mathcal{D}_m\sim P_{XY}^m}  \mathbb{P}\left(   \big |  \widehat{V}_n(f^\star_m) - V(f^\star_m)\big | \geq \varepsilon \,\, | \,\mathcal{D}_m \right) \label{eq-util} \\
& \leq  2\E_{\mathcal{D}_m\sim P_{XY}^m}     \exp \left(-\frac{n\,\varepsilon^2}{2\sigma^2_{f^\star_m}+\nicefrac{2\left(b-a\right)\varepsilon}{3}}\right), \label{eq-propo1-last}
\end{align}
where \eqref{eq-util}  follows from the definition of $\mathbb{P}$:
we consider the expectation of the probability over all training sets $\mathcal{D}_m$ sampled from $P_{XY}$, and then for each $\mathcal{D}_m$ we take the probability over all possible validation sets $\mathcal{T}_n$ sampled again from $P_{XY}$. Consider the series of $(X_i,Y_i)$'s  sampled from $P_{XY}$ that constitute the validation set, and define the
random variables $Z_i= g(f^\star_m(Y_i), X_i)$. 
Note that $\mathcal{D}_m$ determines  $f^\star_m$, hence, for a given $\mathcal{D}_m$ 
the $Z_i$ are i.i.d.. 
The inequality  \eqref{eq-propo1-last} then follows by applying~\autoref{lemma-Bernstein}. Indeed, since the $Z_i$'s  are i.i.d.,  
they all have  the same expectation and the same variance, hence $S_n=\frac{1}{n}\sum_{i=1}^n Z_i -  \mathbb{E}\left[Z \right]  = \widehat{V}_n(f^\star_m)- V(f^\star_m)$, and  $v=\textit{Var}(Z) = \sigma^2_{f^\star_m} $. 
The factor $2$ in front of the $\exp$ is because we consider the absolute value of $S_n$.
}

\rev{
As for the bound \eqref{eq-Pro2}, we have: 
\begin{align}
 \mathbb{P}\left(   V_g - {V}(f^\star_m)   \geq \varepsilon  \right) 
&\leq \mathbb{P}\left(  \max\limits_{f\in\mathcal{H}}\big|  \widehat{V}_m(f) -  {V}(f)   \big|  \geq \varepsilon /2 \right)\label{eq-Propo1-2}\\
&= \mathbb{P}\left(   \bigcup\limits_{f\in \mathcal{H}} \big\{ \big |  \widehat{V}_m(f) - V(f)\big | \geq \varepsilon/2 \big\}  \right) \label{eq-Propo1-2-2}\\
&\leq \sum\limits_{f\in \mathcal{H}} \mathbb{P}\left(   \big |  \widehat{V}_m(f) - V(f)\big | \geq \varepsilon/2  \right) \label{prop1.1-aux} \\
&  \leq \sum\limits_{f\in \mathcal{H}}  2   \exp \left(-\frac{m\,\varepsilon^2}{8\sigma^2_{f}+\nicefrac{4\left(b-a\right)\varepsilon}{3}}\right), \label{eq-Propo1-3}
\end{align}
where \eqref{eq-Propo1-2} follows from~\autoref{lemmaworstcase}, steps \eqref{eq-Propo1-2-2} and \eqref{prop1.1-aux}  are standard, and  \eqref{eq-Propo1-3} follows from the same reasoning that we have applied to prove \eqref{eq-Pro1}. Here we do not take the expectation on the training sets because in each term of the summation the $f$ is fixed. 
}
\end{proof}


\subsection{Proof of~\autoref{thm_estimate_convergence}}
\thmestimateconvergence*
\label{App_thm_1}
\begin{proof}
Observe that 
\begin{align*}
\E \big|V_g-  \widehat{V}_n(f^\star_m)  \big|&  = \E \big|V_g - {V}(f^\star_m)  + {V}(f^\star_m) - \widehat{V}_n(f^\star_m)  \big| \nonumber \\
& \leq   V_g-   \E \big[ {V}(f^\star_m)\big]  \big|+ \E \big|{V}(f^\star_m)-  \widehat{V}_n(f^\star_m)  \big|, 
\end{align*}
which follows from the triangular inequality. 

\rev{
First, let us call $\sigma^2$ the worst case variance defined above which, according to Popoviciu's inequality is upper-bounded by $\nicefrac{(b-a)^2}{4}$. Second, let us consider that one main advantage of deriving bounds from Bernstein's inequality is that it allows allow the upper-bound in~\cref{bernstein_bound} grows as $\exp(-n\varepsilon)$ instead of $\exp(-n\varepsilon^2)$ if $v\leq\varepsilon$. Moreover~\cref{eq-propo1-last} is upper-bounded by $ 2   \exp \left(-\nicefrac{n\,\varepsilon^2}{2\sigma^2+\nicefrac{2\left(b-a\right)\varepsilon}{3}}\right)$. This said we are going to consider two cases:\\
\\i) $\boldsymbol{\sigma^2\leq\varepsilon}$:
	\begin{align}
	&\mathbb{E}|V(f_m^\star) - \widehat{V}_n(f_m^\star)|\leq \int_{\sigma^2}^{\infty}2\exp\left(-\frac{n\varepsilon^2}{2\sigma^2+\nicefrac{2(b-a)\varepsilon}{3}}\right)d\varepsilon\\
	 & \leq \int_{\sigma^2}^{\infty}2\exp\left(-\frac{n\varepsilon}{2+\nicefrac{2(b-a)}{3}}\right)d\varepsilon\\
	 & = \frac{4(1+\nicefrac{(b-a)}{3})}{n}\exp\left(-\frac{n\sigma^2}{2(1+\nicefrac{(b-a)}{3})}\right),
	\end{align}
	and then,
	\mrev{
	\begin{align}
	&  V_g-   \E \big[ {V}(f^\star_m)\big] \leq\int_{\sigma^2}^{\infty}\sum\limits_{f\in \mathcal{H}}2   \exp \left(-\frac{m\,\varepsilon^2}{8\sigma^2_{f}+\nicefrac{4\left(b-a\right)\varepsilon}{3}}\right)d\varepsilon\\
	&\leq\sum\limits_{f\in \mathcal{H}}\int_{\sigma^2}^{\infty}2   \exp \left(-\frac{m\,\varepsilon^2}{8\sigma^2+\nicefrac{4\left(b-a\right)\varepsilon}{3}}\right)d\varepsilon\\
	&\leq\sum\limits_{f\in \mathcal{H}}\int_{\sigma^2}^{\infty}2   \exp \left(-\frac{m\,\varepsilon}{8+\nicefrac{4\left(b-a\right)}{3}}\right)d\varepsilon\\
	&\leq|\mathcal{H}|\frac{8(2+\nicefrac{(b-a)}{3})}{m}\exp\left(-\frac{3m\sigma^2}{4(6+(b-a))}\right),
	\end{align}
ii) $\boldsymbol{\sigma^2 >\varepsilon}$:
	\begin{align}
		&\mathbb{E}|V(f_m^\star) - \widehat{V}_n(f_m^\star)|\leq \int_{0}^{\sigma^2}2\exp\left(-\frac{n\varepsilon^2}{2\sigma^2+\nicefrac{2(b-a)\varepsilon}{3}}\right)d\varepsilon\\ 
		&\leq \int_{0}^{\sigma^2}2\exp\left(-\frac{n\varepsilon^2}{2\sigma^2+\nicefrac{2(b-a)\sigma^2}{3}}\right)d\varepsilon\\ 
		&= \sqrt{\frac{2\sigma^2(1+\nicefrac{(b-a)}{3})}{n}} \sqrt{\pi} \erf\left(\frac{\sigma^2}{\sqrt{\frac{2\sigma^2(1+\nicefrac{(b-a)}{3})}{n}}}\right),
	\end{align}
	considering $r = \sqrt{\frac{2\sigma^2(1+\nicefrac{(b-a)}{3})}{n}}$, $q=1$ and applying~\cref{lemma-inequality}.
	And finally,\\
\begin{align}
		 &V_g-   \E \big[ {V}(f^\star_m)\big]\leq \sum\limits_{f\in \mathcal{H}}\int_{0}^{\sigma^2}2   \exp \left(-\frac{m\,\varepsilon^2}{8\sigma^2_{f}+\nicefrac{4\left(b-a\right)\varepsilon}{3}}\right)d\varepsilon\\
		 &\leq 2|\mathcal{H}|\int_{0}^{\sigma^2}\exp \left(-\frac{m\,\varepsilon^2}{8\sigma^2+\nicefrac{4\left(b-a\right)\varepsilon}{3}}\right)d\varepsilon\\
		 &\leq 2|\mathcal{H}|\int_{0}^{\sigma^2}\exp \left(-\frac{m\,\varepsilon^2}{8\sigma^2+\nicefrac{4\left(b-a\right)\sigma^2}{3}}\right)d\varepsilon\\
		 &= |\mathcal{H}|\sqrt\frac{8\sigma^2+\nicefrac{4\sigma^2 (b-a)}{3}}{m}\sqrt{\pi}\erf\left(\frac{\sigma^2}{\sqrt\frac{8\sigma^2+\nicefrac{4\sigma^2 (b-a)}{3}}{m}}\right),
	\end{align}
	according to~\cref{lemma-inequality} with $r=\sqrt\frac{8\sigma^2+\nicefrac{4\sigma^2 (b-a)}{3}}{m}$ and $q=|\Hcal|$.}
}

 \end{proof}
 

\subsection{Proof of~\autoref{training_samples_upperbound}}
\label{AppD}
\trainingsamplesupperboudrestatable*
\begin{proof}
We first notice that 

\begin{align}
& \mathbb{P}\left( |  V_g - \widehat{V}_n(f^\star_m) |   \geq \varepsilon  \right) \nonumber\\
&\leq  \mathbb{P}\left(   V_g - {V}(f^\star_m)    \geq \varepsilon  \right)  +  \mathbb{P}\left( |  {V}(f^\star_m) - \widehat{V}_n(f^\star_m) |   \geq \varepsilon  \right),
\end{align}
\rev{and thus from~\eqref{eq-Pro1}, \eqref{eq-Pro2} in Proposition \ref{mainproposition}, we have:
\begin{align}
&\mathbb{P}\left( |  V_g - \widehat{V}_n(f^\star_m) |   \geq \varepsilon  \right) \nonumber\\
&\leq 2|\mathcal{H}|\exp\left(-\frac{m\,\varepsilon^2}{8\,\sigma^2+\nicefrac{4\,\left(b-a\right)\varepsilon}{3}}\right)  \label{boundsigma_m}  + 2\exp  \left(-\frac{n\,\varepsilon^2}{2\sigma^2+\nicefrac{2\,\left(b-a\right)\varepsilon}{3}}\right).
\end{align}
}
\rev{Let us require:
\begin{align}
2\, |\mathcal{H}| \,\exp  \left(-\frac{m\,\varepsilon^2}{8\,\sigma^2+\nicefrac{4\, \left(b-a\right)\varepsilon}{3}}\right) & \leq (\delta - \Delta),\\
2\, \exp  \left(-\frac{n\,\varepsilon^2}{2\,\sigma^2+\nicefrac{2\,\left(b-a\right)\varepsilon}{3}}\right) & \leq \Delta,
\end{align}
}
which satisfies the desired condition: 
\begin{align}
 \mathbb{P}\left( |  V_g - \widehat{V}_n(f^\star_m) |   \geq \varepsilon  \right) &\leq \delta,
 \end{align}
 for any $0<\Delta<\delta$. Finally, from the previous inequality we can derive lower bounds on $n$ and $m$:
\rev{
\begin{align}
 m &\geq \frac{8\,\sigma^2+\nicefrac{4\, \left(b-a\right)\varepsilon}{3}}{\varepsilon^2}\ln\left(\frac{2\, |\Hcal|}{\delta-\Delta}\right),\\
n &\geq \frac{2\,\sigma^2+\nicefrac{2\,\left(b-a\right)\varepsilon}{3}}{ \varepsilon^2}\ln\left(\frac{2}{\Delta}\right),
\label{samples_train_bound}
\end{align}
}
which  by definition of sample complexity shows the corollary. 
\end{proof}

\section{Pre-processing}
\subsection{Data pre-processing}
\label{App_data_preproc}
\datapreprocthm*
\begin{proof}
\rev{\begin{align}
&V_{\GId}(\xi, E) = \sum\limits_{y }\max\limits_{w}\sum\limits_{w^\prime}\xi{w^\prime}\cdot E_{w^\prime y}\cdot g_{\mathrm{id}}(w, {w^\prime})=\\
&\sum\limits_{y}\max\limits_{w} (\xi_{w}E_{w y})=\sum\limits_{y }\max\limits_{w} P_{WY}(w, y)=\sum\limits_{y}\max\limits_{w} \frac{U(w,y)}{\alpha}\\
&=\frac{1}{\alpha} \cdot  \sum\limits_{y}\max\limits_{w} {\sum_x \pi_x \cdot C_{xy} \cdot g(w,x)}=\nicefrac{1}{\alpha} \cdot V_{g}(\pi, C) 
\end{align}}
\end{proof}

\subsection{Channel pre-processing}\label{app:channel-preprocessing}
\channelpreproc*
\label{App_channel_preproc}
\begin{proof}
    In this proof we use a notation that highlights the structure of the preprocessing. 
    We will denote by  $G$ be the matrix form of  $g$, i.e., $G_{wx}= g(w,x)$, and by 
    $\Diag{\pi}$  the square matrix with $\pi$ in its diagonal
    and $0$ elsewhere.
	We have that  $\beta = \|G\Diag{\pi}\|_1 = \sum_{w,x} G_{wx} \pi_x$, which is
	strictly positive because of the assumptions on $g$ and $\pi$.
	Furthermore, we have
	\[
		\tau^T
			\Wide= 
			\beta^{-1} G\Diag\pi\OneVect
		~,
		\qquad
			R
			\Wide=
			\beta^{-1} (\Diag{\tau})^{-1}G\Diag\pi
		~,
	\]
	where $\OneVect$ is the vector of $1$s and $\tau^T$ represents the transposition of vector $\tau$. Note that $(\Diag{\tau})^{-1}$ is
	a diagonal matrix with entries $\tau_w^{-1}$ in its diagonal.
	If $\tau_w = 0$ then the row $R_{w, \cdot}$ is not properly defined; but
	its choice does not affect
	$V_{\GId}(\tau, RC)$ since the corresponding prior is $0$; so
	we can choose $R_{w,\cdot}$ arbitrarily (or equivalently
	remove the action $w$, it can never be optimal since it gives $0$ gain).
	It is easy to check that $\tau$ is a proper distribution
	and $R$ is a proper channel:
	\begin{align*}
		\textstyle
		\sum_w \tau_w
		&\Wide=
		\textstyle
		\sum_w \beta^{-1} \sum_x G_{wx} \pi_x
		&&\Wide=
		\beta^{-1} \beta
		&\Wide=
		1
		~,
		\\
		\textstyle
		\sum_x R_{w,x}
		&\Wide=
		\textstyle
		\sum_x \frac{1}{\tau_w} \beta^{-1} G_{wx} \pi_x
		&&\Wide=
		\frac{ \tau_w  }{ \tau_w }
		&\Wide=
		1
		~.
	\end{align*}
	Moreover, it holds that:
	\[
		\beta \Diag\tau R
		\Wide= \beta \Diag{\tau} \beta^{-1} \Diag{\tau}^{-1}G\Diag\pi
		\Wide= G\Diag\pi
		~.
	\]
	The main result follows from the trace-based formulation of posterior $g$-vulnerability
    \cite{Alvim:12:CSF},
	since for any channel $C$ and strategy $S$, the above equation
	directly yields
	\begin{align*}
		V_g(\pi, C)
		\Wide=
		\max_{S} \Trace{G\Diag\pi CS}
		&\Wide=
		 \beta\cdot\max_{S}\Trace{ \Diag\tau RCS}
        \\
		&\Wide=
		\beta\cdot V_{\GId}(\tau,RC)
		~,
	\end{align*}
	where $\Trace{\cdot}$ is the matrix trace.
\end{proof}

\subsection{Data pre-processing when $g$ is not integer}
\label{App_data_preproc_no_real}
Approximating $g$ so that it only takes values $\in \mathbb{Q}_{\geq 0}$ allows us to represent each gain as a quotient of two integers, namely 
\begin{equation}
\nicefrac{\text{Numerator}(G_{w,x})}{\text{Denominator }(G_{w,x})}. \nonumber
\end{equation}
Let us also define
\begin{equation}
K \defsym\lcm_{wx}(\text{Denominator}(G_{w,x})), 
\end{equation}
where $\lcm(\cdot)$ is the least common multiple. Multiplying $G$ by $K$ gives the integer version of the gain matrix that can replace the original one. It is clear that the calculation of the least common multiplier, as well as the increase in the amount of data produced during the dataset building using a gain matrix forced to be integer, might constitute a relevant computational burden. 
\section{ANN models}
\label{AppG}
We list here the specifics for the ANNs models used in the experiments. All the models are simple feed-forward networks whose layers are fully connected. The activation functions for the hidden neurons are rectifier linear functions, while the output layer has softmax activation function. 

The loss function minimized during the training is the cross entropy, a popular choice in classification problems. The remapping $\Ycal\rightarrow\Wcal$ can be in fact considered as a classification problem such that, given an observable, a model learns to make the best guess.

For each experiments, the models have been tuned by cross-validating them using one randomly chosen training sets among the available ones choosing among the largest in terms of samples. The specs are listed experiment by experiment in~\cref{tab:ann-param}.
\begin{table*}[!h]
\centering
\resizebox{\textwidth}{!}{%
\begin{tabular}{llccccc}
\cline{3-7} 
 &
   &
  \multicolumn{5}{|c|}{\textbf{Hyper-parameters}} \\ \cline{1-7} 
\multicolumn{1}{|l|}{\textbf{Experiment}} &
  \multicolumn{1}{l|}{\textbf{Pre-processing}} &
  \multicolumn{1}{c|}{learning rate} &
  \multicolumn{1}{c|}{hidden layers} &
  \multicolumn{1}{c|}{epochs} &
  \multicolumn{1}{c|}{hidden units per layer} &
  \multicolumn{1}{c|}{batch size} \\ \cline{1-7} 
\multicolumn{1}{|l|}{\multirow{2}{*}{Multiple guesses}} &
  \multicolumn{1}{l|}{Data} &
  \multicolumn{1}{c|}{$10^{-3}$} &
  \multicolumn{1}{c|}{$3$} &
  \multicolumn{1}{c|}{700} &
  \multicolumn{1}{c|}{$[100,100,100]$} &
  \multicolumn{1}{c|}{1000} \\ \cline{2-7} 
\multicolumn{1}{|c|}{} &
  \multicolumn{1}{l|}{Channel} &
  \multicolumn{1}{c|}{$10^{-3}$} &
  \multicolumn{1}{c|}{$3$} &
  \multicolumn{1}{c|}{500} &
  \multicolumn{1}{c|}{$[100,100,100]$} &
  \multicolumn{1}{c|}{1000} \\ \cline{1-7} 
\multicolumn{1}{|l|}{\multirow{2}{*}{Location Priv.}} &
  \multicolumn{1}{l|}{Data} &
  \multicolumn{1}{c|}{$10^{-3}$} &
  \multicolumn{1}{c|}{$3$} &
  \multicolumn{1}{c|}{1000} &
  \multicolumn{1}{c|}{$[500,500,500]$} &
  \multicolumn{1}{c|}{$200,500,1000$} \\ \cline{2-7} 
\multicolumn{1}{|c|}{} &
  \multicolumn{1}{l|}{Channel} &
  \multicolumn{1}{c|}{$10^{-3}$} &
  \multicolumn{1}{c|}{$3$} &
  \multicolumn{1}{c|}{$200, 500, 1000$} &
  \multicolumn{1}{c|}{$[500,500,500]$} &
  \multicolumn{1}{c|}{$20, 200, 500$} \\ \cline{1-7} 
\multicolumn{1}{|l|}{\multirow{2}{*}{Diff. Priv.}} &
  \multicolumn{1}{l|}{Data} &
  \multicolumn{1}{c|}{$10^{-3}$} &
  \multicolumn{1}{c|}{$3$} &
  \multicolumn{1}{c|}{500} &
  \multicolumn{1}{c|}{$[100,100,100]$} &
  \multicolumn{1}{c|}{$200$} \\ \cline{2-7} 
\multicolumn{1}{|c|}{} &
  \multicolumn{1}{l|}{Channel} &
  \multicolumn{1}{c|}{$10^{-3}$} &
  \multicolumn{1}{c|}{$3$} &
  \multicolumn{1}{c|}{500} &
  \multicolumn{1}{c|}{$[100,100,100]$} &
  \multicolumn{1}{c|}{$200$} \\ \cline{1-7} 
  \multicolumn{1}{|l|}{\multirow{1}{*}{Psw SCA}} &
  \multicolumn{1}{l|}{-} &
  \multicolumn{1}{c|}{$10^{-3}$} &
  \multicolumn{1}{c|}{$3$} &
  \multicolumn{1}{c|}{700} &
  \multicolumn{1}{c|}{$[100,100,100]$} &
  \multicolumn{1}{c|}{$1000$} \\ \cline{1-7} 
\end{tabular}%
}
\caption[Hyper-parameters settings table for the experiments in~\cref{chp:leakage-estimation}]{Table with the hyper-parameters setting for each one of the experiments above. When multiple values are provided for the parameters of an experiment it is to be intended that each value corresponds to a specific size of the training set (sorted from the smallest to the largest number of samples).}
\label{tab:ann-param}
\end{table*}

\section{Frequentist approach description}
\label{freq_description}
In the frequentist approach the elements of the channel, namely the conditional probabilities ${P}_{Y|X}(y|x)$, are estimated directly in the following way: the empirical  prior probability of $x$, $\widehat{\pi}_x$, is computed by counting the number of occurrences of $x$ in the training set and dividing the result by the total number of elements. Analogously, the empirical joint probability $\widehat{P}_{XY}(x,y)$ is computed by counting the number of occurrences 
of the pair $(x,y)$  and dividing the result by the total number of elements in the set. 
The estimation $\widehat{C}_{xy}$ of ${C}_{xy}$ is then defined as
\begin{equation}
\widehat{C}_{xy} =\frac{\widehat{P}_{XY}(x,y)}{\widehat{\pi}(x)}.
\end{equation}

In order to have a fair comparison with our approach, which takes advantage of the fact that we have several training sets and validation sets at our disposal, 
while preserving at the same time the spirit of the frequentist approach, we proceed as follows:  
Let us consider a training set \trainingsetsub, that we will use to learn the best remapping $\Ycal \rightarrow \Wcal$, and a validation set \validationsetsub\ which is then used to actually estimate the \gv. We first compute $\widehat{\pi}$ using \trainingsetsub. For each $y$ in 
$\Ycal$ and for each $x\in \Xcal$, the empirical probability $\widehat{P}_{X|Y}$ is computed using \trainingsetsub\ as well. In particular, $\widehat{P}_{X|Y}(x|y)$ is given by the number of times $x$ appears in pair with $y$ divided by the number of occurrences of $y$. In case a certain $y$ is in \validationsetsub\ but not in \trainingsetsub, it is assigned the secret $x^\prime = \argmax_{x\in\Xcal}\widehat{\pi}$ so that $\widehat{P}_{X|Y}(x^\prime|y)=1$ and $\widehat{P}_{X|Y}(x)=0, \forall x\neq x^\prime$. It is now possible to find the best mapping for each $y$ defined as $w(y)=\argmax_{w\in\Wcal}\sum_{x\in\Xcal} \widehat{P}_{X|Y}(x|y) g(w,x)$. Now we compute the empirical joint distribution for each $(x,y)$ in \validationsetsub, namely $\widehat{Q}_{XY}$, as the number of occurrences of $(x,y)$ divided by the total number of samples in \validationsetsub. We now estimate the \gv\ on the validation samples according to:
\begin{equation}
	\widehat{V}_n=\sum\limits_{y\in\Ycal}\sum\limits_{x\in\Xcal}\widehat{Q}_{XY}(x,y)g(w(y), x).
\end{equation}
\\

\rev{
\section{ANN: Model selection and impact on the estimation}
\label{mod_sel}
In this section we are going to: 
\begin{itemize}
    \item briefly summarize the widely known background of the model selection problem from the machine learning standpoint;
    \item show, through a new set of experiments, how this problem affects the leakage estimation;
    \item propose a heuristic which can help the practitioner in the choice of the model on the same line as classical machine learning techniques.
\end{itemize}
The problem of model selection in machine learning is still an open one and, although a state-of-the art algorithm does not exists, heuristics can be proposed to lead the practitioner in the hard task of choosing a model over others.
First, let us underline that the choice of a specific model in the context of machine learning, and specifically neural networks (and deep learning), must go through the hyper-parameter optimization procedure. In fact, if nowadays neural nets and especially deep models represent the state-of-the-art solutions to most of the automatic decision problem, it is also true that, with respect to other simpler methods, they introduce the need for hyper-parameters optimization. Some techniques, such as grid and random search as well as Bayesian optimization have been suggested during the years, especially when not so many parameters need to be tuned. Two aspects must be considered:
\begin{enumerate}
    \item the hyper-parameter optimization relies on try and error strategy,
    \item the results are highly dependent on the data distribution and how well the samples represent such distribution. 
\end{enumerate}
In particular, if we consider the typical classification problem framework in neural nets (which we build on to create our framework) we expect the network to reproduce in output the distribution $P_{class|input}$ from the observed data. In this context, the practitioner should be careful to avoid two main problems which might affect the models, namely under-fitting and over-fitting. Both problems undermine the generalization capabilities of the models: the former occurs when the model is too simple to correctly represent the distribution; the latter occurs when the model is over-complicated, especially for the given amount of samples, and it fits the training data ``too well'' but this does not translates into good performances on other samples drawn from the same distribution.}

\rev{In order to understand how these problems impact our framework, we propose an analysis of the first experiment presented in the paper, considering different networks models and focusing on the different choice of number of hidden layers. In~\cref{fig:bp_03} we compare a model with no hidden layers, \emph{hl0}, and the three hidden layers model presented in the paper, \emph{hl3}. Using neural networks without hidden layers is not common. Indeed, theoretical results (cfr.~\cite{Cybenko:92:MCSS}) state that the simplest universal approximator can be modeled as a one hidden layered neural network. Although this holds in theory, in practice it is well known that this might require layers with too many neurons, and therefore, multiple hidden layers architectures have been gaining ground in real world applications. }

\rev{Therefore, we do not expect too much from the network with no hidden layers and, indeed, the results represented in~\cref{fig:bp_03_no_error} and~\cref{fig:bp_03_error} show that the estimation capabilities of this shallow model are very far from the performance obtained with the three layered model.
\begin{figure*}[!htb]
    \centering
    \begin{subfigure}{.5\textwidth}
        \centering
        \includegraphics[scale=0.31]{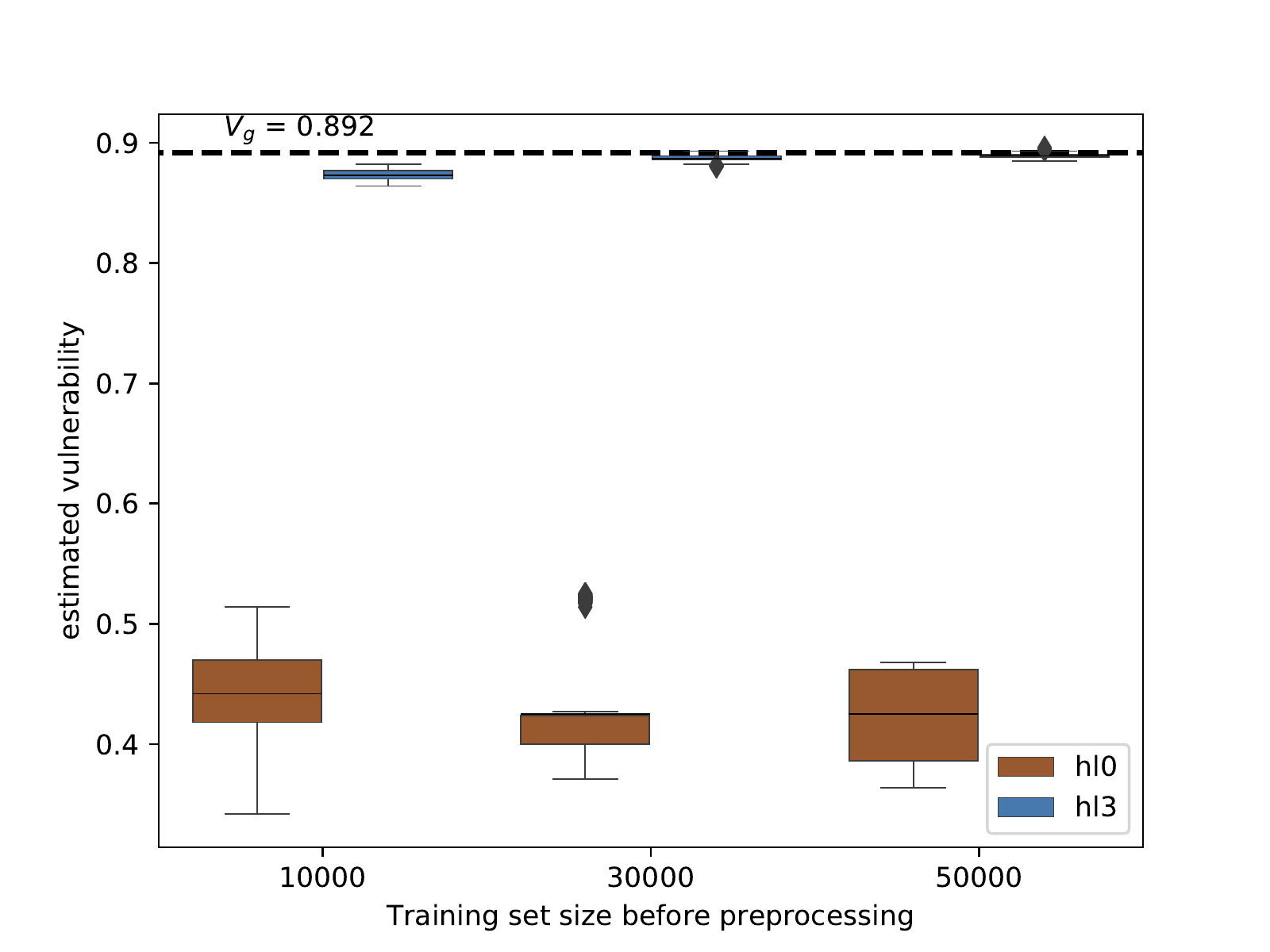}
        \caption{Estimated vulnerability.\\ \ }
        \label{fig:bp_03_no_error}
    \end{subfigure}%
    \begin{subfigure}{.5\textwidth}
        \centering
        \includegraphics[scale=0.31]{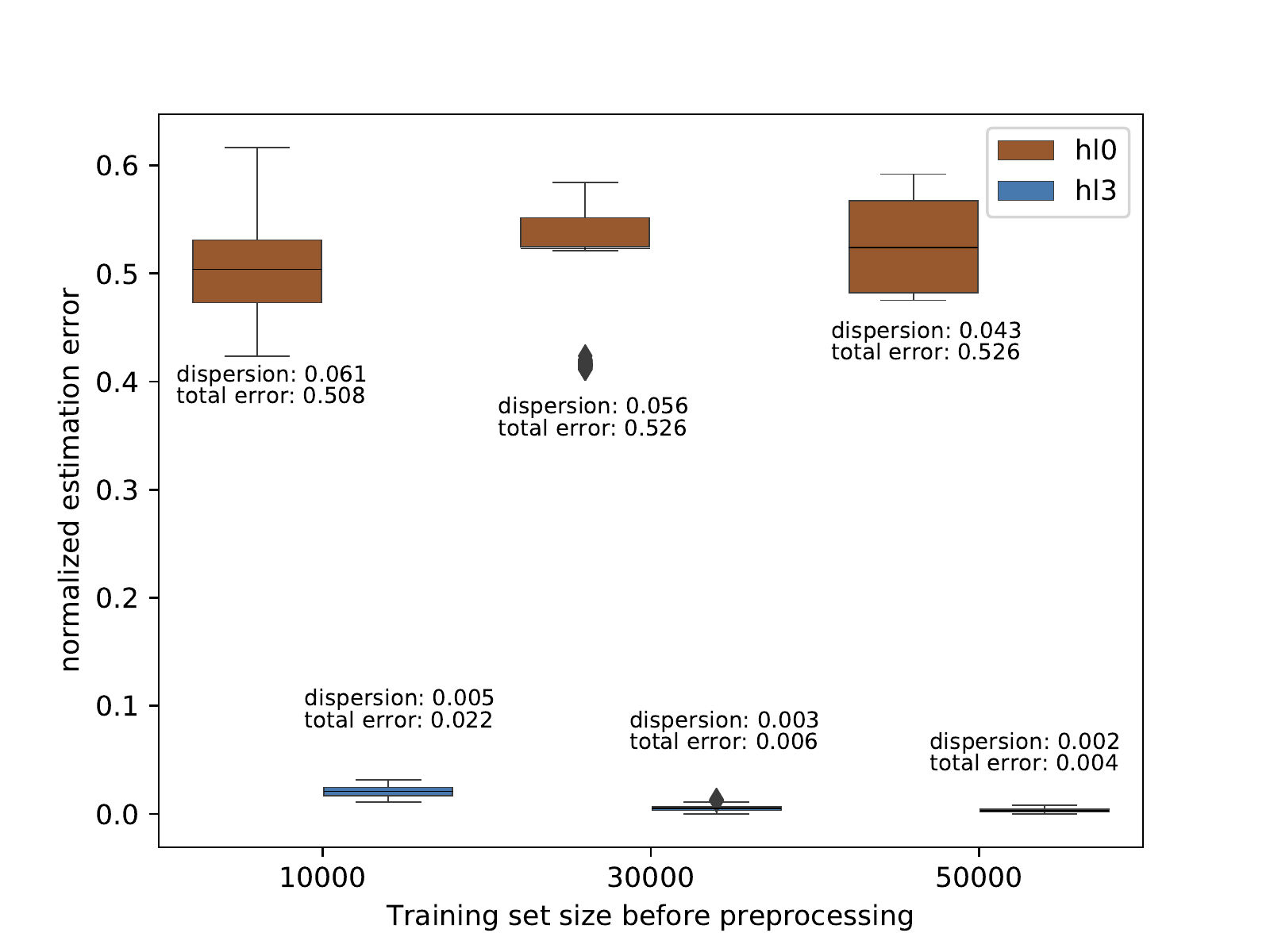}
        \caption{Normalized estimation error.\\ \ }
        \label{fig:bp_03_error}
    \end{subfigure}
    \caption{Multiple guesses scenario, comparison between a model with no hidden layers (\emph{hl0}) and the three hidden layered model present in the paper  (\emph{hl3}).}
    \label{fig:bp_03}
\end{figure*}
The model with no hidden layer is too simple to reproduce the input data distribution and therefore it does not generalize well in the problem of predicting the best (i.e. most probable) $w$ when a certain $y$ is input.}

Let us now focus on the results in~\cref{fig:bp_123}, where we compare the results of three different models with one, two and three hidden layers (respectively \emph{hl1}, \emph{hl2}, and \emph{hl3}). With this specific experiment we want to:
\begin{itemize}
    \item analyze the behavior of different models when we change the number of hidden layer (maintaining the number of hidden neurons per layer fixed to 100) and we consider different sizes for the learning set;
    \item introduce a possible heuristic to guide practitioners in the model selection task.
\end{itemize}
As we can see, observing the box-plots in~\cref{fig:bp_123_no_error} and~\cref{fig:bp_123_error} from left to right, when the amount of samples is relatively small, the shallow model, i.e. \emph{hl1} performs slightly better than the other two models. This is because, since the samples provided are not enough to accurately describe the distribution, a shallow model would be prone to under-fitting the training data, producing what seems to be a more general decision. However, as the number of samples increases, and consequently, we have a better representation of the distribution through the the data samples, a deeper model is able to reproduce the data distribution with increased accuracy. This results in better performances when trying to predict the best $w$ for a given input $y$. As one can see the three hidden layered model keeps improving when the amount of samples available for the training increases. The improvement is limited for the two hidden layers model while for the shallowest model, i.e. \emph{hl1}, there is not meaningful improvement at all.
\begin{figure*}[!htb]
    \centering
    \begin{subfigure}{.5\textwidth}
        \centering
        \includegraphics[scale=0.31]{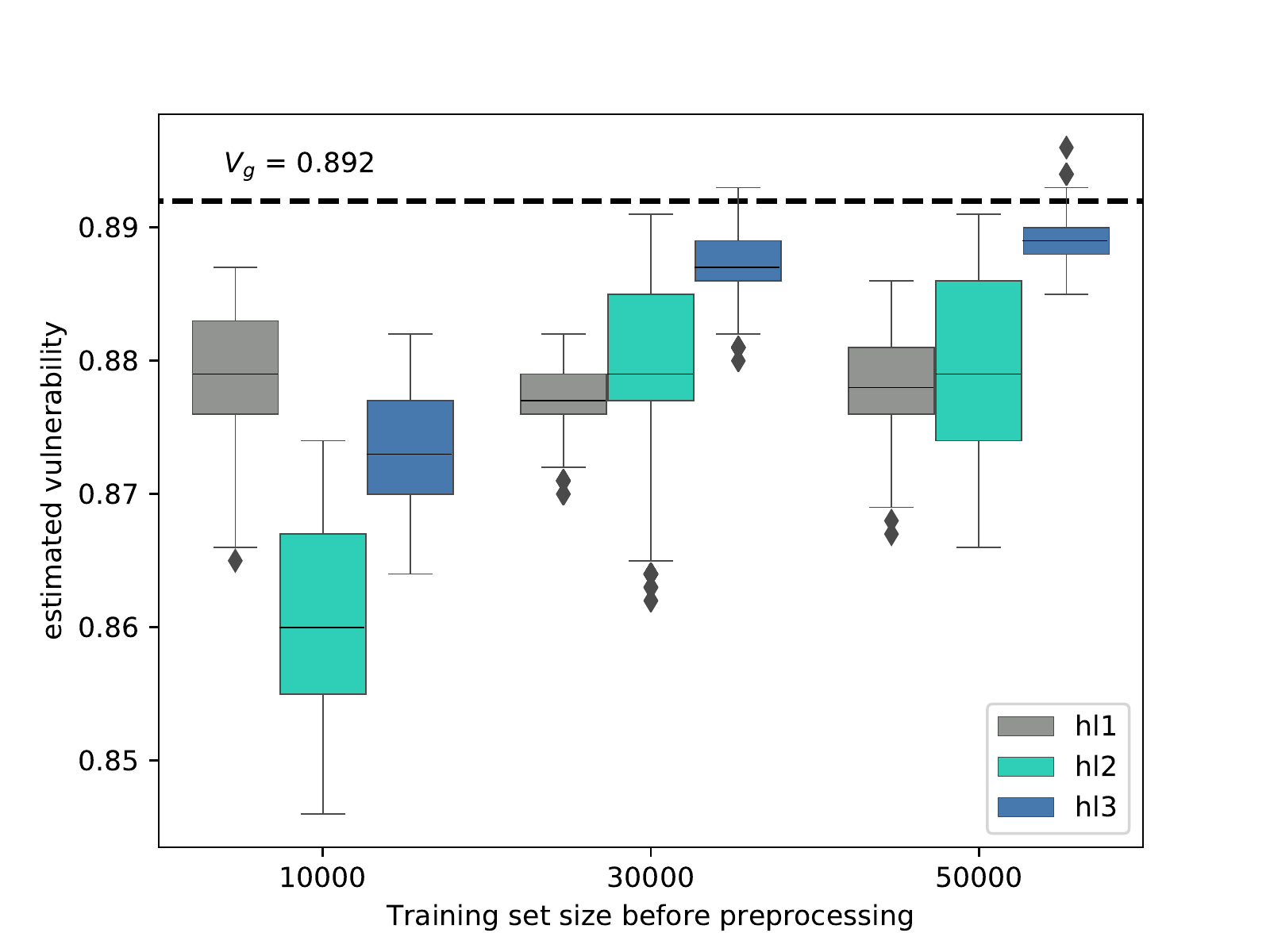}
        \caption{Estimated vulnerability.\\ \ }
        \label{fig:bp_123_no_error}
    \end{subfigure}%
    \begin{subfigure}{.5\textwidth}
        \centering
        \includegraphics[scale=0.31]{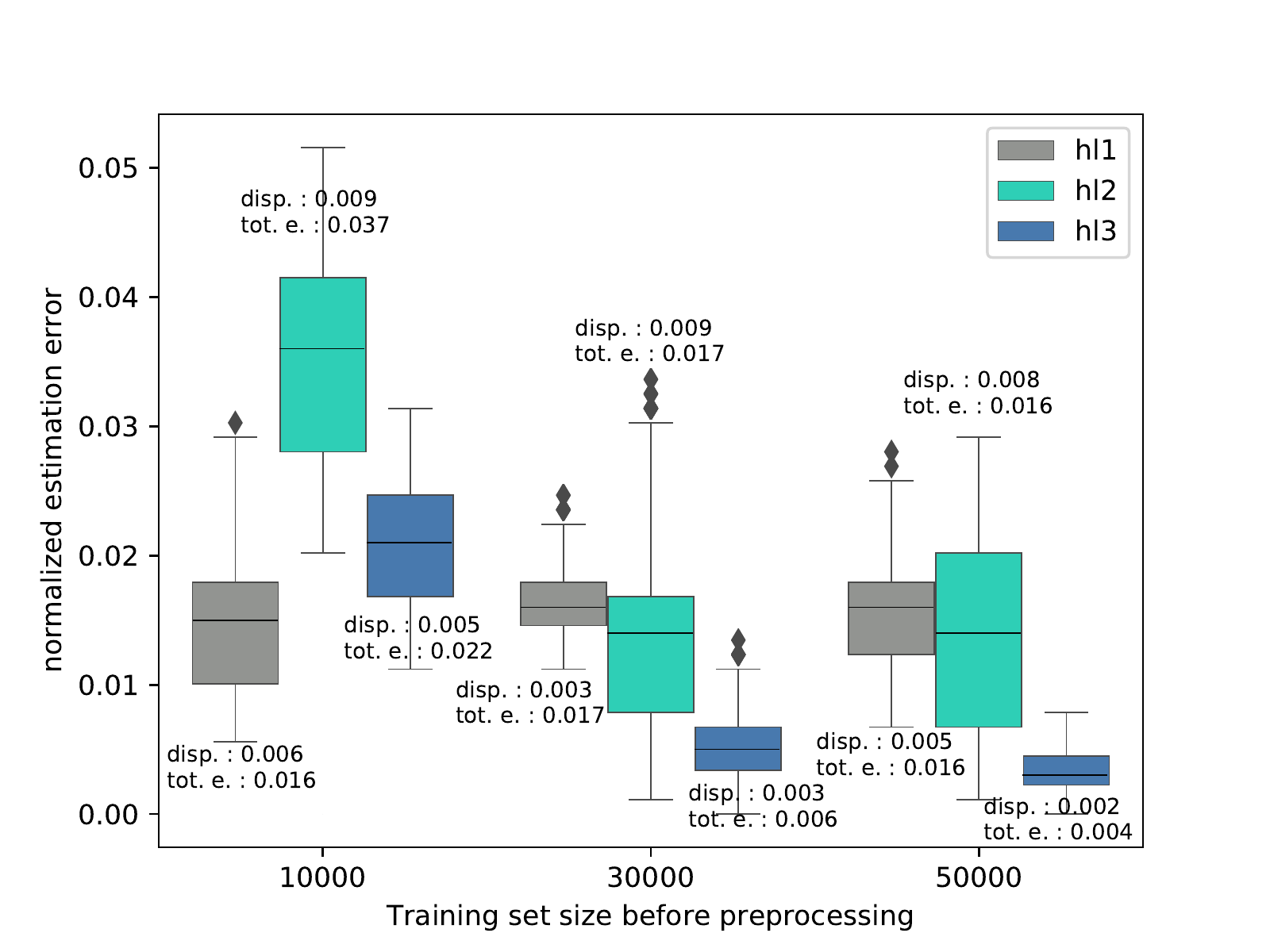}
        \caption{Normalized estimation error.\\ \ }
        \label{fig:bp_123_error}
    \end{subfigure}
    \caption{Multiple guesses scenario, comparison between a model with one hidden layers (\emph{hl1}), a model with two hidden layers (\emph{hl2}), and the three hidden layered model presented in the paper (\emph{hl3}).}
    \label{fig:bp_123}
\end{figure*}

Therefore, provided the availability of enough samples, a good heuristic for the practitioner would be to pick the model that maximizes the leakage, which represents the strongest adversary. In practice, this boils down to trying several models increasing their complexity as long as an increased complexity translates into a higher leakage estimation. This also holds for models with different architecture: if switching from dense to convolutional layers results in a higher leakage estimation, then the convolutional model should be preferred. Even though this is still an open problem for the machine learning field, and we do not provide any no guarantees that, eventually, the optimal model is going to be retrieved, this can be considered a valid empirical approach. Indeed, the principle of no free lunch in machine learning tells that no model can guarantee convergence on its final sample performance. Therefore when a finite amount of sample is available, the practitioner should evaluate several model and stick to a heuristic, as the one we suggest, in order to select the best model. 

\rev{In order to strengthen this point and also address the comment on the use of different architecture and their impact on the estimation, we produce yet another experimental result. 
\begin{figure}[!htb]
    \centering
    \includegraphics[scale=0.31]{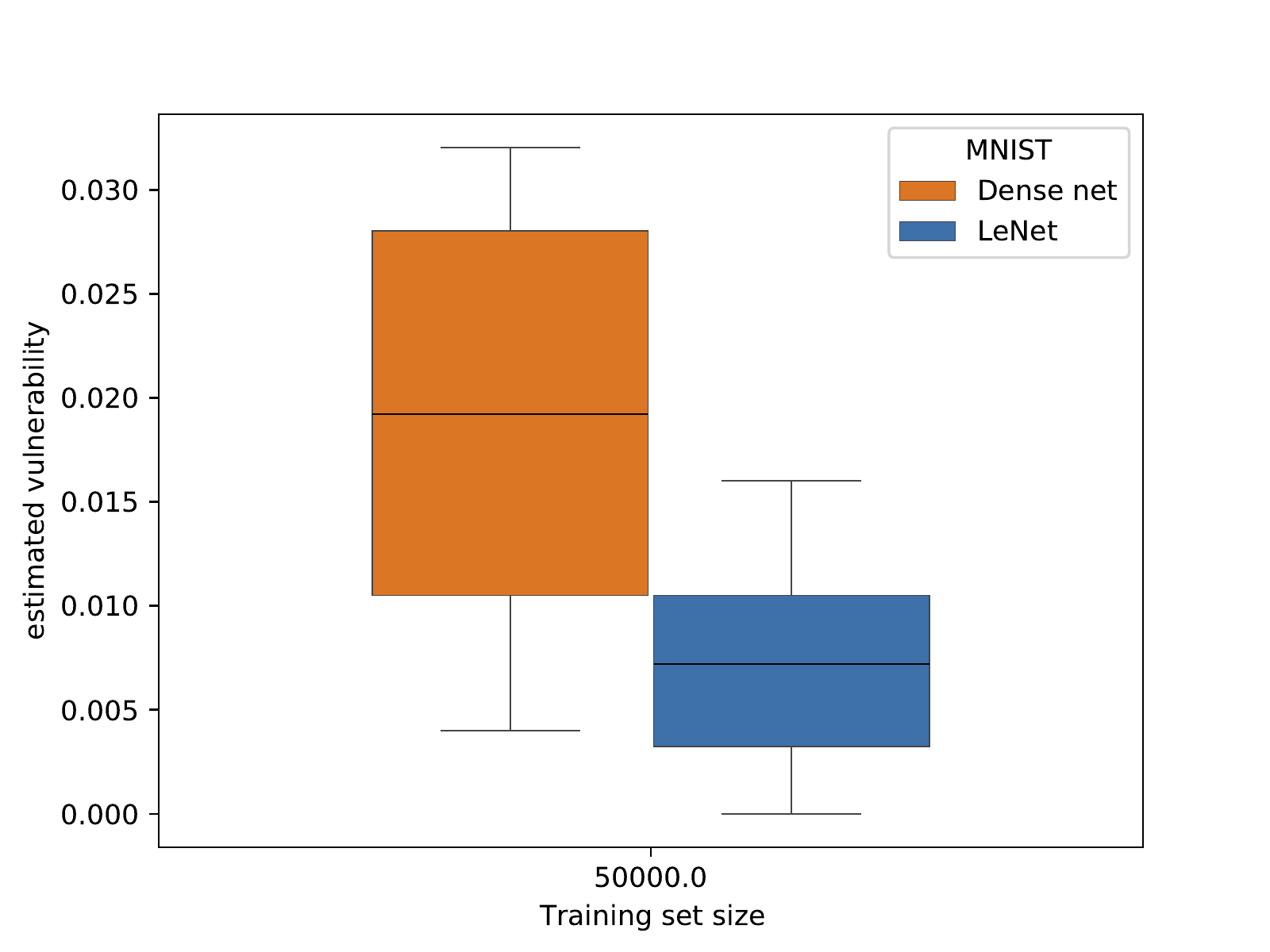}
    \caption{MNIST experiment: dense network vs. LeNet model estimation of the Bayes risk. A smaller risk corresponds to a higher leakage and a more powerful adversary able to take more advantage of the information revealed by the data.}
    \label{fig:mnist}
\end{figure}
Inspired by~\cite{Sekeh:20:TSP}, we consider the problem of estimating the Bayes error rate (BER, aka Bayes risk or Bayes leakage) using the MNIST dataset (a benchmark dataset for ML classification problems). In~\cite{Sekeh:20:TSP}, the authors use an empirical method to estimate bounds for the BER and, in order to do so, they need to perform dimensionality reduction (they try both principal component analysis, or PCA, and auto-encoding via an auto-encoder neural network). Indeed MNIST, being a $28\text{px}\times28\text{px}$ images dataset, contains high dimensionality samples and all the three applied methods (bounds estimation, k-NN and random forest) benefit from dimensionality reduction. Given that the samples distribution for MNIST is unknown, in~\cite{Sekeh:20:TSP} the authors aim at checking whether the three estimations confirm each other. And indeed they do. However, it is well known that in many image classification tasks, the state-of-the-art is represented by deep learning and, for instance convolutional neural nets. While in~\cite{Sekeh:20:TSP} the authors study how the addition of new layers to the model affects the bound estimation, we focus on the comparison between two network models, a dense and a convolutional one. }

\rev{In particular, the dense network model we consider has one hidden layer with 128 hidden units and ReLu activation function. Such an architecture corresponds to a model with 101700 training parameters (connection weights and bias weights). In order to reduce the tendency to over-fit the training data, we consider improving the model with two dropout layers, one before the hidden layer and one after, with dropout rate of $20\%$ and $50\%$ respectively.}

\rev{The convolutional neural network we consider is based on the one proposed in~\cite{Lecun:98:IEEE}, which goes by the name of LeNet. In particular, the implementation of this net only requires 38552 training parameters, almost one third of those required by the dense model which is shallower but has many hidden neurons. This convolutional network consists of a couple of 2D convolutional layers with $3\times3$ kernel, alternated with a couple of average pool layer. We then have two dense layers with 60 and 42 hidden neurons respectively and both preceded by dropout layers with dropout rate of $40\%$ and $30\%$ respectively.}

\rev{Both the dense and the LeNet model have a soft max final layer with 10 nodes, one for each class. We train both models on the MNIST 50000 training samples. We split the remaining 10000 samples set into 10 subsets, each with 1000 samples. We use the trained models to estimate the BER on these 10 subsets and we obtain the results represented in~\cref{fig:mnist}, where the estimated vulnerability is the estimated BER. The average values, represented by the black horizontal lines within the box-plots, are directly comparable to the results in~\cite{Sekeh:20:TSP}. }

\rev{We notice that:
\begin{itemize}
    \item considering the aforementioned paper, the dense network's performances are comparable to those of the random forest and k-NN algorithms and the LeNet's performances are comparable to those of the convolutional net considered by the authors;
    \item the LeNet model's estimate is lower, i.e. the modeled Bayesian adversary is stronger than in the case of the dense net;
    \item according to the previously proposed heuristics, given that we can only estimate the Bayes leakage from samples for the MNIST case, the results of the LeNet model are of higher interest, given that it is bigger than the leakage estimate through the bound;
    \item it is therefore important, when only relying on data for leakage estimation, to compare different models and always look for the state-of-the-art to design the adversary which exploits the system's leakage.
\end{itemize}}

\section{Majority vote}
\begin{figure*}[!htb]
    \centering
    \begin{subfigure}{.5\textwidth}
        \centering
        \includegraphics[scale=0.31]{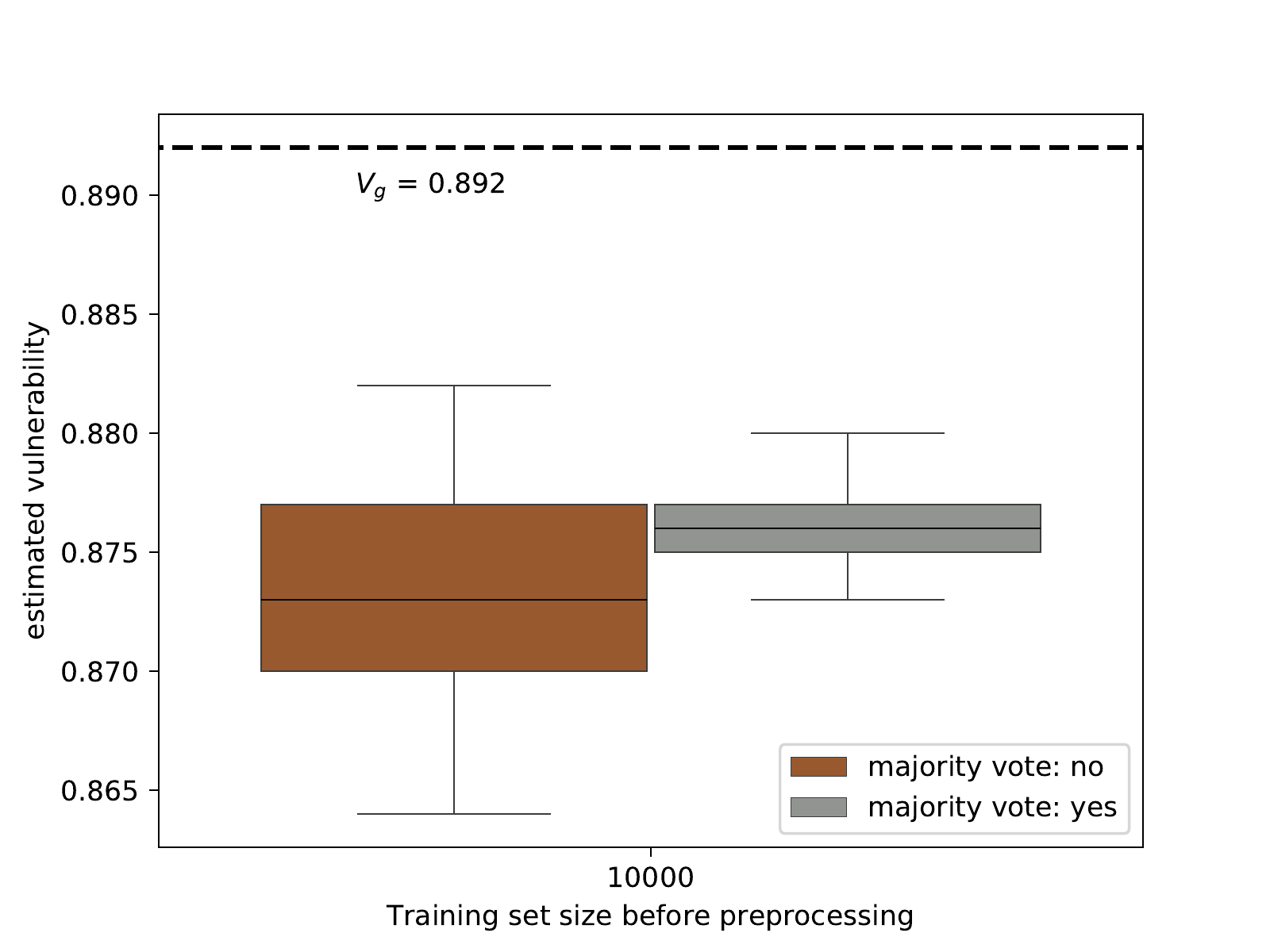}
        \caption{Estimated vulnerability.\\ \ }
        \label{fig:maj_vote_a_1}
    \end{subfigure}%
    \begin{subfigure}{.5\textwidth}
        \centering
        \includegraphics[scale=0.31]{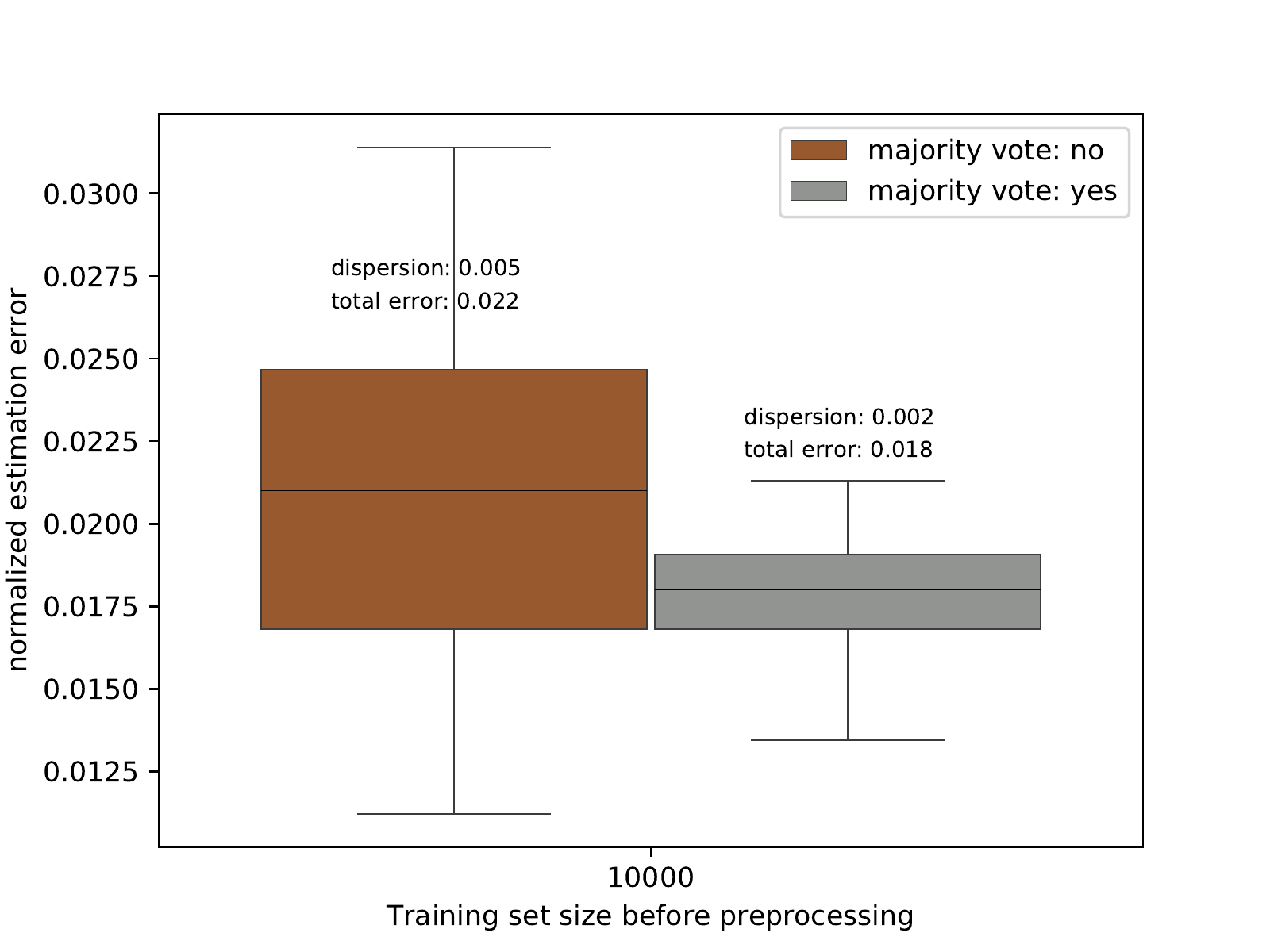}
        \caption{Normalized estimation error.\\ \ }
        \label{fig:maj_vote_a_2}
    \end{subfigure}
    \caption{Comparison between our leakage estimate and the one obtained with a majority vote based model. In both cases each model is trained on 10K samples.}
    \label{fig:maj_vote_a}
\end{figure*}

\rev{In this section we show an alternative procedure to perform leakage estimation. Given a set of models, instead of using each one of them to obtain an estimate, we derive a new model from them and we use this new model to estimate the leakage. According to~\cite{Tumer:03:IJSESD}, we obtain the new model by taking a majority vote on the predictions of each model in the model set, i.e. when each one of the models receives the input it outputs a class. The class eventually assigned to the observable is the class most frequently predicted by the model ensemble (or a random one in case of ties, since we are considering the simplest way to aggregate models). }

\rev{We consider the first experiment proposed in our paper, in particular the 5 models obtained by training on the 5 i.i.d. training sets of size 10K samples. }

\rev{As we can see in~\cref{fig:maj_vote_a}, the box-plot corresponding to the estimation not based on the majority vote is the same already reported (cfr.~\cref{mult_guess_section}). Compared to it, the one based on majority vote shows a higher (and therefore more precise) average estimation and a lower dispersion. }

\begin{figure*}[!htb]
    \centering
    \begin{subfigure}{.5\textwidth}
        \centering
        \includegraphics[scale=0.31]{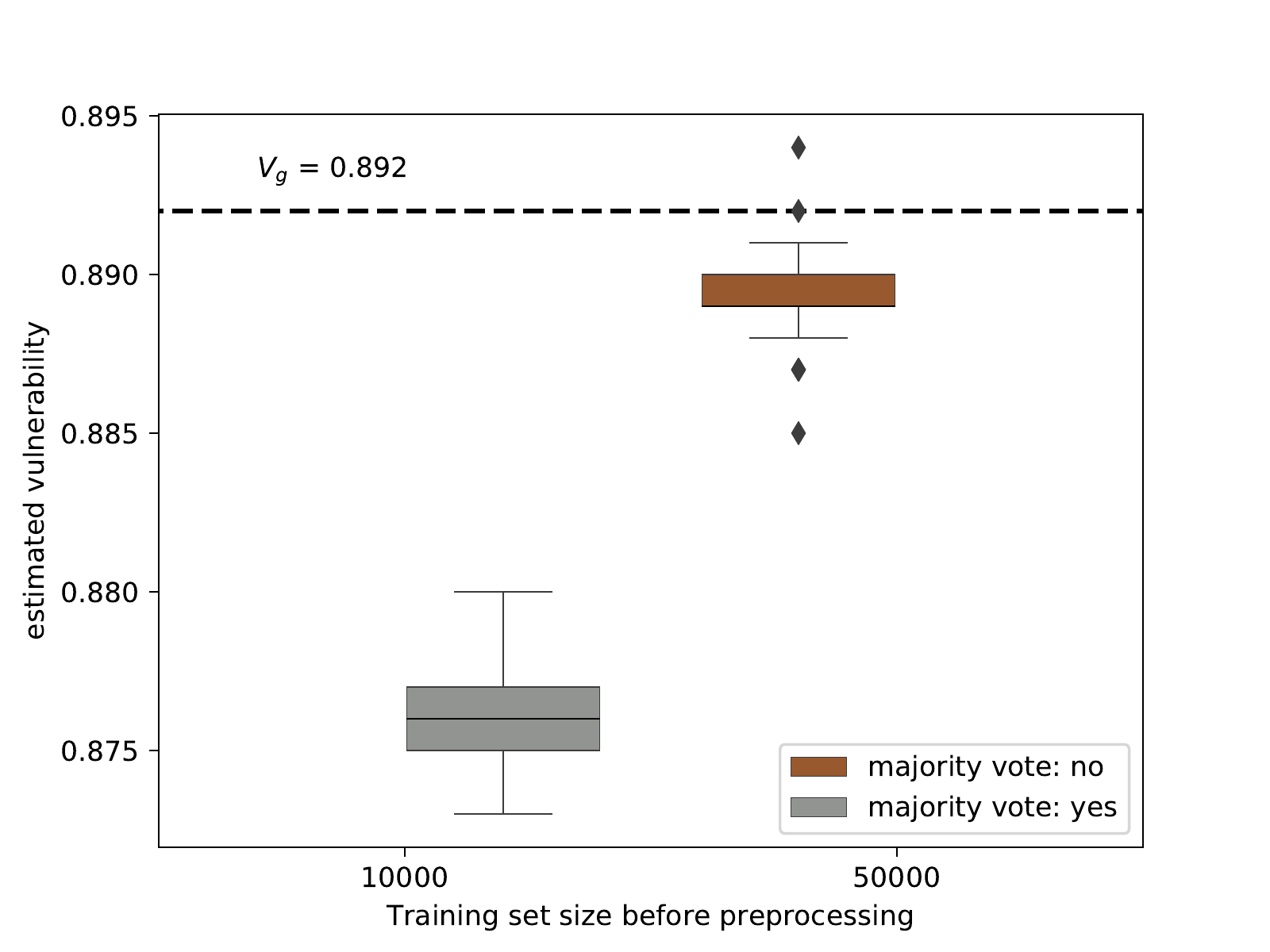}
        \caption{Estimated vulnerability.\\ \ }
        \label{fig:maj_vote_b_1}
    \end{subfigure}%
    \begin{subfigure}{.5\textwidth}
        \centering
        \includegraphics[scale=0.31]{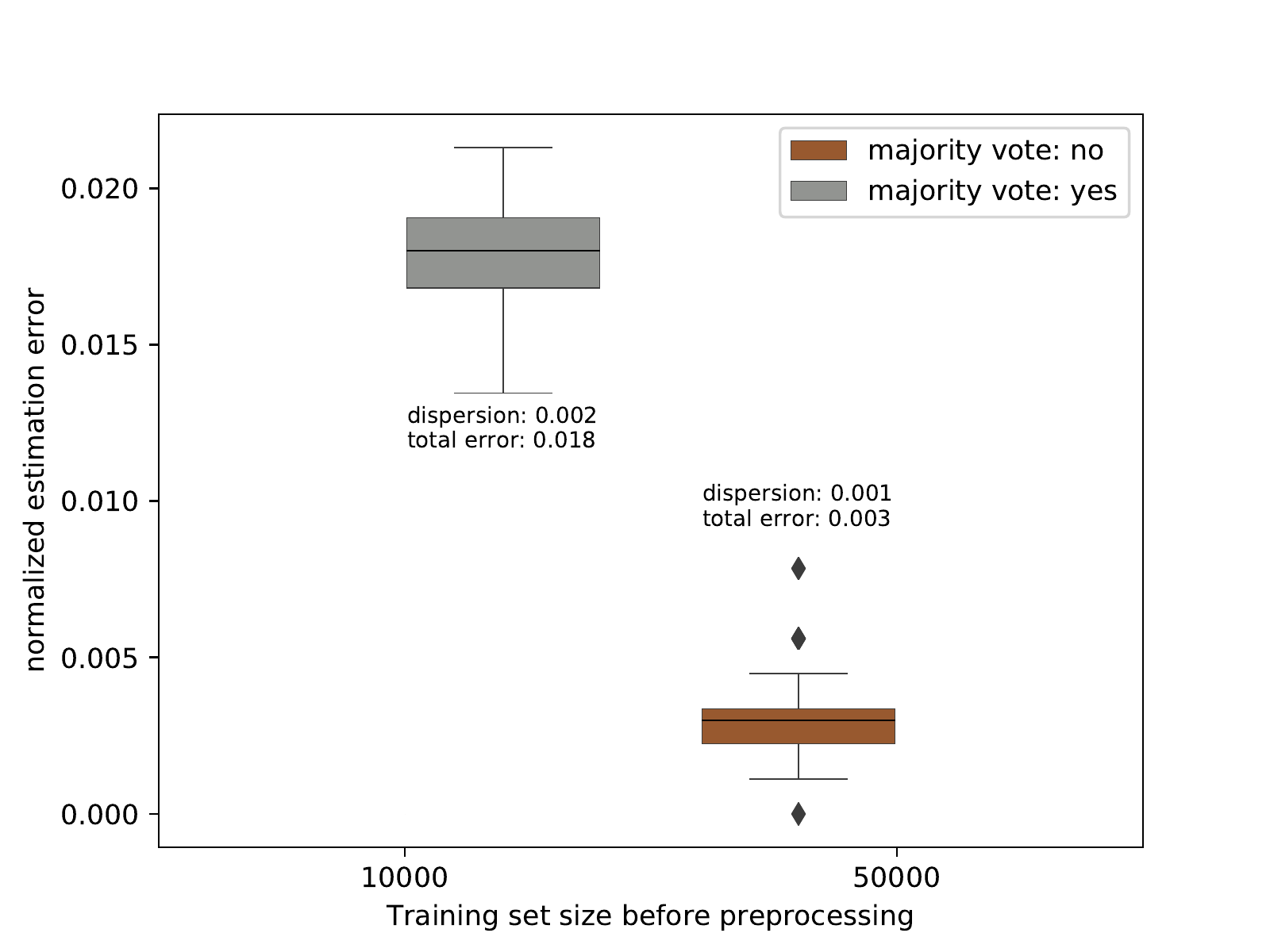}
        \caption{Normalized estimation error.\\ \ }
        \label{fig:maj_vote_b_2}
    \end{subfigure}
    \caption{Comparison between the majority vote model and a model trained on all the samples available to each model involved in the majority vote.}
    \label{fig:maj_vote_b}
\end{figure*}

\rev{However, it is important to notice that, when in this case we consider a majority vote ensemble, since we are exploiting the ``expertise'' of many models trained on i.i.d. training samples of size 10K, we are actually exploiting the knowledge coming from 50K samples. We therefore analyze the case in which, according to what has already been done in~\cref{mult_guess_section}, the same 50K samples, obtained by merging the 5 datasets with 10K samples each, are used to train a single model. The results are showed in~\cref{fig:maj_vote_b}. In this case, it is clear that having multiple weak classifiers and taking the majority vote according to the simple procedure described above, gives worse results in terms of leakage estimation performances than using all the samples to train a strong classifier.}

\section{Supplementary plots}
\label{extra_plots}
In the next pages, we show supplementary plots for each experiment presented in~\autoref{experiments}. 
In particular, we have included the  plots representing the comparison between the estimated  vulnerability and the real one, 
and the plots showing the comparison between the frequentist approach and ours. 

\begin{figure*}[!htbp]
\centering
\begin{minipage}[b]{\textwidth}
  \centering
  \begin{subfigure}[t]{.31\linewidth}
    \centering\includegraphics[width=0.8\linewidth]{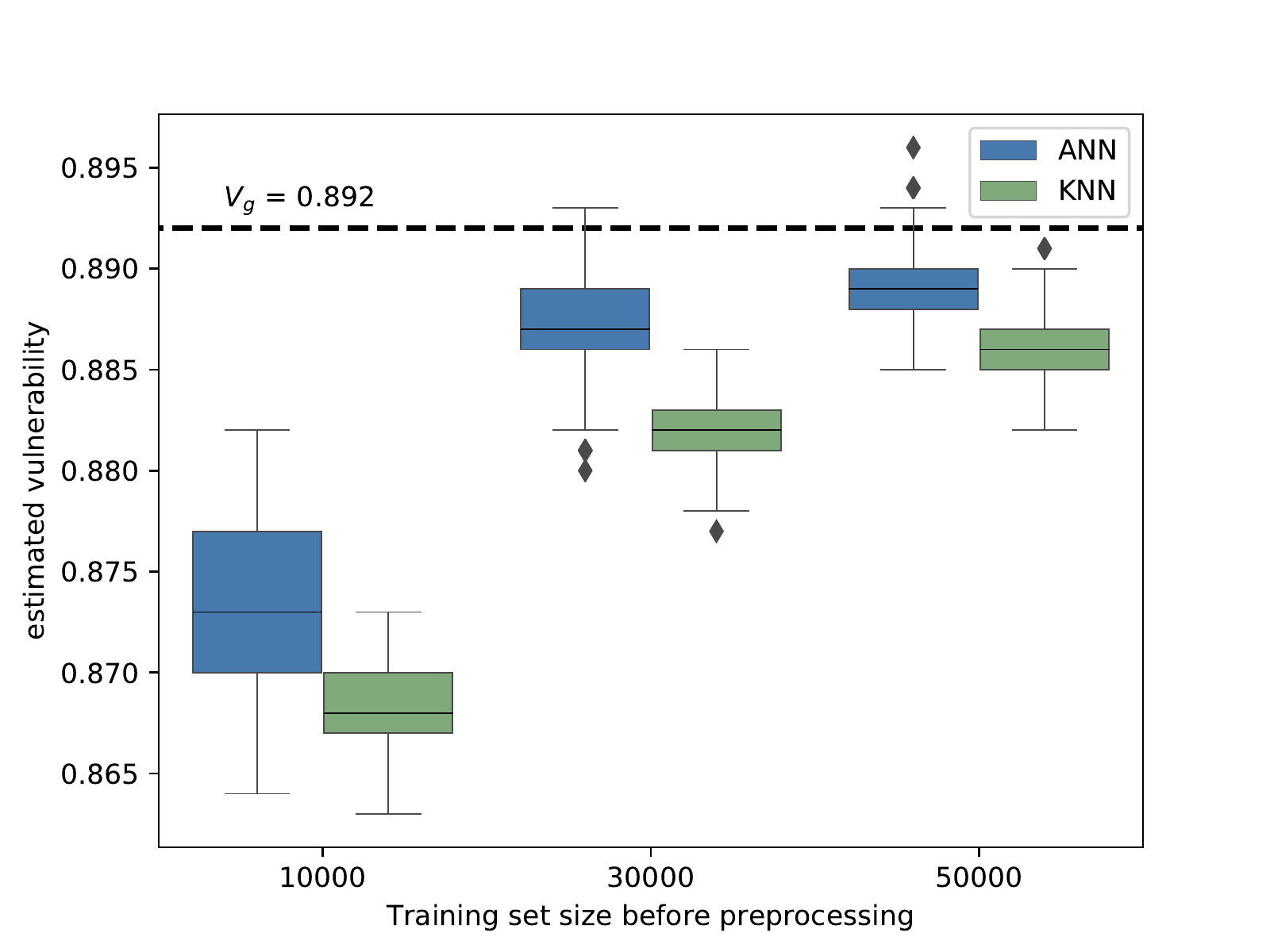}
    \caption{Vulnerability estimation for ANN and k-NN with data pre-processing.\\}
  \end{subfigure}
  \quad
  \begin{subfigure}[t]{.31\linewidth}
    \centering\includegraphics[width=0.8\linewidth]{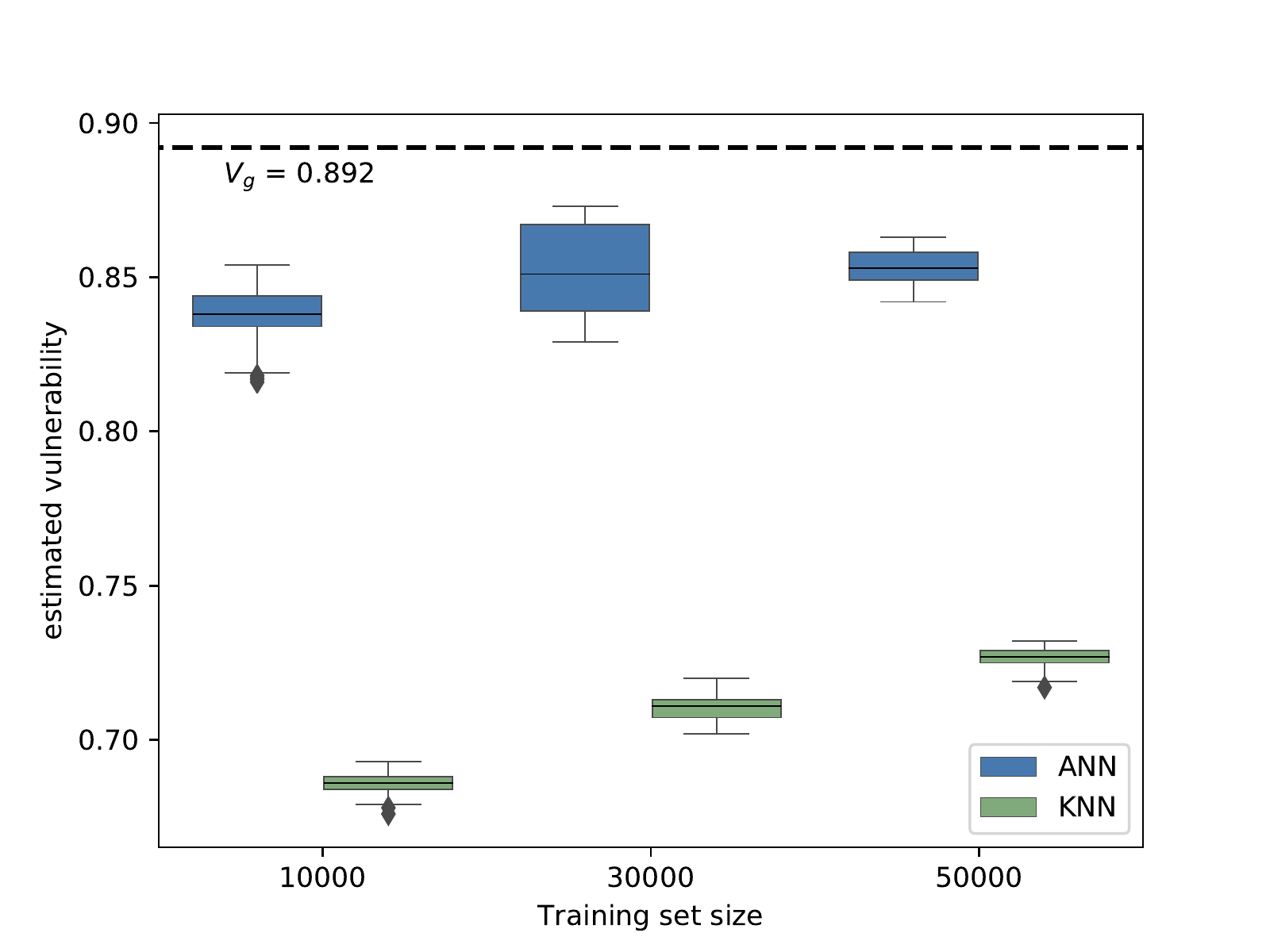}
    \caption{Vulnerability estimation for ANN and k-NN with channel pre-processing.\\}
  \end{subfigure}
  \quad
  \begin{subfigure}[t]{.31\linewidth}
    \centering\includegraphics[width=0.8\linewidth]{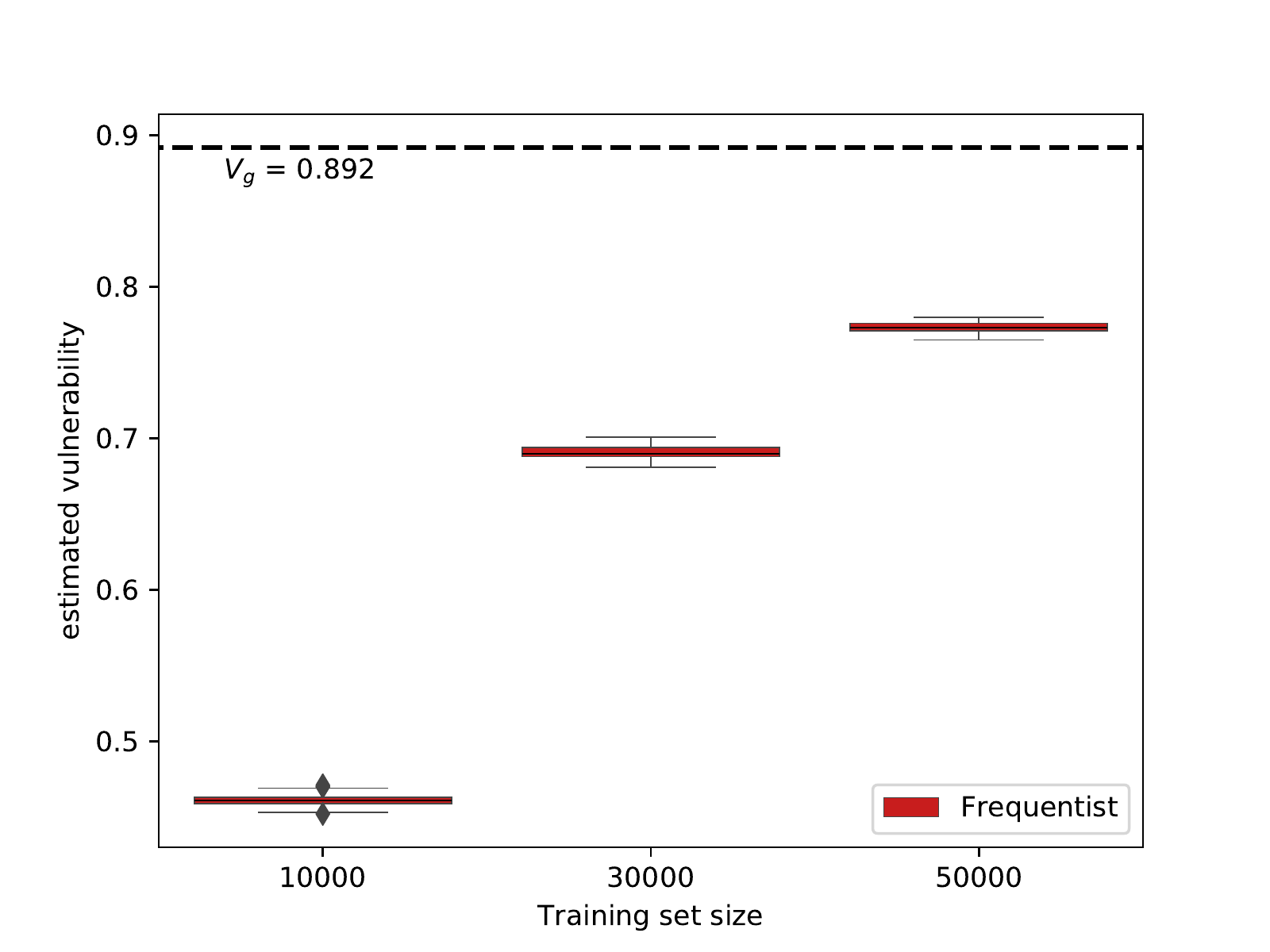}
    \caption{Vulnerability estimation for the frequentist approach.\\}
  \end{subfigure}
\end{minipage}
\begin{minipage}[b]{\textwidth}
  \centering
  \begin{subfigure}[t]{.31\linewidth}
   \centering\includegraphics[width=0.8\linewidth]{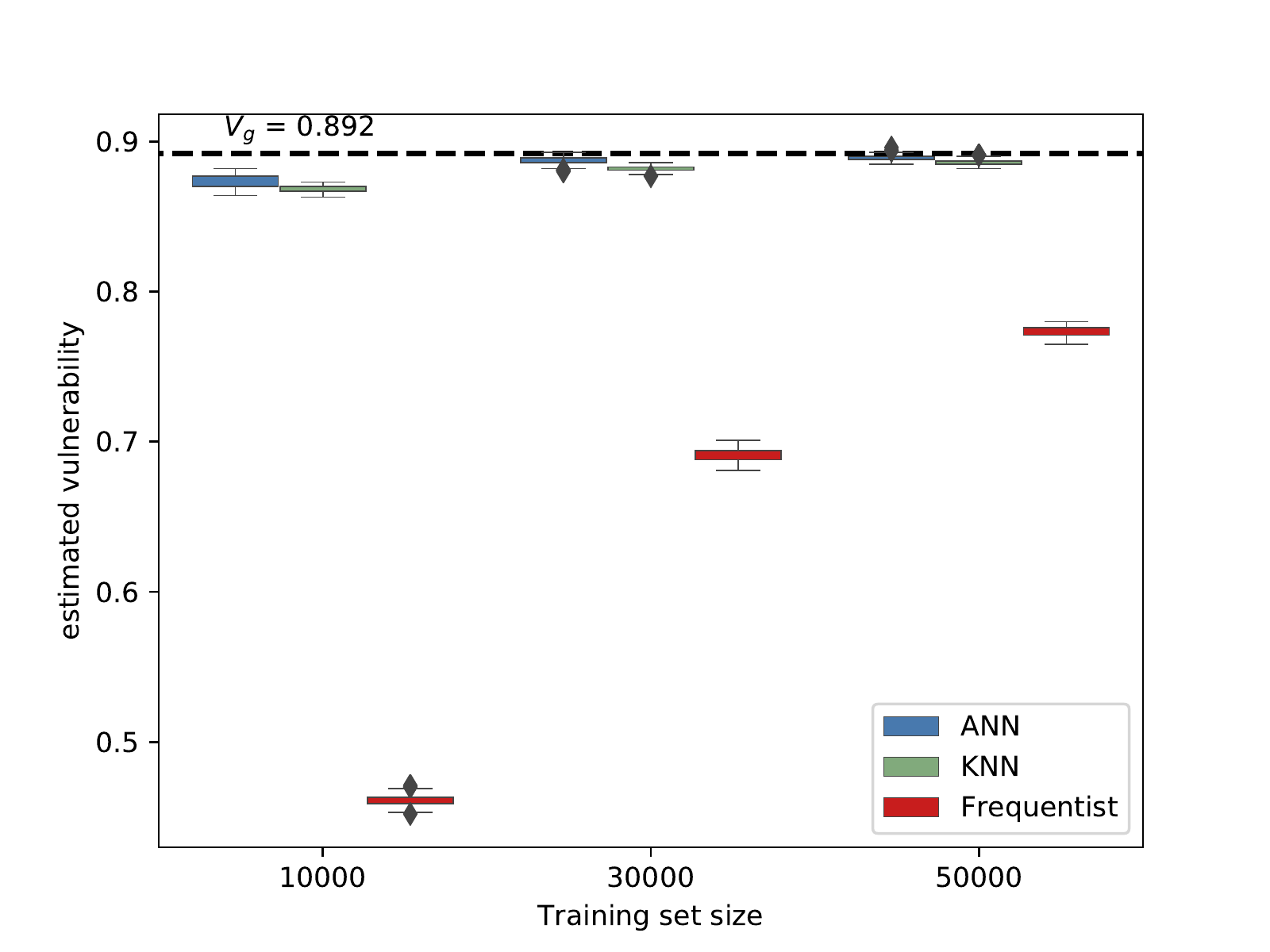}
    \caption{Vulnerability estimation for ANN and k-NN with data pre-processing, and the frequentist approach.}
  \end{subfigure}
  \quad
  \begin{subfigure}[t]{.31\linewidth}
    \centering\includegraphics[width=0.8\linewidth]{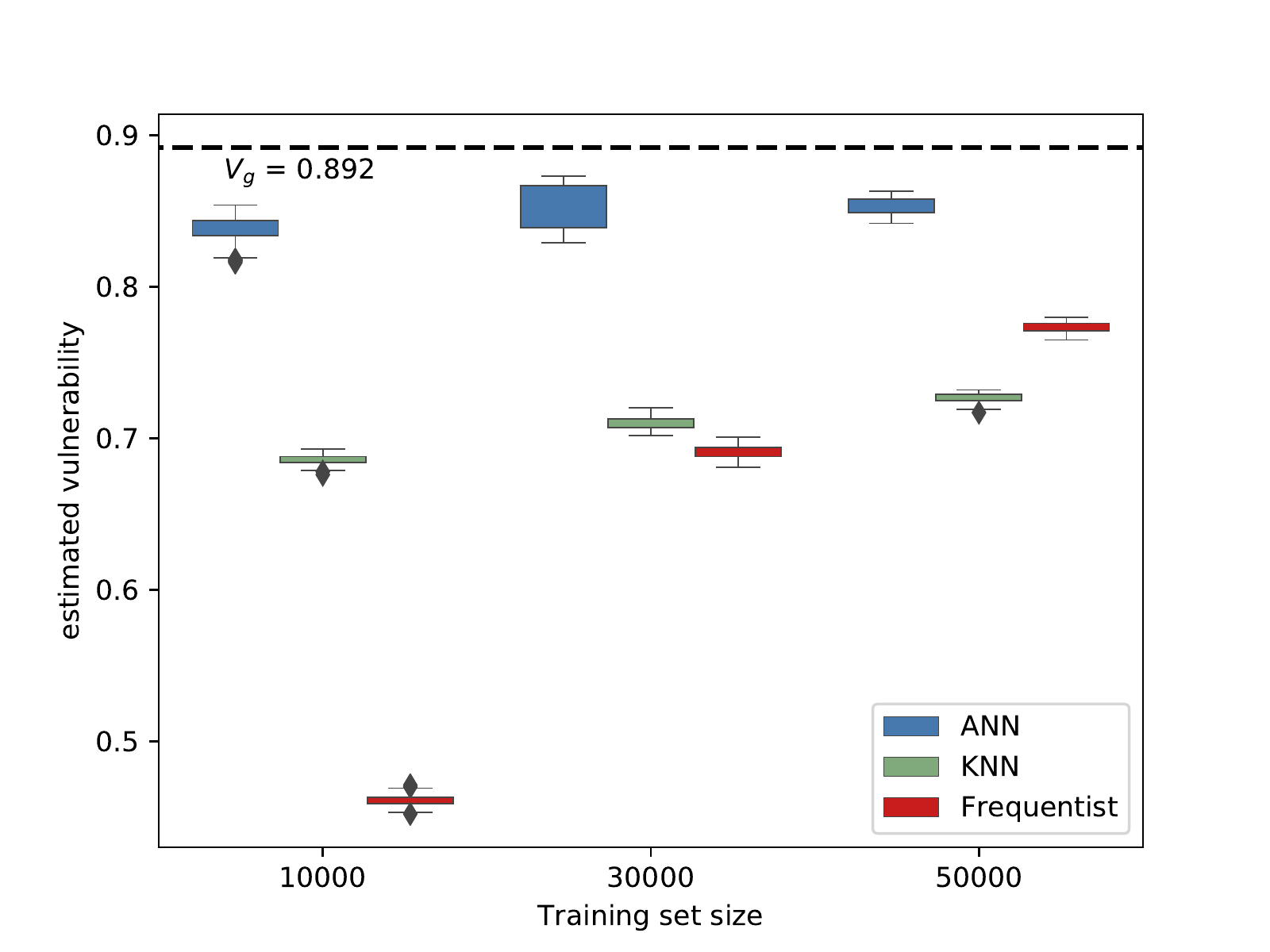}
        \caption{Vulnerability estimation for ANN and k-NN with channel pre-processing, and the frequentist approach.\\}
  \end{subfigure}
\quad
  \begin{subfigure}[t]{.31\linewidth}
  \centering\includegraphics[width=0.8\linewidth]{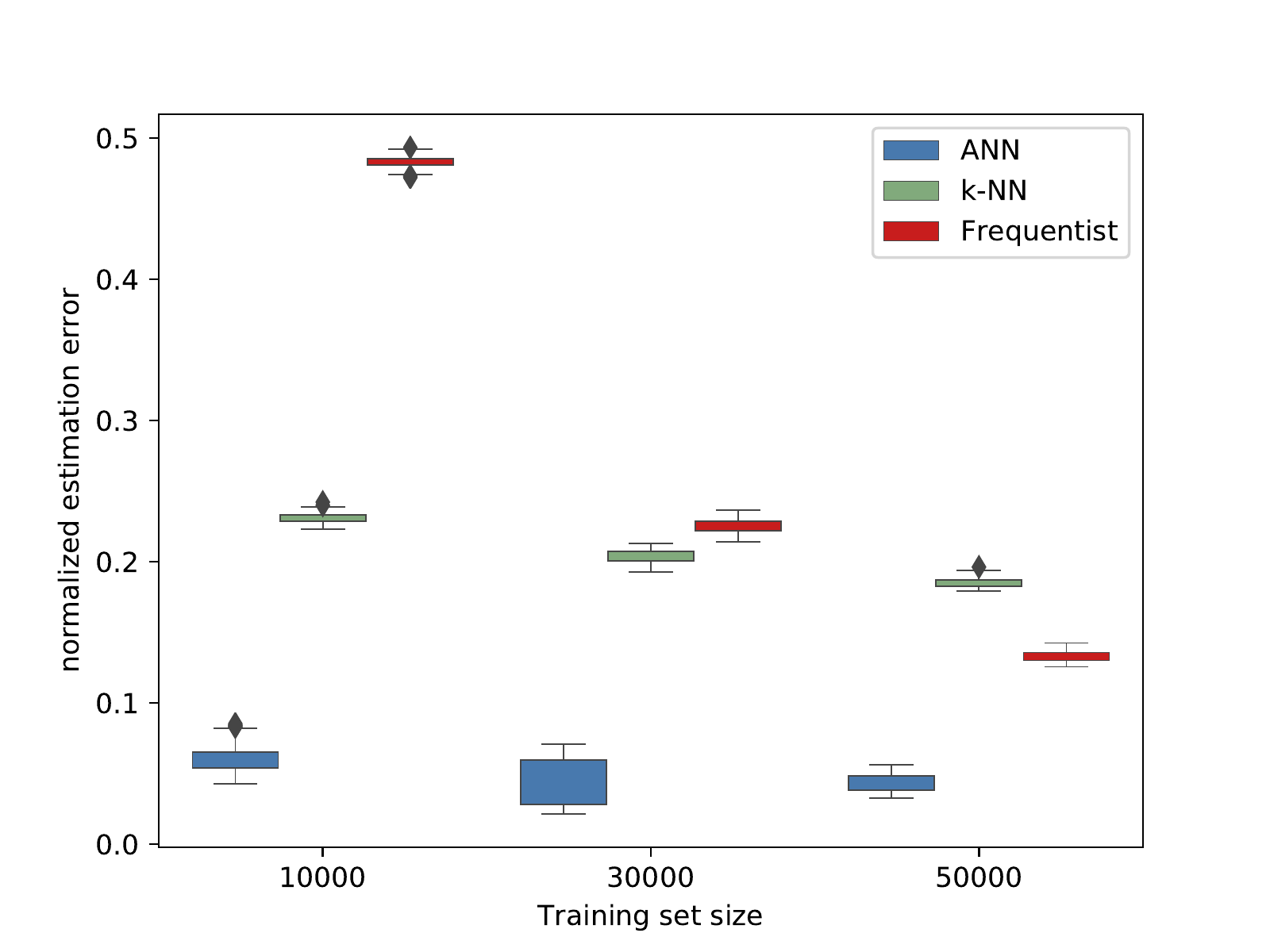}
        \caption{Normalized estimation error for   ANN and k-NN with channel pre-processing, and the frequentist approach.\\}
  \end{subfigure}
\end{minipage}
\caption{Supplementary plots for the multiple-guesses experiment.}\label{suppl_mult_guess}
\end{figure*}

\begin{figure*}[!htbp]
\centering
\begin{minipage}[b]{\textwidth}
  \centering
  \begin{subfigure}[t]{.31\linewidth}
    \centering\includegraphics[width=0.8\linewidth]{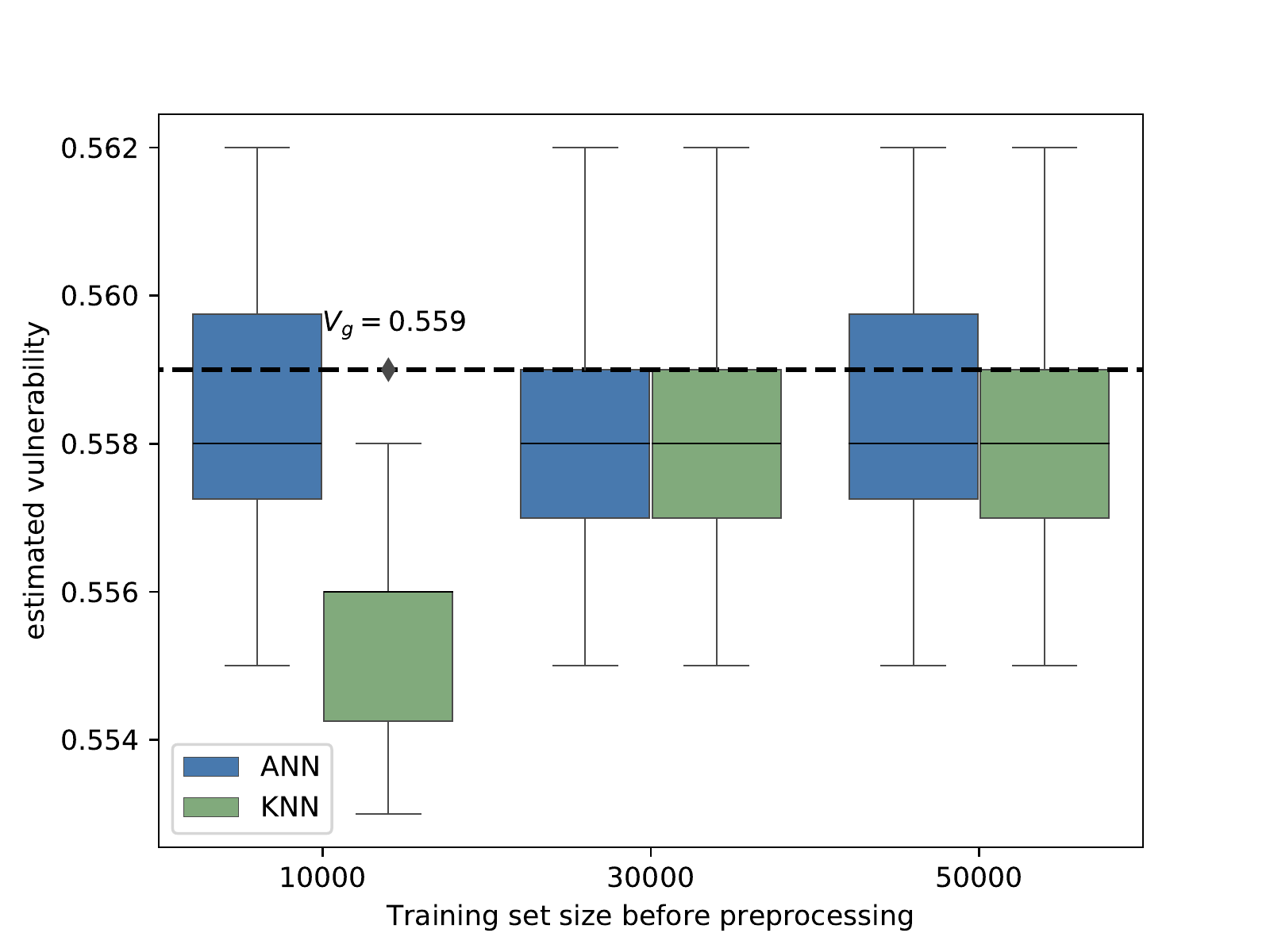}
     \caption{Vulnerability estimation for ANN and k-NN with data and channel pre-processing.\\}
  \end{subfigure}
  \quad
  \begin{subfigure}[t]{.31\linewidth}
    \centering\includegraphics[width=0.8\linewidth]{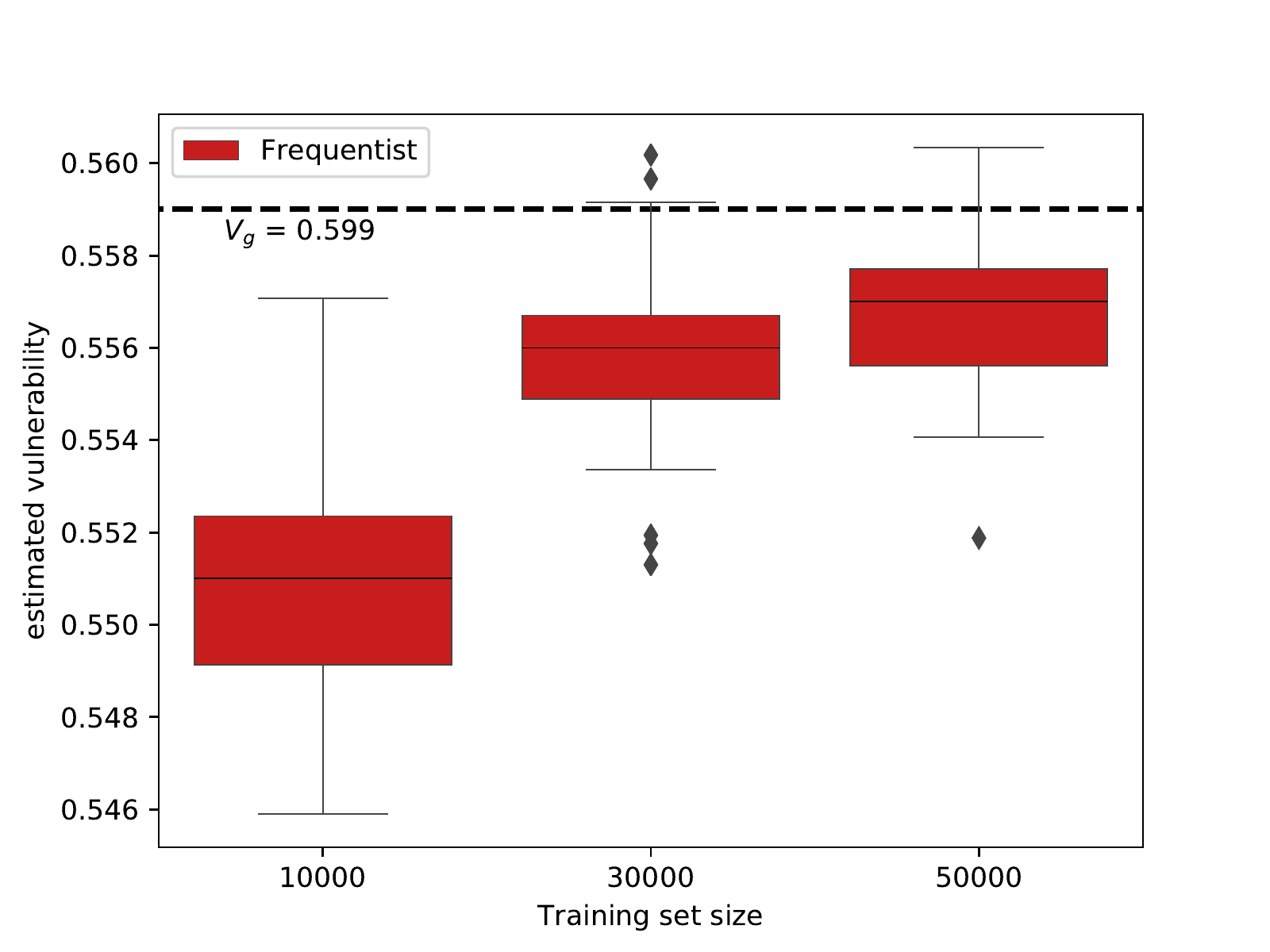}
    \caption{Vulnerability estimation for the frequentist approach.\\}
  \end{subfigure}
  \quad
  \begin{subfigure}[t]{.31\linewidth}
    \centering\includegraphics[width=0.8\linewidth]{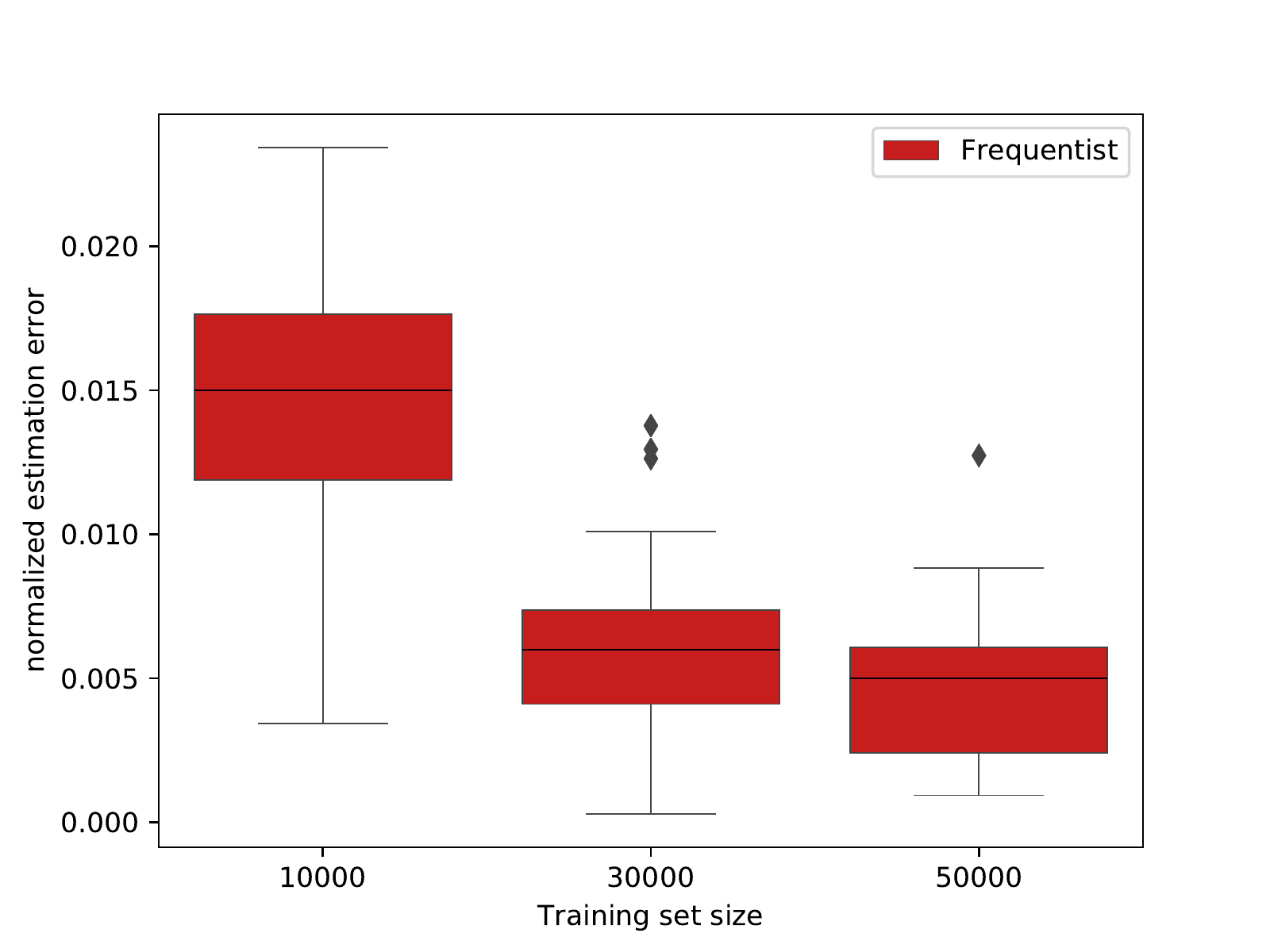}
    \caption{Normalized estimation error for the frequentist approach.\\}
  \end{subfigure}
 \end{minipage}
\begin{minipage}[b]{\textwidth}
  \centering
  \begin{subfigure}[t]{.31\linewidth}
    \centering\includegraphics[width=0.8\linewidth]{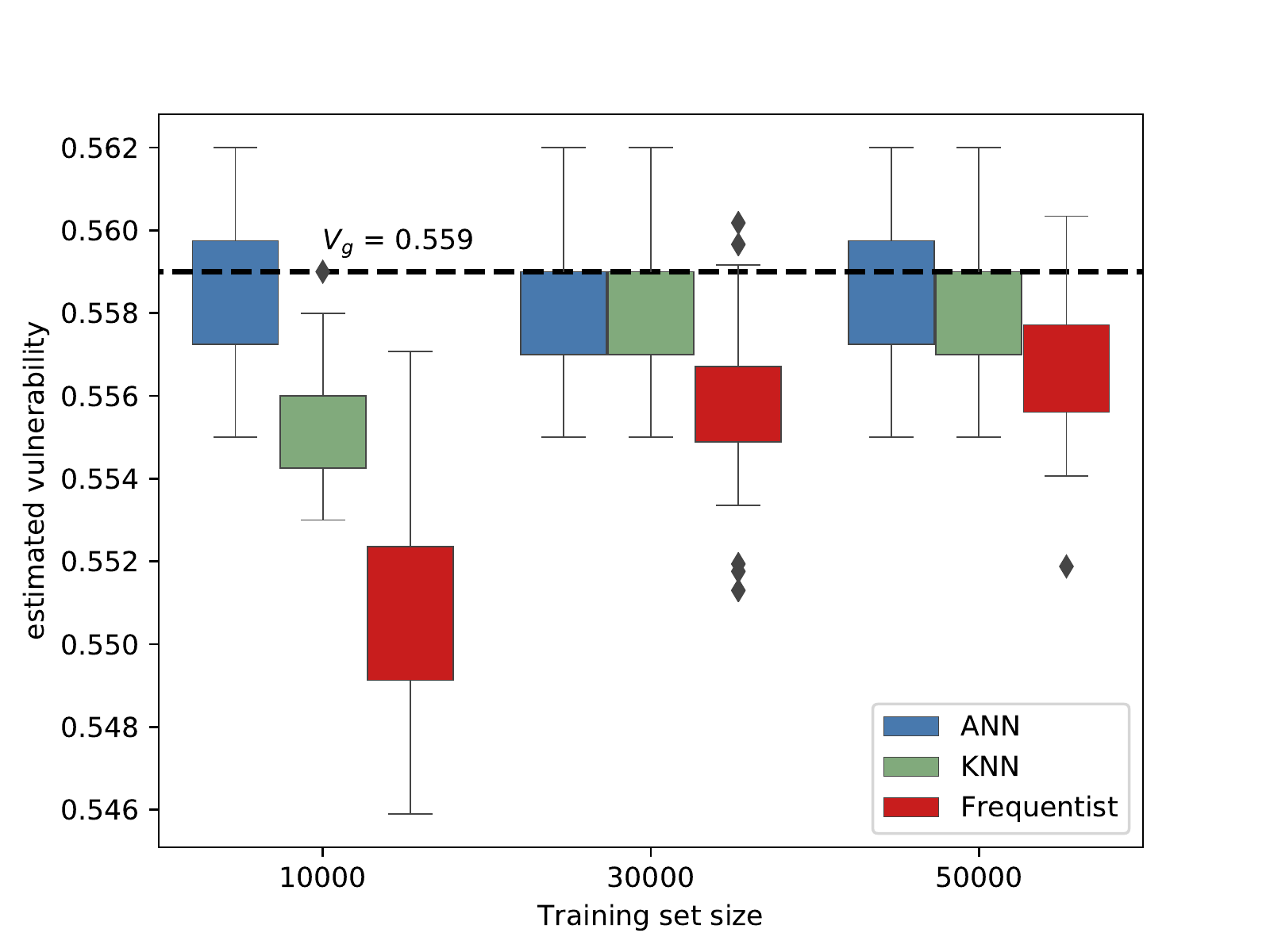}
    \caption{Vulnerability estimation for ANN and k-NN with data and channel pre-processing, and the frequentist approach.\\}
  \end{subfigure}
  \quad
  \begin{subfigure}[t]{.31\linewidth}
    \centering\includegraphics[width=0.8\linewidth]{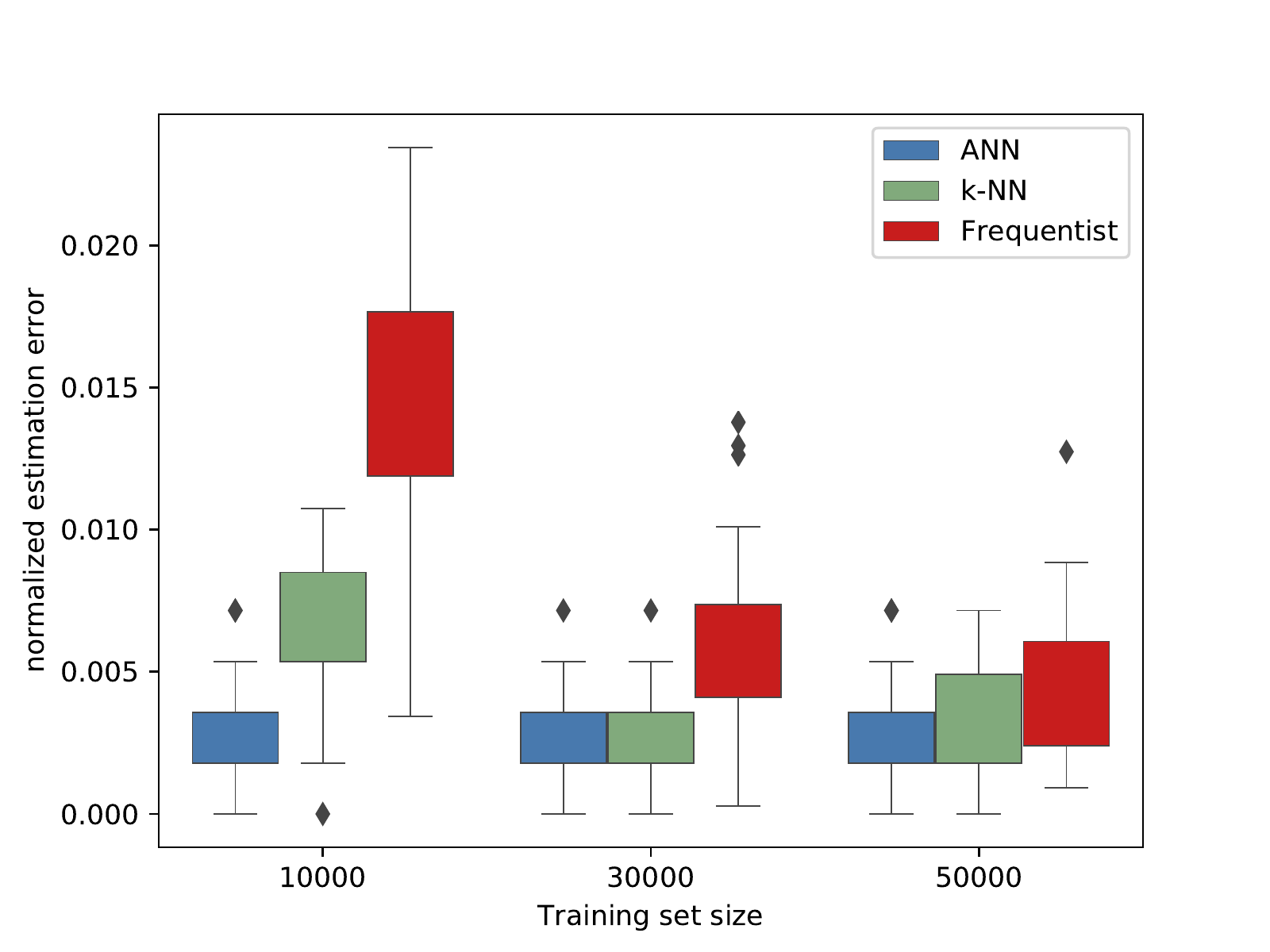}
    \caption{Normalized estimation error for ANN and k-NN with data and channel pre-processing, and the frequentist approach.\\}
  \end{subfigure}
 \end{minipage}
\caption{Supplementary plots for the password-checker experiment.}\label{suppl_side_channel}
\end{figure*}

\begin{figure*}[!htbp]
\centering
\begin{minipage}[b]{\textwidth}
  \centering
  \begin{subfigure}[t]{.4\linewidth}
    \centering\includegraphics[width=0.8\linewidth]{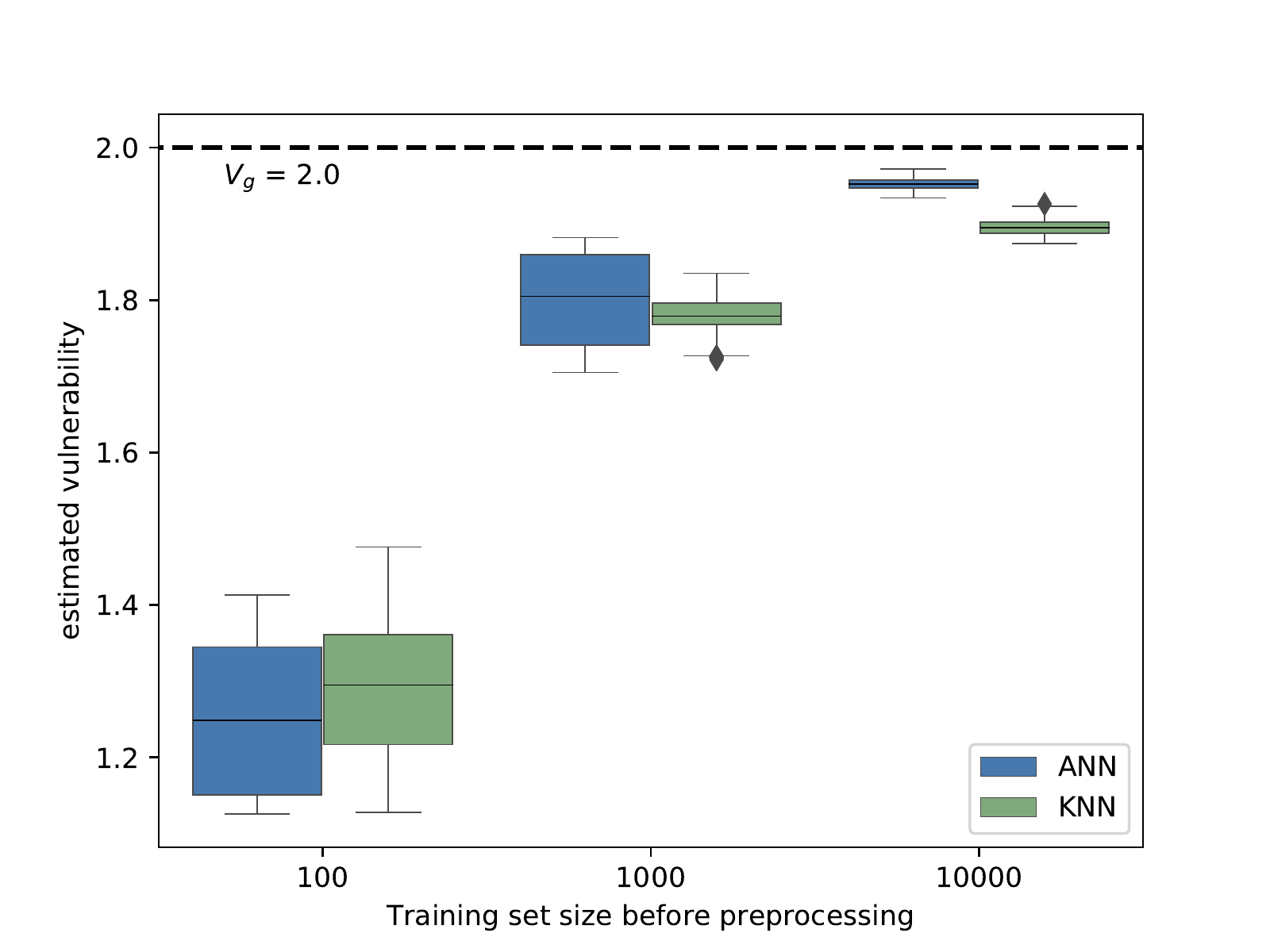}
    \caption{Vulnerability estimation for ANN and k-NN with data pre-processing.\\}
  \end{subfigure}
  \quad
  \begin{subfigure}[t]{.4\linewidth}
    \centering\includegraphics[width=0.8\linewidth]{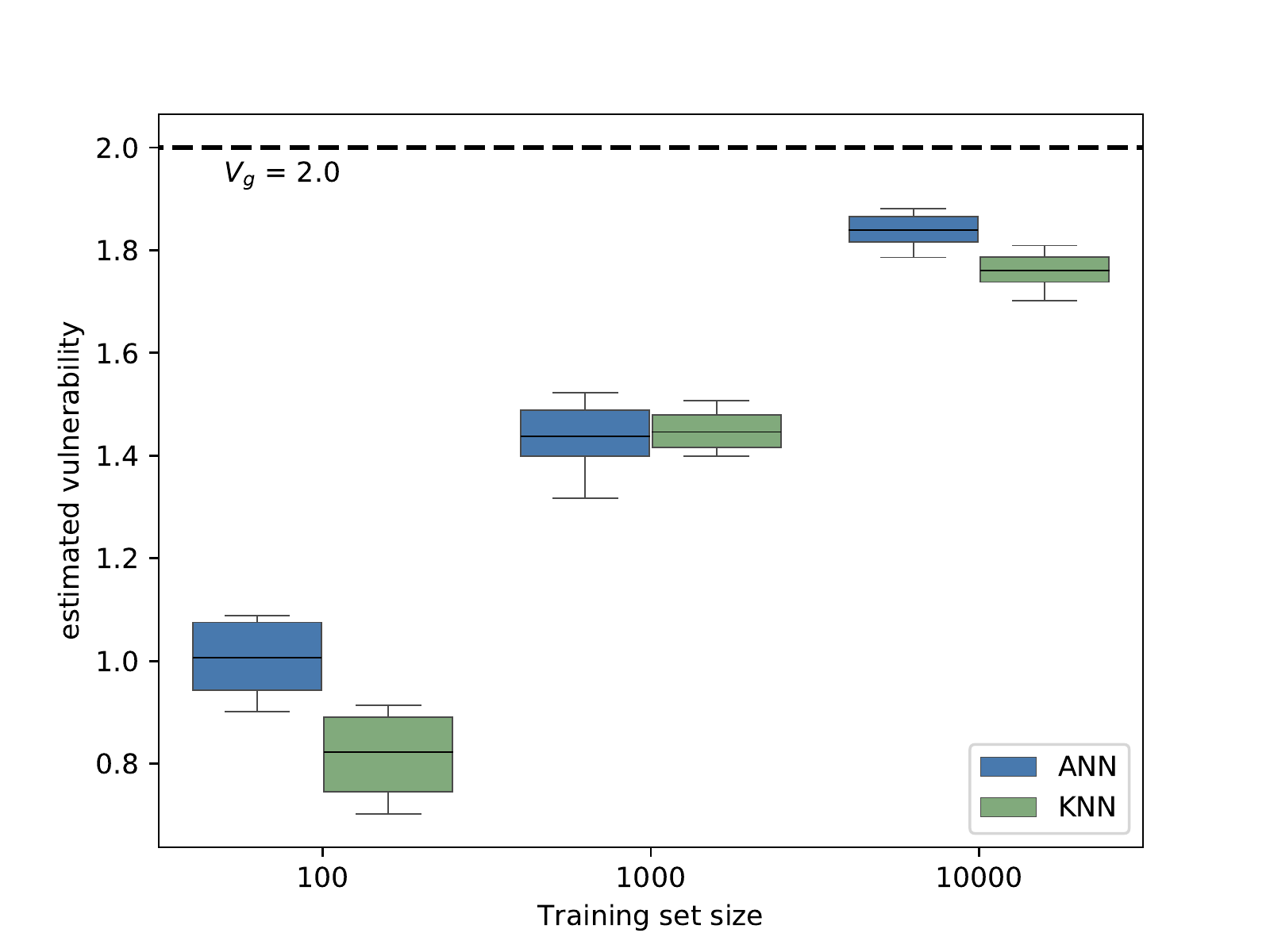}
    \caption{Vulnerability estimation for ANN and k-NN with channel pre-processing.\\}
  \end{subfigure}
\end{minipage}
\begin{minipage}[b]{\textwidth}
  \centering
  \begin{subfigure}[t]{.4\linewidth}
    \centering\includegraphics[width=0.8\linewidth]{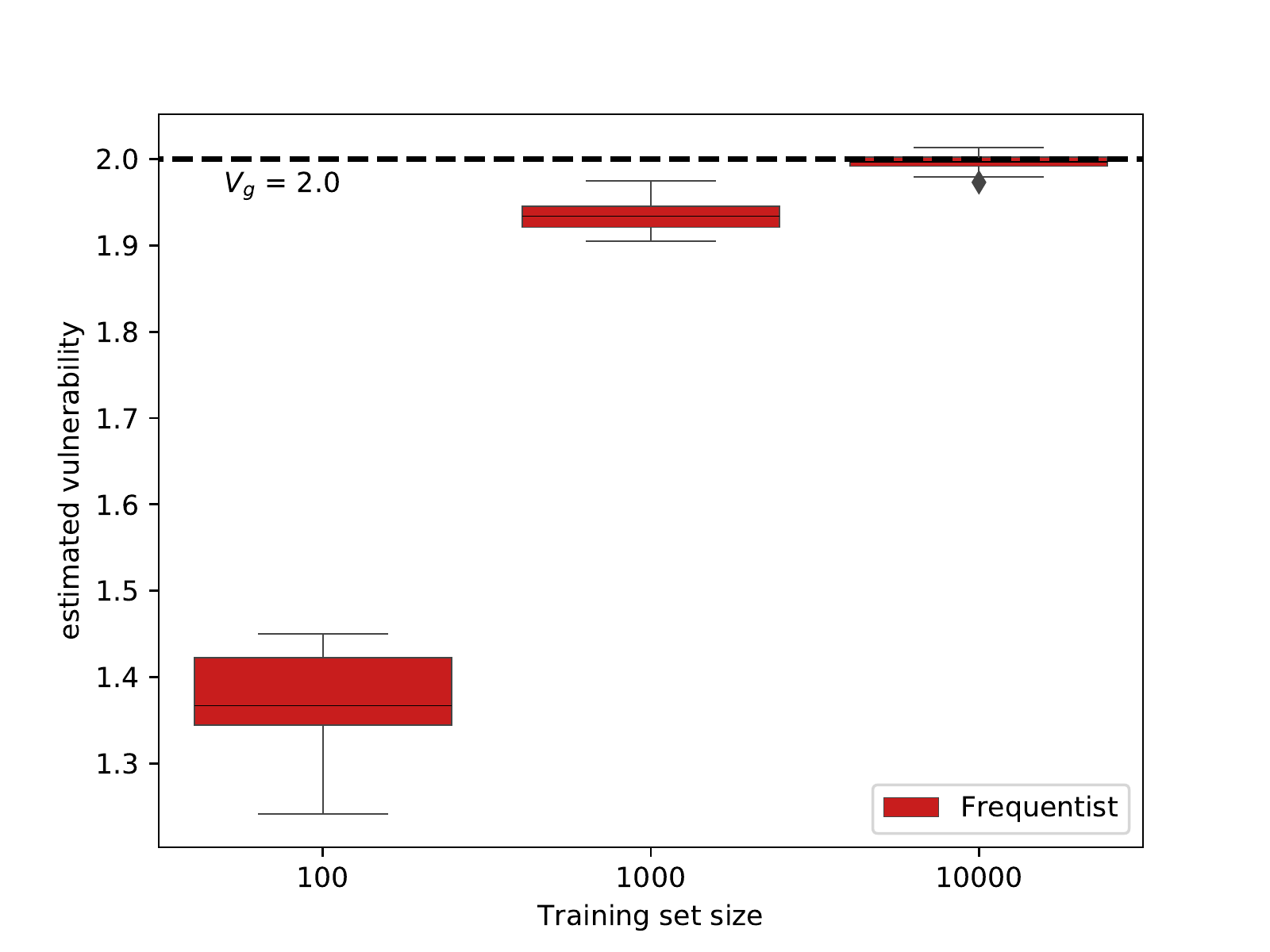}
    \caption{Vulnerability estimation for the frequentist approach.\\}
  \end{subfigure}
  \quad
  \begin{subfigure}[t]{.4\linewidth}
    \centering\includegraphics[width=0.8\linewidth]{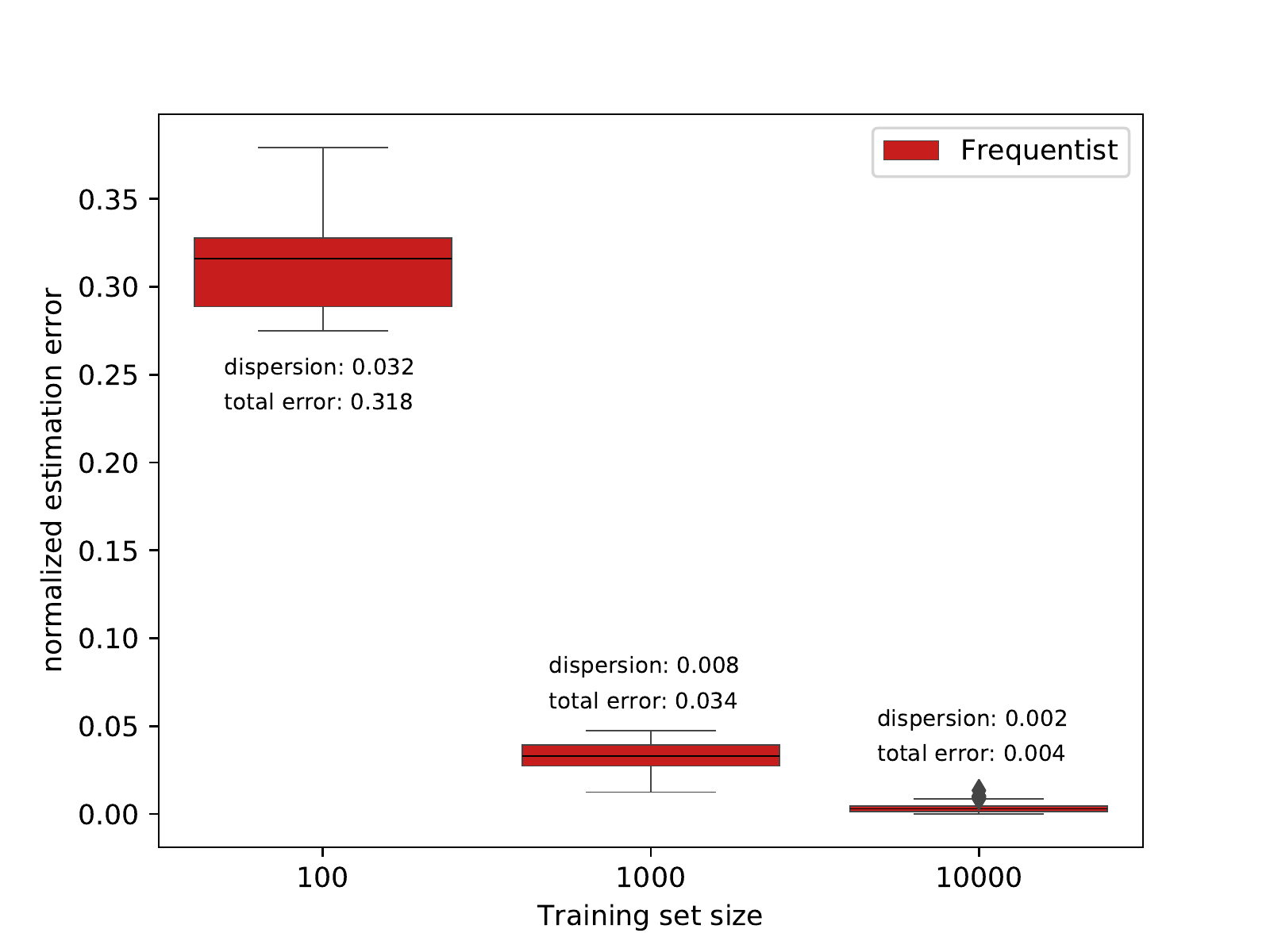}
    \caption{Normalized estimation error for the frequentist approach.\\}
  \end{subfigure}
\end{minipage}
\begin{minipage}[b]{\textwidth}
  \centering
  \begin{subfigure}[t]{.4\linewidth}
    \centering\includegraphics[width=0.8\linewidth]{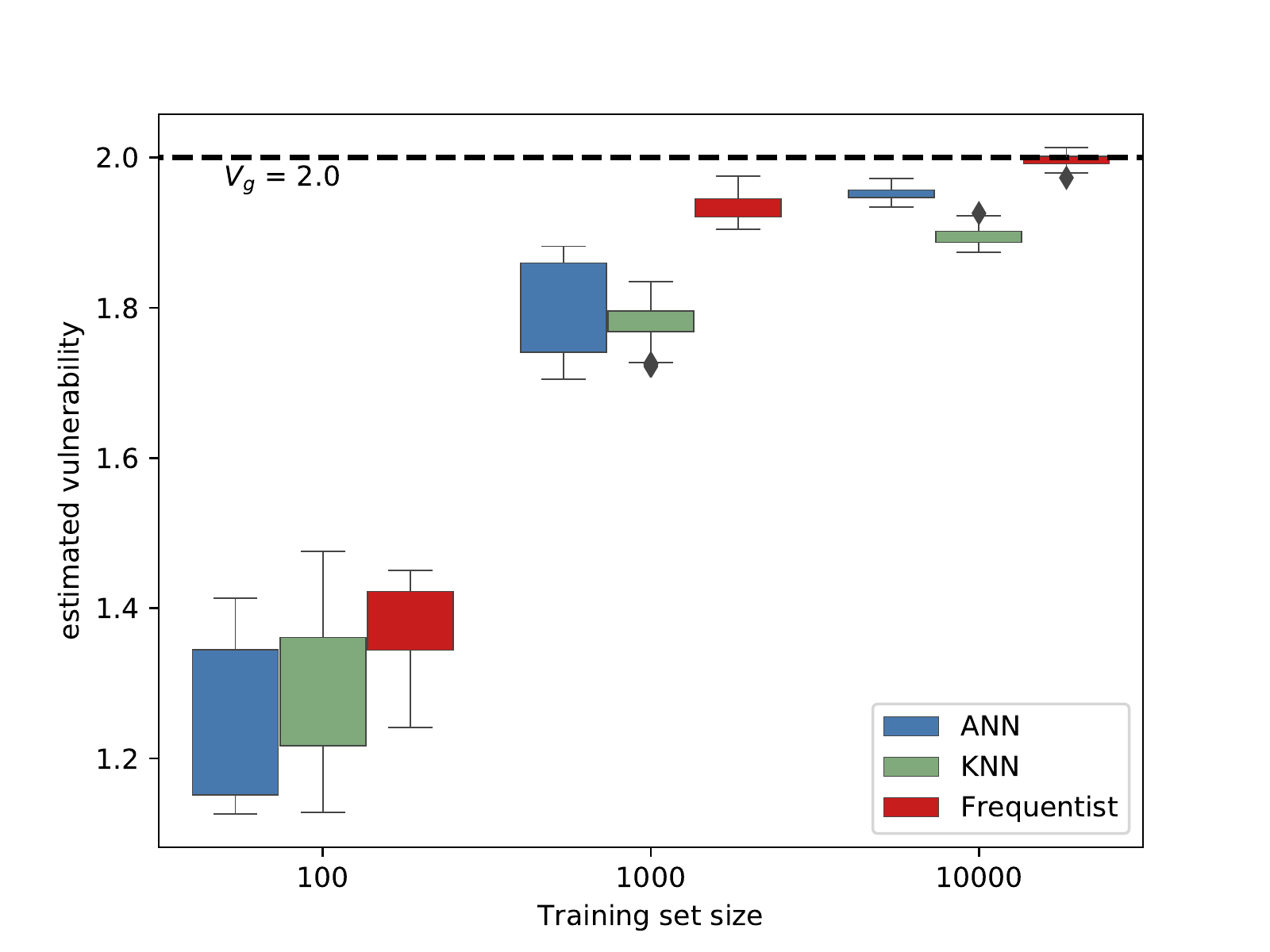}
    \caption{Vulnerability estimation for ANN and k-NN with data pre-processing, and the frequentist approach.\\}
  \end{subfigure}
  \quad
  \begin{subfigure}[t]{.4\linewidth}
    \centering\includegraphics[width=0.8\linewidth]{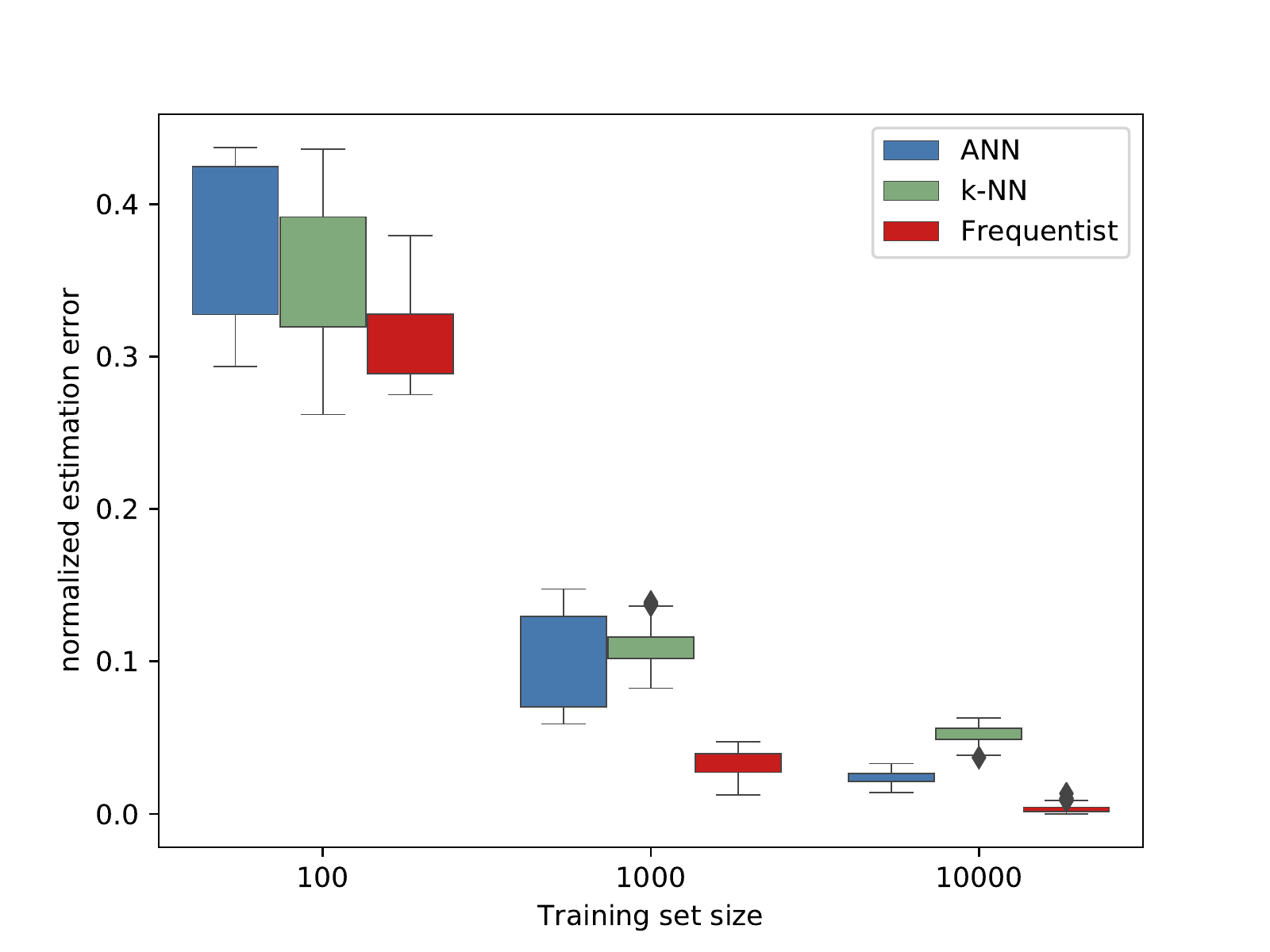}
    \caption{Normalized estimation error for ANN and k-NN with data pre-processing, and the frequentist approach.\\}
  \end{subfigure}
\end{minipage}
\begin{minipage}[b]{\textwidth}
  \centering
  \begin{subfigure}[t]{.4\linewidth}
    \centering\includegraphics[width=0.8\linewidth]{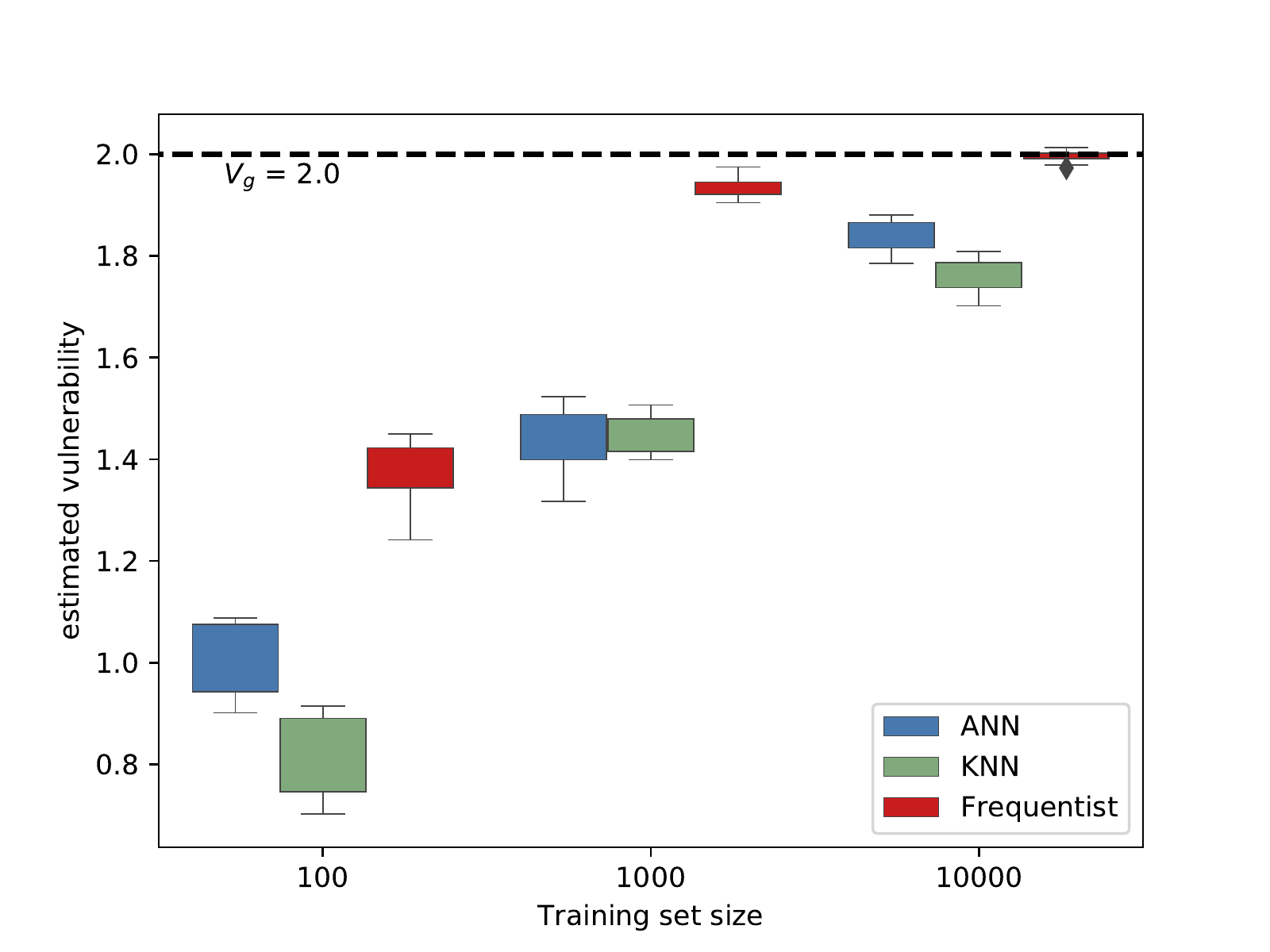}
    \caption{Vulnerability estimation for ANN and k-NN with channel pre-processing, and the frequentist approach.\\}
  \end{subfigure}
   \quad
  \begin{subfigure}[t]{.4\linewidth}
    \centering\includegraphics[width=0.8\linewidth]{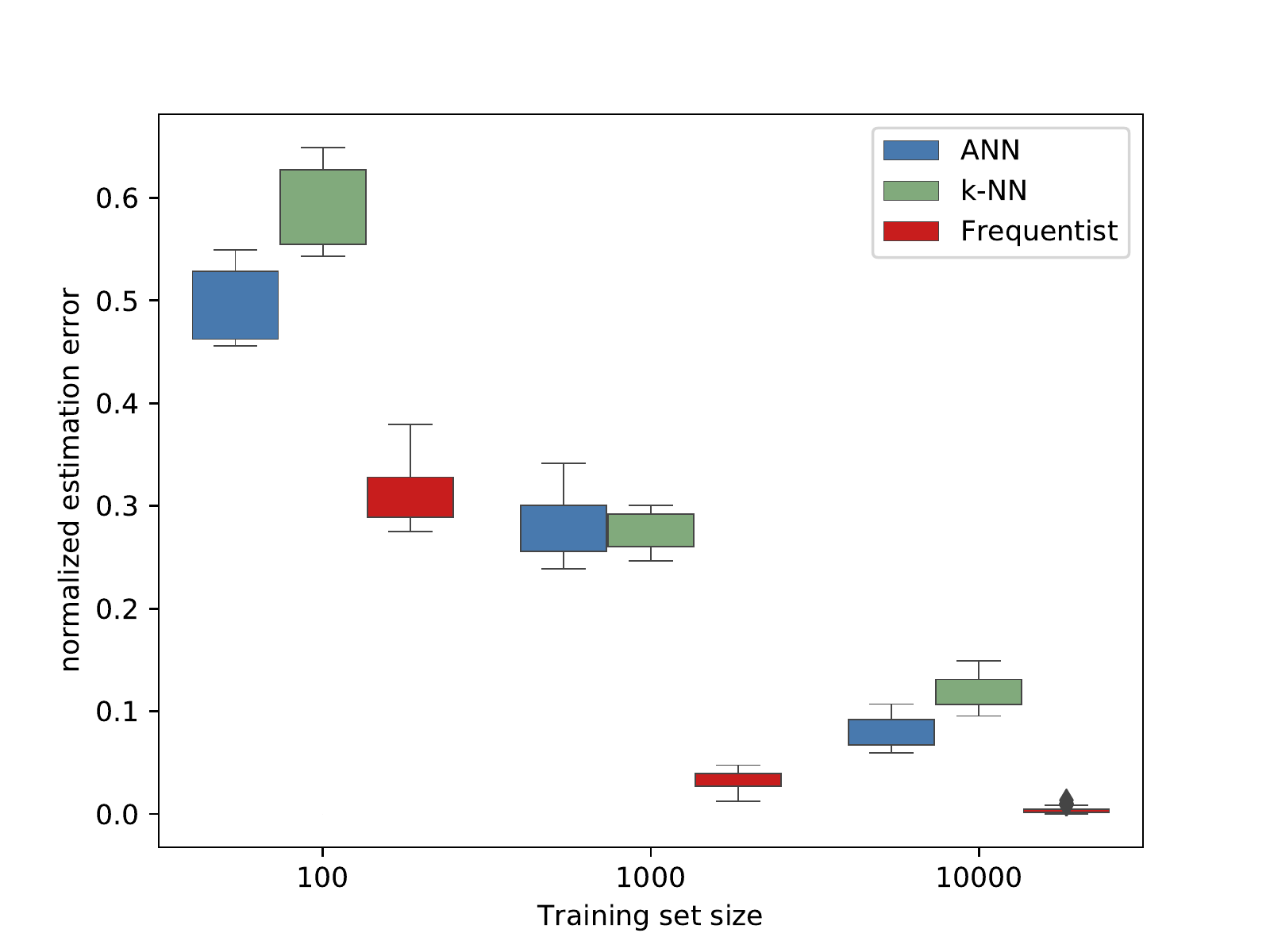}
    \caption{Normalized estimation error for ANN and k-NN with channel pre-processing, and the frequentist approach.\\}
  \end{subfigure}
\end{minipage}
\caption{Supplementary plots for the location-privacy experiment. }\label{suppl_loc_pr}
\end{figure*}

\begin{figure*}[!htbp]
\centering
\begin{minipage}[b]{\textwidth}
  \centering
  \begin{subfigure}[t]{.4\linewidth}
    \centering\includegraphics[width=0.8\linewidth]{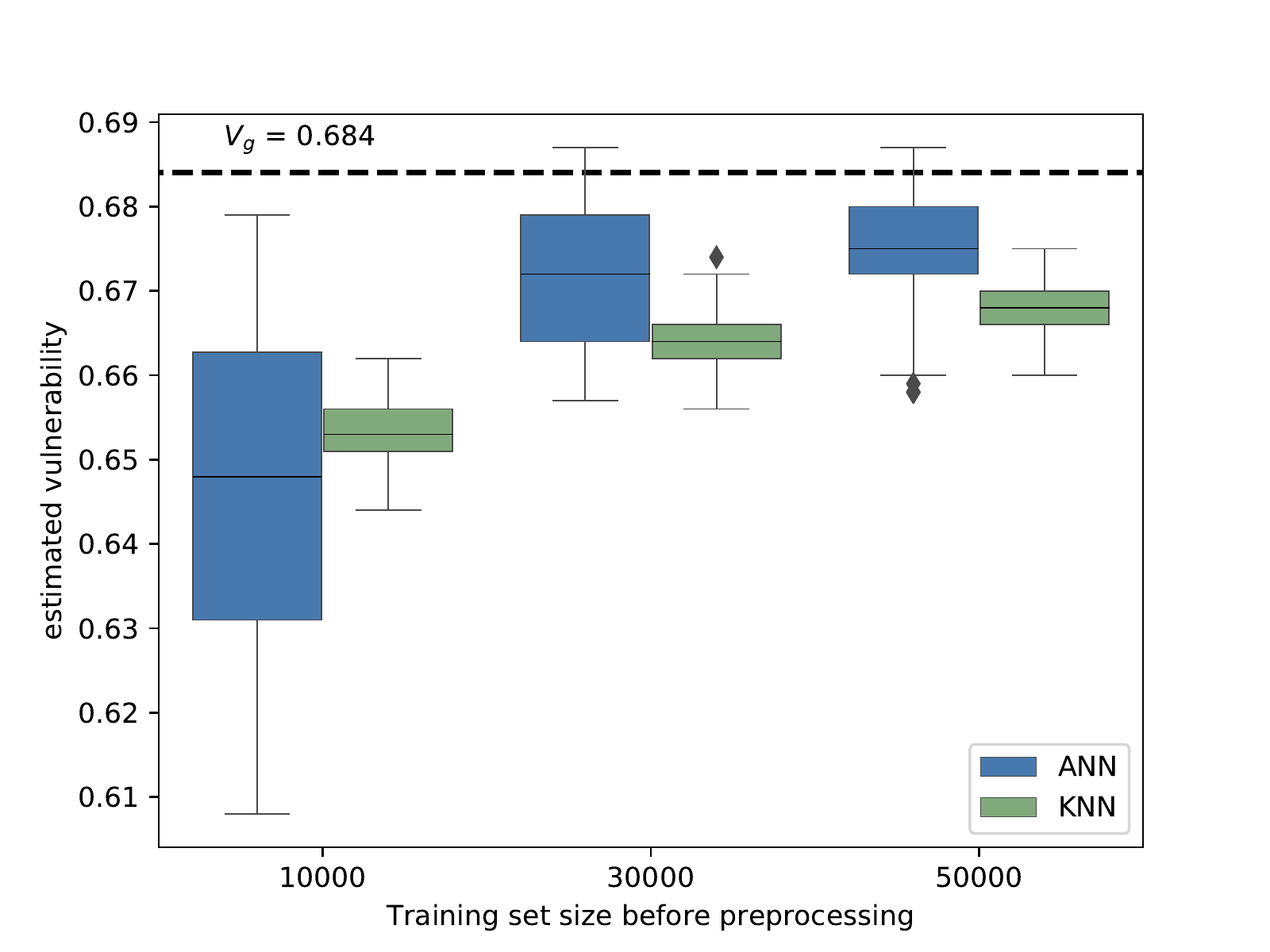}
     \caption{Vulnerability estimation for ANN and k-NN with data pre-processing.\\}
       \end{subfigure}
  \quad
  \begin{subfigure}[t]{.4\linewidth}
    \centering\includegraphics[width=0.8\linewidth]{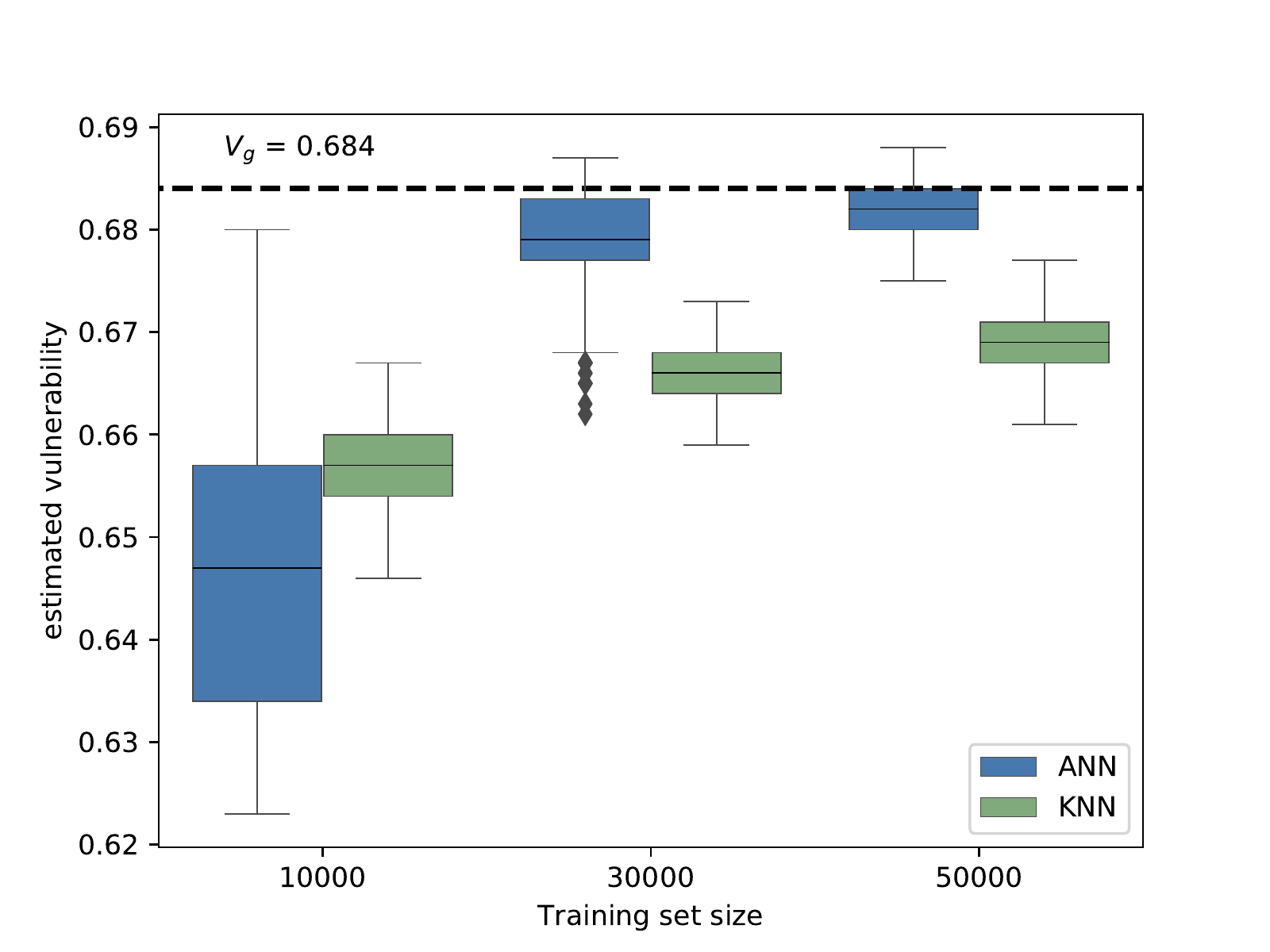}
    \caption{Vulnerability estimation for ANN and k-NN with channel pre-processing.\\}
  \end{subfigure}
\end{minipage}
\begin{minipage}[b]{\textwidth}
  \centering
  \begin{subfigure}[t]{.4\linewidth}
    \centering\includegraphics[width=0.8\linewidth]{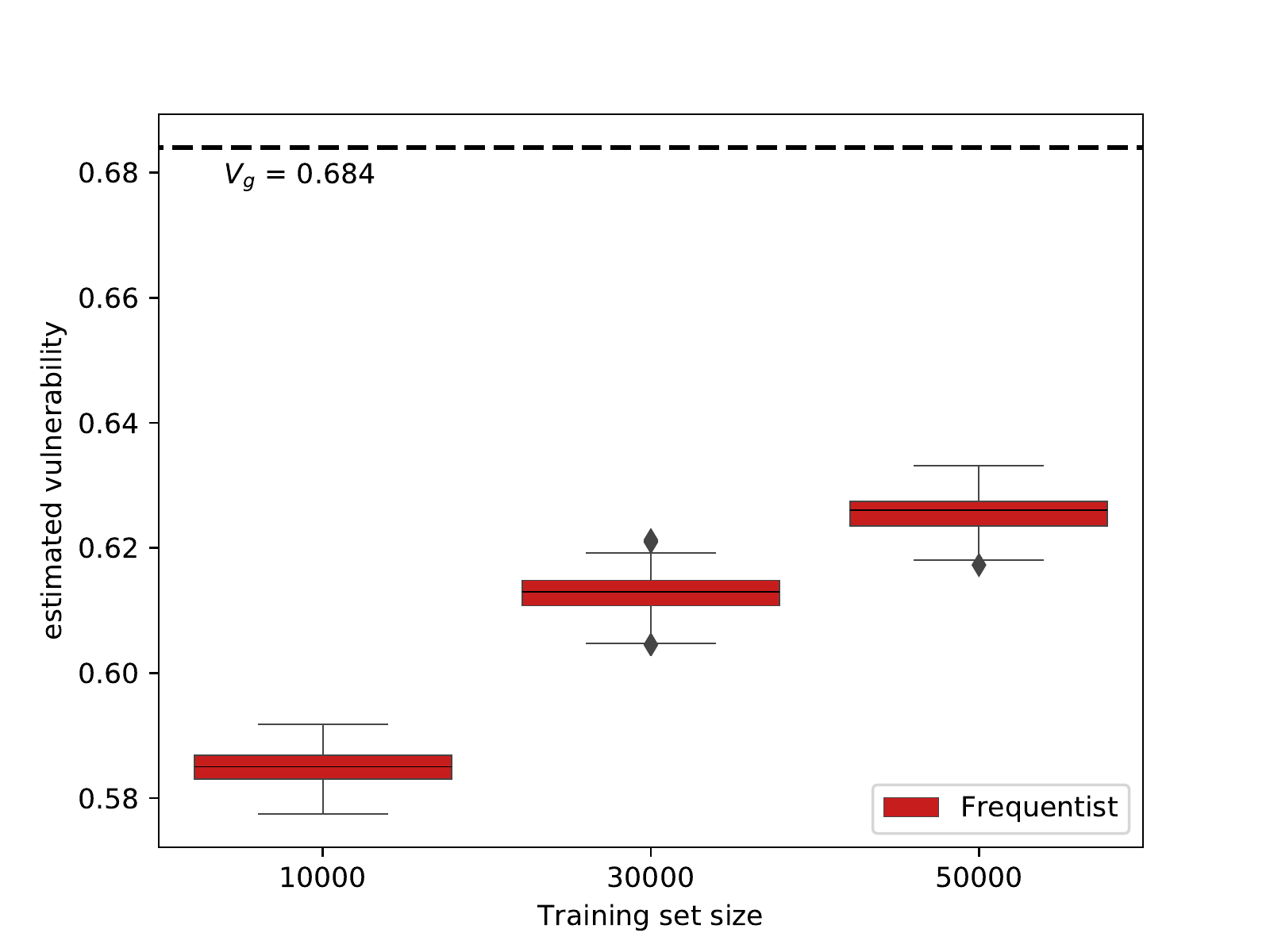}
    \caption{Vulnerability estimation for the frequentist approach.\\}
  \end{subfigure}
  \quad
  \begin{subfigure}[t]{.4\linewidth}
    \centering\includegraphics[width=0.8\linewidth]{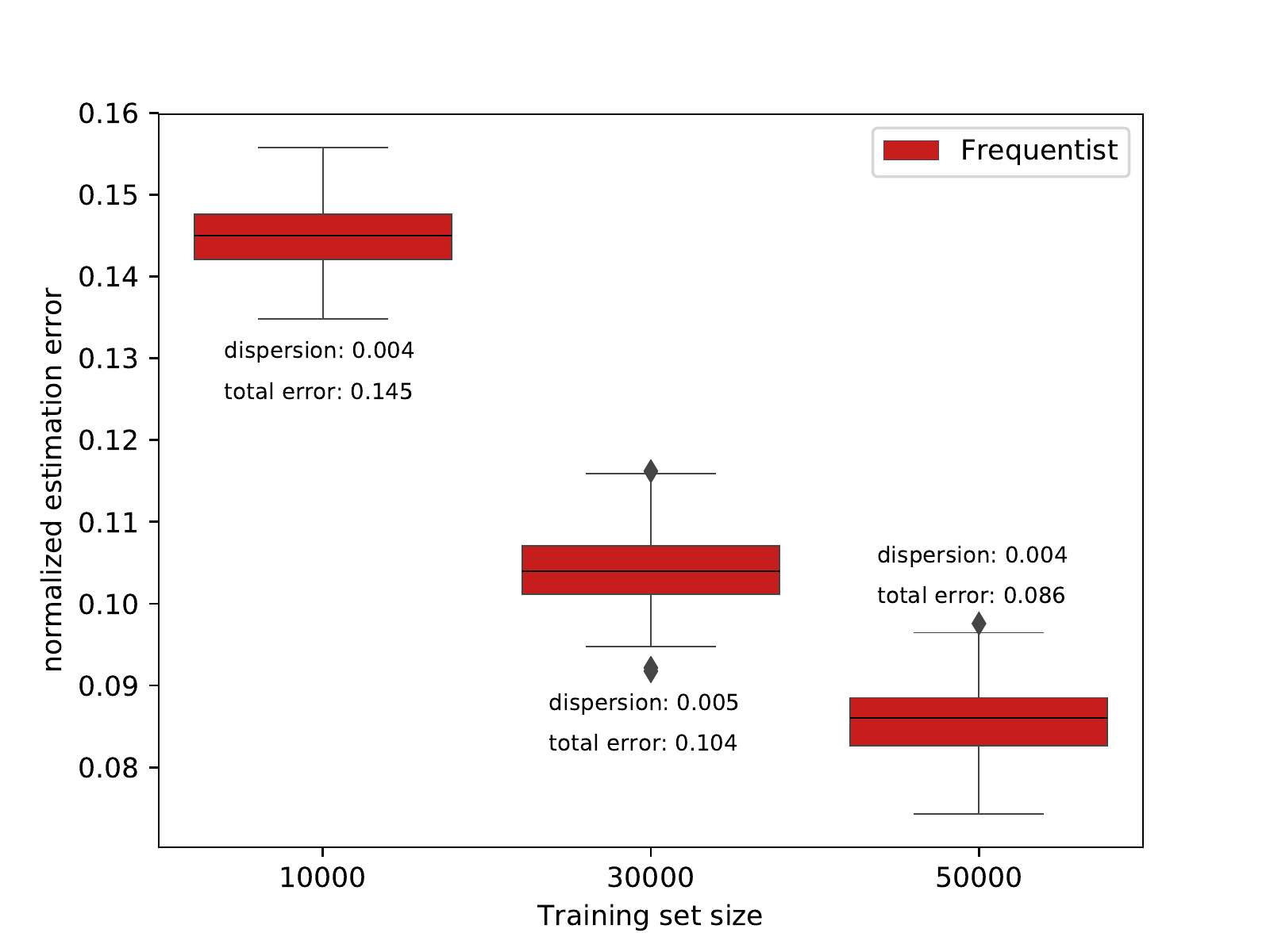}
    \caption{Normalized estimation error for the frequentist approach.\\}
  \end{subfigure}
\end{minipage}
\begin{minipage}[b]{\textwidth}
  \centering
  \begin{subfigure}[t]{.4\linewidth}
    \centering\includegraphics[width=0.8\linewidth]{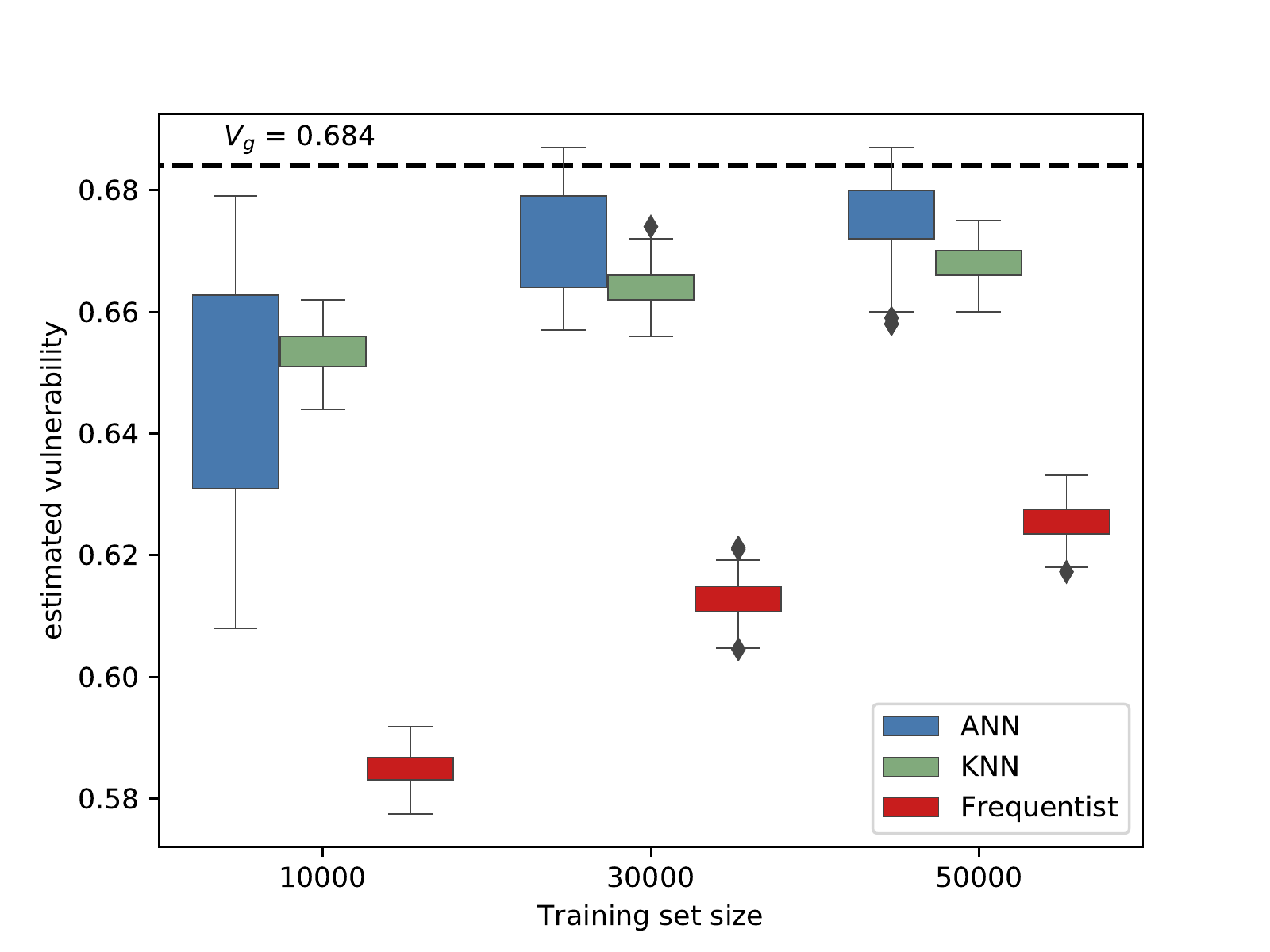}
    \caption{Vulnerability estimation for ANN and k-NN with data pre-processing, and the frequentist approach.\\}
  \end{subfigure}
  \quad
  \begin{subfigure}[t]{.4\linewidth}
    \centering\includegraphics[width=0.8\linewidth]{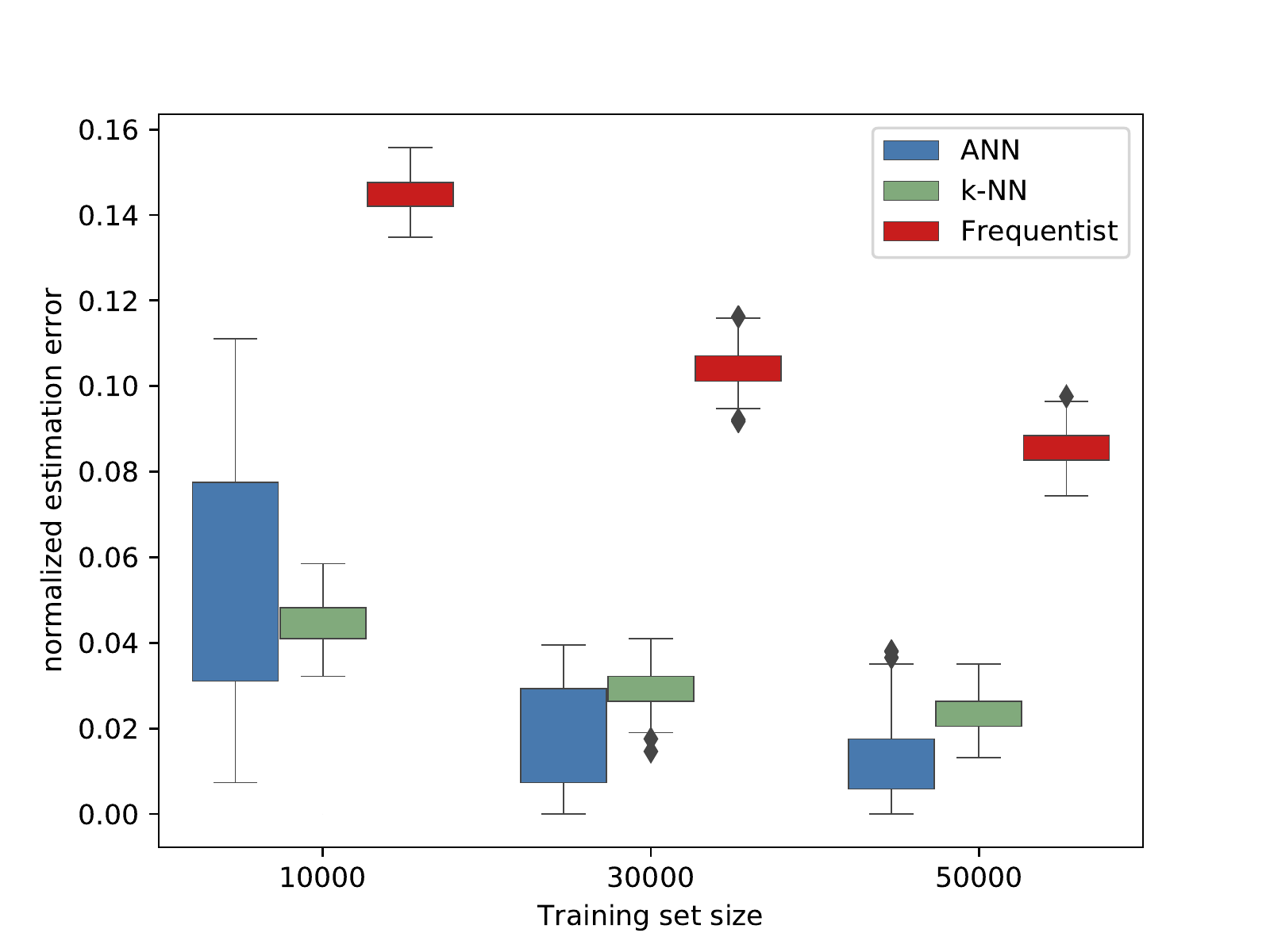}
    \caption{Normalized estimation error for ANN and k-NN with data pre-processing, and the frequentist approach.\\}
  \end{subfigure}
\end{minipage}
\begin{minipage}[b]{\textwidth}
  \centering
  \begin{subfigure}[t]{.4\linewidth}
    \centering\includegraphics[width=0.8\linewidth]{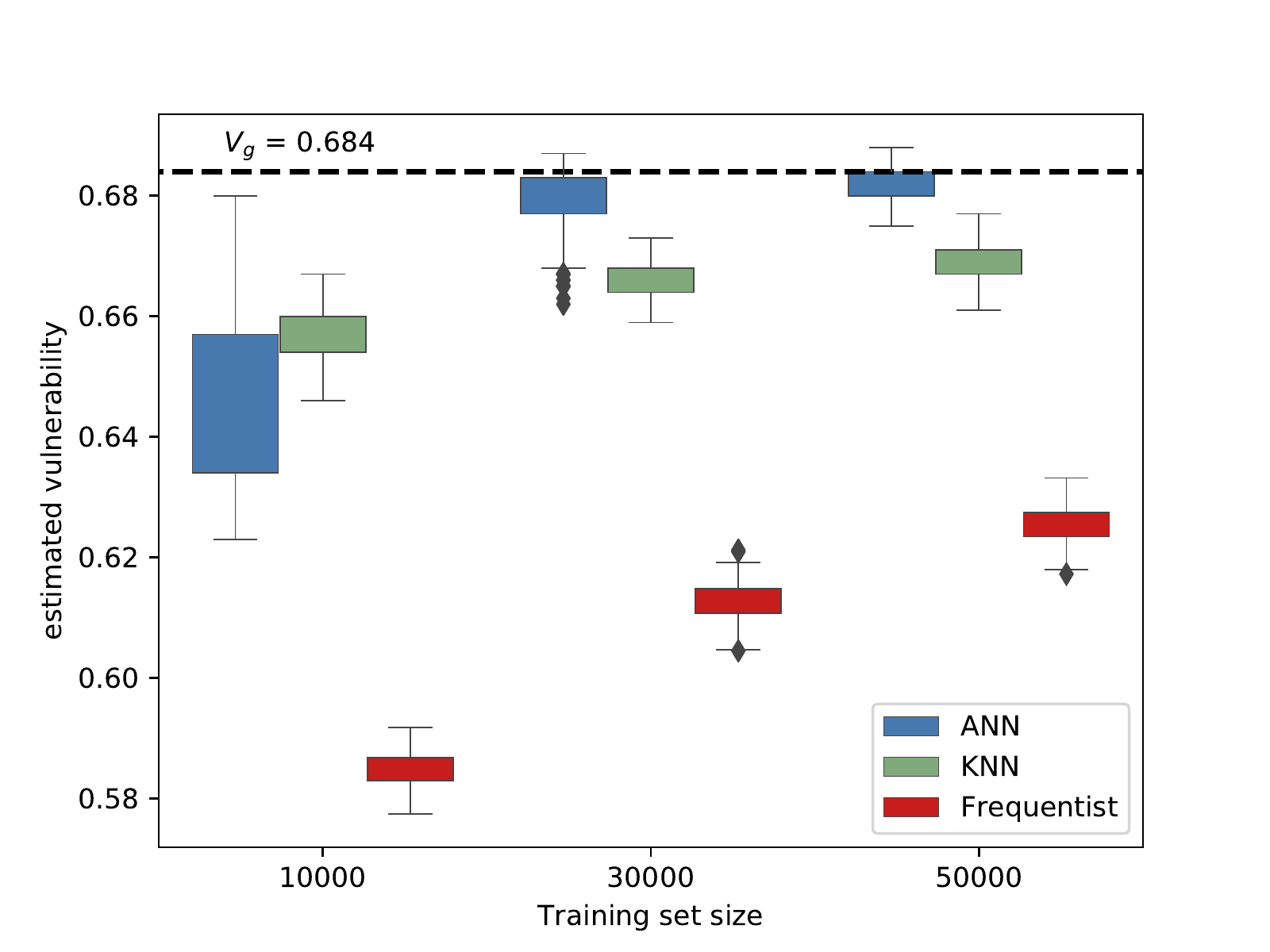}
    \caption{Vulnerability estimation for ANN and k-NN with channel pre-processing, and the frequentist approach.\\}
  \end{subfigure}
  \quad
  \begin{subfigure}[t]{.4\linewidth}
    \centering\includegraphics[width=0.8\linewidth]{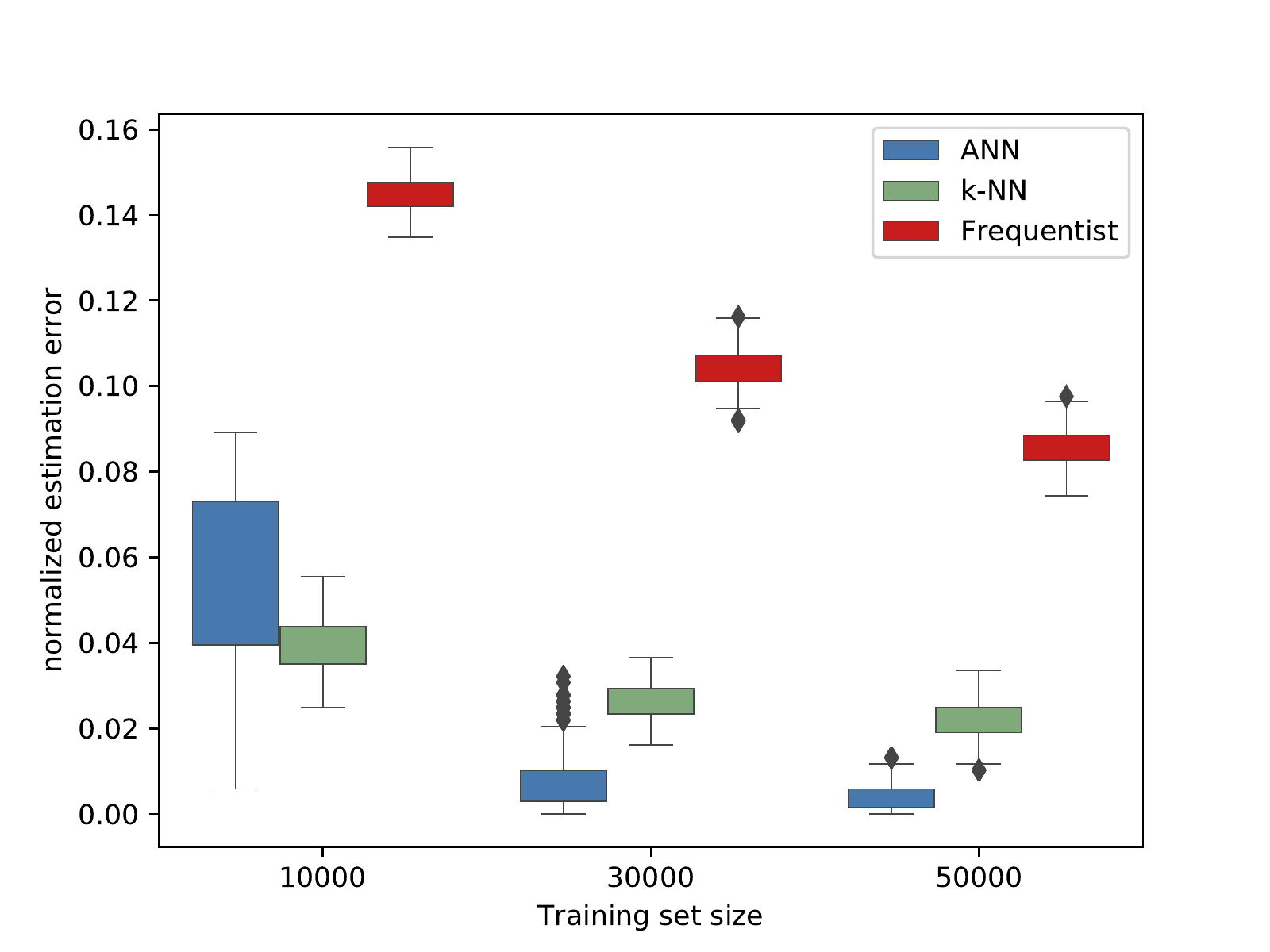}
    \caption{Normalized estimation error for ANN and k-NN with channel pre-processing, and the frequentist approach.\\}
  \end{subfigure}
\end{minipage}
\caption{Supplementary plots for the differential-privacy  experiment.}\label{suppl_dp}
\end{figure*}

\autoref{suppl_mult_guess} is related to the multiple guesses scenario,~\autoref{suppl_loc_pr}  is related to the location privacy one,~\autoref{suppl_dp} is related to the differential privacy experiment, and~\autoref{suppl_side_channel} to the password checker one.   

\end{document}